\documentclass{article}

\usepackage{amsfonts}
\usepackage{amssymb}
\usepackage{amsmath}
\usepackage{amsthm}
\usepackage{algorithm}
\usepackage{algorithmic}

\newtheorem{definition}{Definition}
\newtheorem{remark}{Remark}

\newtheorem{proposition}{Proposition}
\newtheorem{corollary}{Corollary}
\newtheorem{theorem}{Theorem}  
\newtheorem{lemma}{Lemma}

\DeclareMathOperator*{\argmax}{argmax}

\begin{document}

\begin{center}

{\Large Maximal Cost-Bounded Reachability Probability} \\
{\Large on Continuous-Time Markov Decision Processes} \\
{~}\\
{\large Hongfei Fu}\\
{~}\\
{Lehrstuhl f{\"u}r Informatik II, RWTH Aachen University, Germany}\\

\end{center}

\begin{abstract}
In this paper, we consider multi-dimensional maximal cost-bounded reachability probability over continuous-time Markov decision processes 
(CTMDPs). Our major contributions are as follows. Firstly, we derive an integral characterization which states that the maximal 
cost-bounded reachability probability function is the least fixed-point of a system of integral equations. Secondly, we prove that the 
maximal cost-bounded reachability probability can be attained by a measurable deterministic cost-positional scheduler. Thirdly, we provide 
a numerical approximation algorithm for maximal cost-bounded reachability probability. We present these results under the setting of both 
early and late schedulers. Besides, we correct a fundamental proof error in the PhD Thesis by Martin Neuh\"{a}u\ss{}er on maximal 
time-bounded reachability probability by completely new proofs for the more general case of multi-dimensional maximal cost-bounded 
reachability probability.
\end{abstract}

\section{Introduction}

The class of continuous-time Markov decision processes (CTMDPs) (or controlled Markov chains)~\cite{Puterman:1994:MDP:528623, controlledCTMG} is a stochastic model that incorporates
both features from continuous-time Markov chains (CTMCs)~\cite{feller.william.aiptia} and discrete-time Markov decision processes 
(MDPs)~\cite{Puterman:1994:MDP:528623}. A CTMDP extends a CTMC in the 
sense that it allows non-deterministic choices, and it extends an MDP in the sense that it incorporates negative exponential time-delays. 
Due to its modelling capability of real-time probabilistic behaviour and non-determinism, CTMDPs are widely used in 
dependability analysis and performance evaluation~\cite{DBLP:journals/cacm/BaierHHK10}.

In a CTMDP, non-determinism is resolved by \emph{schedulers}~\cite{DBLP:conf/formats/WolovickJ06}. Informally, a scheduler determines the 
non-deterministic choices depending on the finite trajectory of the CTMDP so far and possibly the sojourn time of the current state. A 
scheduler is assumed to be \emph{measurable} so that it induces a well-defined probability space over the infinite trajectories of the 
underlying CTMDP. Measurable schedulers are further divided into categories of \emph{early schedulers} and 
\emph{late schedulers}~\cite{DBLP:conf/fossacs/NeuhausserSK09,DBLP:conf/formats/WolovickJ06}. A scheduler that makes the choice solely by the trajectory 
so far is called an early scheduler, while a scheduler utilizes both the trajectory and the sojourn time (at the current state) is 
called a late scheduler. With schedulers, one can reason about quantitative information such as the maximal/minimal 
probability/expectation of certain property.

In this paper, we focus on the problem to compute \emph{max/min resource-bounded reachability probability} on a CTMDP. Typical 
resource types considered here are time and cost, where a time bound can be deemed as a special cost bound with unit-cost 1. 
In general, the task is to compute or approximate the optimal (max/min) reachability probability to certain target states within a given 
resource bound (e.g., a time bound). 

Optimal time-bounded reachability probability over CTMDPs has been widely studied in recent years. 
Neuh\"{a}u\ss{}er~\emph{et al.}~\cite{DBLP:conf/qest/NeuhausserZ10} proved that the maximal time-bounded reachability probability function is the least 
fixed point of a system of integral equations. Rabe and Schewe~\cite{DBLP:journals/acta/RabeS11} showed that the max/min time-bounded 
reachability probability can be attained by a deterministic piecewise-constant time-positional scheduler. Efficient approximation algorithms are also 
developed by, e.g., Neuh\"{a}u\ss{}er \emph{et al.}~\cite{DBLP:conf/qest/NeuhausserZ10}, Br{\'a}zdil 
\emph{et al.}~\cite{DBLP:journals/iandc/BrazdilFKKK13}, Hatefi \emph{et al.}~\cite{DBLP:conf/fsen/HatefiH13} and 
Rabe~\emph{et al.}~\cite{DBLP:conf/fsttcs/FearnleyRSZ11}. 

As to optimal cost-bounded reachability probability, much less is known. To the best of the author's knowledge, the only prominent result is 
by Baier~\emph{et al.}~\cite{DBLP:conf/birthday/BaierHHK08}, which establishes a certain duality property between time and cost bound. 
Their result is restrictive in the sense that (i) it assumes that the CTMDP have everywhere positive unit-cost values, (ii) it only 
takes into account one-dimensional cost-bound aside the time-bound, and (iii) it does not really provide an approximation algorithm
when both time- and cost-bounds are present. 

Besides resource-bounded reachability probability, we would like to mention another research field on CTMDPs with costs (or dually, 
rewards), which is (discounted) accumulated reward over finite/infinite horizon (cf.~\cite{DBLP:journals/cor/BuchholzS11,controlledCTMG}, just to mention a little). 

\textbf{Our Contribution.} We consider multi-dimensional maximal cost-bounded reachability probability (\emph{abbr.} MMCRP) over CTMDPs 
under the setting of both early and late schedulers, for which the unit-cost is constant. We first prove that the 
MMCRP function is the least fixed-point of a system of integral equations. Then we prove that deterministic 
cost-positional measurable schedulers suffice to achieve the MMCRP value. Finally, we describe a numerical algorithm which 
approximates the MMCRP value with an error bound. The approximation algorithm relies on a differential characterization 
which in turn is derived from the least fixed-point characterization. 
The complexity of the approximation algorithm is polynomial in the size of the CTMDP and the reciprocal of the error bound, and exponential 
in the dimension of cost vectors. 

Besides, we point out a fundamental proof error in the treatment of maximal time-bounded reachability probability on continuous-time Markov 
decision processes~\cite{DBLP:phd/de/Neuhausser2010,DBLP:conf/qest/NeuhausserZ10}. 
We fix this error in the more general setting of maximal cost-bounded reachability probability by completely new proofs. 

\textbf{Structure of The Paper.} Section 2 introduces some preliminaries of CTMDPs. Section 3 illustrates the definition of schedulers and 
the probability spaces they induce. Section 4 establishes a general fixed-point theorem for the probability measures induced by 
schedulers. In Section 5, we define the notion of maximal cost-bounded reachability probability and derive the least-fixed-point 
characterization, while we also point out the proof error in~\cite{DBLP:phd/de/Neuhausser2010,DBLP:conf/qest/NeuhausserZ10}. 
In Section 6, we prove that the maximal cost-bounded reachability probability can be reached by a measurable 
deterministic cost-positional scheduler. In Section 7, we derive a differential characterization which is crucial to our approximation 
algorithm. In Section 8, we present our approximation algorithm. Finally, Section 9 concludes the paper. 

\section{Continuous-Time Markov Decision Processes}

In the whole paper, we will use the following convention for notations. We will denote by $\mathbb{R}_{\ge 0}$ the set of non-negative 
real numbers and by $\mathbb{N}_0$ the set of non-negative integers. We use $x,c,d,t,\tau$ to range over real numbers, $l,m,n,i,j$ to 
range over $\mathbb{N}_0$, and bold-face letters $\mathbf{x},\mathbf{c},\mathbf{d}$ to range over (column) real vectors. Given 
$\mathbf{c}\in\mathbb{R}^k$ ($k\in\mathbb{N}$), we denote by $\mathbf{c}_i$ ($1\le i\le k$) the $i$-th coordinate of $\mathbf{c}$. We denote by $\vec{x}$ the 
real vector whose coordinates are all equal to $x\in\mathbb{R}$ (with the implicitly known dimension). We extend $\{\le,<,\ge,>\}$ to real 
vectors and functions in a pointwise fashion: for two real vectors $\mathbf{c},\mathbf{d}$, 
$\mathbf{c}\le\mathbf{d}$ iff $\mathbf{c}_i\le\mathbf{d}_i$ for all $i$; 
for two real-valued functions $g,h$, $g\le h$ iff $g(y)\le h(y)$ for all $y$.  
Given a set $Y$, we let 
$\mathbf{1}_Y$ be the indicator function of $Y$, i.e,  $\mathbf{1}_Y(y)=1$ if $y\in Y$ and $\mathbf{1}_Y(y)=0$ for $y\in X-Y$, where 
$X\supseteq Y$ is an implicitly known set. Given a positive real number $\lambda>0$, let 
${f}_\lambda(t):=\lambda\cdot e^{-\lambda\cdot t}~(t\ge 0)$ be the probability density function of the negative exponential distribution 
with rate $\lambda$. Besides, we will use $g,h$ to range over general functions.

\subsection{The Model}

\begin{definition}
A \emph{Continuous-Time Markov Decision Process} (CTMDP) is a tuple $\left(L,Act,\mathbf{R},\{\mathbf{w}_i\}_{1\le i\le k}\right)$ where 
\begin{itemize} \itemsep1pt \parskip0pt \parsep0pt
\item $L$ is a finite set of \emph{states} (or \emph{locations});
\item $Act$ is a finite set of \emph{actions};
\item $\mathbf{R}:L\times Act\times L\rightarrow\mathbb{R}_{\ge 0}$ is the \emph{rate matrix};
\item $\{\mathbf{w}_i:L\times Act\rightarrow\mathbb{R}_{\ge 0}\}_{1\le i\le k}$ is the family of $k$ \emph{unit-cost functions} ($k\in\mathbb{N}$); 
\end{itemize}
An action $a\in Act$ is \emph{enabled} at state $s\in L$ if $\mathbf{E}(s,a):=\sum_{u\in L}\mathbf{R}(s,a,u)$ is non-zero. The set of enabled 
actions at $s\in L$ is denoted by $\mathrm{En}(s)$. We assume that for each state $s\in L$, $\mathrm{En}(s)\ne\emptyset$.
\end{definition}
Let $\left(L,Act,\mathbf{R},\{\mathbf{w}_i\}_{1\le i\le k}\right)$ be a CTMDP. For each $s,s'\in L$ and $a\in\mathrm{En}(s)$, we define 
\[
\mathbf{P}(s,a,s'):=\frac{\mathbf{R}(s,a,s')}{\mathbf{E}(s,a)}
\]
to be the discrete transition probability from $s$ to $s'$ via $a$. We denote by $\mathbf{w}(s,a)$ the real vector 
$\{\mathbf{w}_i(s,a)\}_{1\le i\le k}$ for each $(s,a)\in L\times Act$\enskip. Given $s\in L$ and $a\in Act$, we denote by $\mathcal{D}[s]$ 
the Dirac distribution (over $L$) at $s$ (i.e., $\mathcal{D}[s](s)=1$ and $\mathcal{D}[s](s')=0$ for $s'\in L-\{s\}$) and by 
$\mathcal{D}[a]$ the Dirac distribution (over $Act$) at $a$. Moreover, we define (with $\min\emptyset:=1$):
\begin{itemize}\itemsep1pt \parskip0pt \parsep0pt
\item $\mathbf{w}_{\min}:=\min\{\mathbf{w}_i(s,a)\mid 1\le i\le k, s\in L, a\in\mathrm{En}(s),\mathbf{w}_i(s,a)>0\}$~;
\item $\mathbf{w}_{\max}:=\max\{\mathbf{w}_i(s,a)\mid 1\le i\le k, s\in L, a\in\mathrm{En}(s)\}$~;
\item $\mathbf{E}_{\max}:=\max\{\mathbf{E}(s,a)\mid s\in L, a\in\mathrm{En}(s)\}$~;
\item $\mathbf{P}_{\min}:=\min\{\mathbf{P}(s,a,s')\mid s,s'\in L, a\in\mathrm{En}(s),\mathbf{P}(s,a,s')>0\}$\enskip.
\end{itemize}
We will use $s,s'$ (resp. $a,b$) to range over states (resp. actions) of a CTMDP. 

Often, a CTMDP is accompanied with an initial distribution which specifies the initial stochastic environment (for the CTMDP).

\begin{definition}
Let $\mathcal{M}=\left(L,Act,\mathbf{R},\{\mathbf{w}_i\}_{1\le i\le k}\right)$ be a CTMDP. An \emph{initial distribution} 
(for $\mathcal{M}$) is a function $\alpha:L\rightarrow [0,1]$ such that $\sum_{s\in L}\alpha(s)=1$\enskip.
\end{definition}

Intuitively, the execution of a CTMDP $\left(L,Act,\mathbf{R},\{\mathbf{w}_i\}_{1\le i\le k}\right)$ with a  
\emph{scheduler} is as follows. At the beginning, an initial state $s$ is chosen (as the current state) w.r.t the initial distribution 
$\alpha$. Then the scheduler chooses an action $a$ enabled at $s$. Afterwards, a time-delay occurs at the state $s$ whose 
duration observes the negative exponential distribution with rate $\mathbf{E}(s,a)$. After the time-delay, the current state is 
switched to an arbitrary state $s'\in L$ with probability $\mathbf{P}(s,a,s')$, and so forth. Besides, each cost function $\mathbf{w}_i$ 
assigns to each state-action pair $(s,a)$ the $i$-th constant unit-cost $\mathbf{w}_i(s,a)$ (per time unit) when the CTMDP dwells at state 
$s$. Basically, the scheduler makes the decision of the action to be chosen when entering a new state, and has two distinct 
objectives: either to maximize a certain property or (in contrast) to minimize a certain property. In this paper, we will focus on the objective to 
maximize a cost-bounded reachability probability for a certain target set of states. 

In this paper, we also study an important subclass of CTMDPs, called \emph{locally-uniform} CTMDPs 
(cf.~\cite{DBLP:conf/fossacs/NeuhausserSK09}).

\begin{definition}
A CTMDP $\left(L,Act,\mathbf{R},\{\mathbf{w}_i\}_{1\le i\le k}\right)$ is \emph{locally-uniform} iff $\mathbf{E}(s,a)=\mathbf{E}(s,b)$ and 
$\mathbf{w}_i(s,a)=\mathbf{w}_i(s,b)$ for all $1\le i\le k$, $s\in L$ and $a,b\in\mathrm{En}(s)$. 
\end{definition}
Intuitively, a locally uniform CTMDP has the property that the time-delay and the cost is independent of the action chosen at each state.
For locally-uniform CTMDPs, we simply use $\mathbf{E}(s)$ to denote $\mathbf{E}(s,a)$ ($a\in\mathrm{En}(s)$ is arbitrary), and 
$\mathbf{w}(s),\mathbf{w}_i(s)$ for $\mathbf{w}(s,a),\mathbf{w}_i(s,a)$ likewise.

\subsection{Paths and Histories}

In this part, we introduce the notion of paths and histories. Intuitively, paths reflect infinite executions of a CTMDP, whereas histories
reflect finite executions of a CTMDP. Below we fix a CTMDP 
$\mathcal{M}=\left(L,Act,\mathbf{R},\{\mathbf{w}_i\}_{1\le i\le k}\right)$\enskip. 

\begin{definition}
A(n infinite) \emph{path} $\pi$ is an infinite sequence 
\[
\pi=\left\langle s_0\xrightarrow{a_0,t_0}s_1\xrightarrow{a_1,t_1}s_2\dots\right\rangle
\]
such that $s_i\in L$, $t_i\in\mathbb{R}_{\ge 0}$ and $a_i\in Act$ for all $i\ge 0$; We denote $s_i,t_i$ and $a_i$ by 
$\pi[i],\pi\langle i\rangle$ and $\pi(i)$, respectively. A (finite) \emph{history} $\xi$ is a finite sequence 
\[
\xi=\left\langle s_0\xrightarrow{a_0,t_0}s_1\xrightarrow{a_1,t_1}s_2\dots s_m\right\rangle\quad(m\ge 0) 
\]
such that  $s_i\in L$, $t_i\in\mathbb{R}_{\ge 0}$ and $a_i\in Act$ for all $0\le i\le m-1$, and $s_m\in L$; We denote $s_i,t_i,a_i$ and $m$ by 
$\xi[i],\xi\langle i\rangle,\xi(i)$ and $|\xi|$, respectively. Moreover, we define $\xi\downarrow:=\xi\left[|\xi|\right]$ to be the last 
state of the history $\xi$\enskip.
\end{definition}
Below we introduce more notations on paths and histories. We denote the set of paths and histories (of $\mathcal{M}$) by 
$Paths(\mathcal{M})$ and $Hists(\mathcal{M})$, respectively. We define 
$Hists^n(\mathcal{M}):=\left\{\xi\in Hists(\mathcal{M})\mid |\xi|=n\right\}$ to be the set of all histories with length $n$ ($n\ge 0$). 
For each $n\in\mathbb{N}_0$ and $\pi\in Paths(\mathcal{M})$, we define the history $\pi[0..n]$ to be the finite prefix of $\pi$ up to $n$; Formally,
\[
\pi[0..n]:=\left\langle \pi[0]\xrightarrow{\pi(0),\pi\langle 0\rangle}\dots\pi[n]\right\rangle\enskip.
\]
Given $\pi\in Paths(\mathcal{M})$ and $(s,a,t)\in L\times Act \times\mathbb{R}_{\ge 0}$, we denote by $s\xrightarrow{a,t}\pi$ the path 
obtained by ``putting'' the prefix ``$s\xrightarrow{a,t}$'' before $\pi$; Formally, 
\[
s\xrightarrow{a,t}\pi:=\left\langle s\xrightarrow{a,t}\pi[0]\xrightarrow{\pi(0),\pi\langle 0\rangle}\pi[1]\xrightarrow{\pi(1),\pi\langle 1\rangle}\dots\right\rangle\enskip. 
\]
Analogously, we define $s\xrightarrow{a,t}\xi$ (for $\xi\in Hists(\mathcal{M})$) to be the history obtained by ``putting'' 
``$s\xrightarrow{a,t}$'' before the history $\xi$. 

Intuitively, a path $\pi$ reflects a whole execution (trajectory) of the CTMDP where $\pi[i]$ is the current state at the $i$-th stage, 
$\pi(i)$ is the action chosen at $\pi[i]$ and $\pi\langle i\rangle$ is the dwell-time (time-delay) on $\pi[i]$. On the other hand, a history $\xi$ is a 
finite prefix of a path which reflects the execution up to $|\xi|$ stages. 

Below we extend sets of histories to sets of paths in a cylindrical fashion.

\begin{definition}
Suppose $n\in\mathbb{N}_0$ and $\Xi\subseteq Hists^n(\mathcal{M})$. The \emph{cylinder extension} of $\Xi$, denoted $\mathrm{Cyl}(\Xi)$, 
is defined as follows: 
\[
\mathrm{Cyl}(\Xi):=\{\pi\in Paths(\mathcal{M})\mid \pi[0..n]\in\Xi\}\enskip.
\]
\end{definition}
In this paper, we concern costs on paths and histories. The cost is assigned linearly w.r.t the unit-cost and the time spent in a 
state. The following definition presents the details.

\begin{definition}
Given a path $\pi\in Paths(\mathcal{M})$ and a set $G\subseteq L$ of states, we denote by $\mathbf{C}_j(\pi,G)$ ($1\le j\le k$) the $j$-th 
accumulated cost along $\pi$ until $G$ is reached; Formally, if $\pi[m]\in G$ for some $m\ge 0$ then
\[
\mathbf{C}_j(\pi,G):=\sum_{i=0}^{n} \mathbf{w}_j(\pi[i],\pi(i))\cdot \pi\langle i\rangle
\]
where $n\in\mathbb{N}_0\cup\{-1\}$ is the smallest integer such that $\pi[n+1]\in G$; otherwise $\mathbf{C}_j(\pi,G):=+\infty$\enskip.
Given a history $\xi\in Hists(\mathcal{M})$, we denote by $\mathbf{C}_j(\xi)$ ($1\le j\le k$) the accumulated cost of $\xi$ w.r.t the 
$j$-th unit-cost function; Formally, 
\[
\mathbf{C}_j(\xi):=\sum_{i=0}^{|\xi|-1} \mathbf{w}_j(\xi[i],\xi(i))\cdot \xi\langle i\rangle\enskip.
\]
\end{definition}
We denote by $\mathbf{C}(\pi,G)$ (resp. $\mathbf{C}(\xi)$) the vector $\{\mathbf{C}_j(\pi,G)\}_{1\le j\le k}$ (resp. 
$\{\mathbf{C}_j(\xi)\}_{1\le j\le k}$)\enskip. 

\subsection{Measurable Spaces on Paths and Histories}
In the following, we define the measurable spaces for paths and histories, following the definitions 
of~\cite{DBLP:conf/formats/WolovickJ06, DBLP:conf/fossacs/NeuhausserSK09}. Below we fix a CTMDP 
$\mathcal{M}=\left(L,Act,\mathbf{R},\{\mathbf{w}_i\}_{1\le i\le k}\right)$\enskip. Firstly, we introduce the notion of combined actions 
and its measurable space.

\begin{definition}
A \emph{combined action} is a tuple $(a,t,s)$ where $a\in Act,t\in\mathbb{R}_{\ge 0}$ and $s\in L$. The \emph{measurable space} 
$(\Gamma_\mathcal{M},\mathcal{U}_\mathcal{M})$ over combined actions is defined as follows:
\begin{itemize}\itemsep1pt \parskip0pt \parsep0pt
\item $\Gamma_{\mathcal{M}}:=Act\times \mathbb{R}_{\ge 0}\times L$ is the set of combined actions;
\item $\mathcal{U}_{\mathcal{M}}:=2^{Act}\otimes \mathcal{B}(\mathbb{R}_{\ge 0})\otimes 2^{L}$ is the product $\sigma$-algebra for which 
$\mathcal{B}(\mathbb{R}_{\ge 0})$ is the Borel $\sigma$-field on $\mathbb{R}_{\ge 0}$\enskip.
\end{itemize}
\end{definition}
The following definition introduces the notion of templates which will be used to define the measurable spaces.

\begin{definition}
A \emph{template} $\theta$ is a finite sequence 
$\theta=\langle s,U_1,\dots, U_m\rangle$ ($m\ge 0$) 
such that $s\in L$ and $U_i\in\mathcal{U}_\mathcal{M}$ for $1\le i\le m$; The \emph{length} of $\theta$, denoted by $|\theta|$, is defined 
to be $m$. The set of histories $\mathrm{Hists}(\theta)$ \emph{spanned} by a template $\theta$ is defined by:
\begin{align*}
& \mathrm{Hists}\left(\langle s,U_1,\dots, U_m\rangle\right):=\Big\{\xi\in Hists^m(\mathcal{M})\mid \\
& \qquad\qquad\qquad\xi[0]= s\mbox{ and }\left(\xi(i),\xi\langle i\rangle,\xi[i+1]\right)\in U_{i+1}\mbox{ for all }0\le i< m\Big\}\enskip.
\end{align*}
\end{definition}
Now we introduce the measurable spaces on paths and histories, as in the following definition.

\begin{definition}\label{def:msp}
The \emph{measurable space} $(\Omega^n_\mathcal{M}, \mathcal{S}^n_\mathcal{M})$ over
$Hists^n(\mathcal{M})$ $(n\in\mathbb{N}_{0})$ is defined as follows: $\Omega^n_\mathcal{M}=Hists^n(\mathcal{M})$ and 
$\mathcal{S}^n_\mathcal{M}$ is generated by the family
\[
\left\{\mathrm{Hists}(\theta)\mid \theta\mbox{ is a template and }|\theta|=n\right\}
\]
of subsets of $Hists^n(\mathcal{M})$. 

The \emph{measurable space} $(\Omega_\mathcal{M}, \mathcal{S}_\mathcal{M})$ over $Paths(\mathcal{M})$ is defined as follows:  
$\Omega_\mathcal{M}=Paths(\mathcal{M})$ and $\mathcal{S}_{\mathcal{M}}$ is the smallest $\sigma$-algebra generated by the family
\begin{equation}
\left\{\mathrm{Cyl}(\Xi)\mid \Xi\in\mathcal{S}^n_\mathcal{M}\mbox{ for some }n\ge 0\right\}\tag{\S}
\end{equation}
of subsets of $Paths(\mathcal{M})$.
\end{definition}
\begin{remark}\label{rmk:msp}
An alternative way to define the measurable space on paths can be done by changing \emph{(\S)} to the following set:
\[
\mathcal{C}:=\left\{\mathrm{Cyl}(\mathrm{Hists}(\theta))\mid \theta\mbox{ is a template}\right\}\enskip.
\] 
This can be seen as follows. Let $\mathcal{S}'$ be the $\sigma$-algebra on paths generated by $\mathcal{C}$. Clearly, 
$\mathcal{S}'\subseteq\mathcal{S}_\mathcal{M}$. For each $n\in\mathbb{N}_0$, define 
$\mathcal{S}'_n:=\{\Xi\subseteq\Omega^n_\mathcal{M}\mid\mathrm{Cyl}(\Xi)\in\mathcal{S}'\}$\enskip.
One can verify that $\mathcal{S}'_n$ is a $\sigma$-algebra on $\Omega^n_\mathcal{M}$ by the following facts: 
\begin{enumerate}\itemsep1pt \parskip0pt \parsep0pt
\item $\Omega^n_\mathcal{M}\in\mathcal{S}'_n$; 
\item If $\Xi\in\mathcal{S}'_n$ then $\Omega^n_{\mathcal{M}}-\Xi\in\mathcal{S}'_n$\enskip; 
\item If $\Xi_1,\Xi_2,\dots\in\mathcal{S}'_n$ then $\bigcup_{m\ge 0}\Xi_m\in\mathcal{S}'_n$\enskip.
\end{enumerate}
The second and third fact follows from $\mathrm{Cyl}(\Omega^n_{\mathcal{M}}-\Xi)=\Omega_{\mathcal{M}}-\mathrm{Cyl}(\Xi)$ and 
$\mathrm{Cyl}(\bigcup_{m\ge 0}\Xi_m)=\bigcup_{m\ge 0}\mathrm{Cyl}(\Xi_m)$, respectively. Then one obtains 
$\mathcal{S}^n_\mathcal{M}\subseteq\mathcal{S}'_n$ for all $n\ge 0$ since 
$\{\mathrm{Hists}(\theta)\mid\theta\mbox{ is a template and }|\theta|=n\}\subseteq\mathcal{S}'_n$\enskip. This implies that 
$\mathrm{Cyl}(\Xi)\in\mathcal{S}'$ for all $n\ge 0$ and $\Xi\in\mathcal{S}^n_\mathcal{M}$. It follows that 
$\mathcal{S}_\mathcal{M}\subseteq\mathcal{S}'$.
\end{remark}

\section{Schedulers and Their Probability Spaces}\label{sect:scheduler}

The stochastic feature of a CTMDP is endowed by a (measurable) scheduler which resolves the action when a state is entered. In the
following, we briefly introduce schedulers for CTMDPs, which are divided into two categories: early schedulers and late schedulers. 
Most notions in this part stem from~\cite{DBLP:conf/formats/WolovickJ06,DBLP:conf/fossacs/NeuhausserSK09}. 
Below we fix a CTMDP $\mathcal{M}=\left(L,Act,\mathbf{R},\{\mathbf{w}_i\}_{1\le i\le k}\right)$\enskip. 

\begin{definition}
An \emph{early scheduler} $D$ is a function 
\[
D:Hists(\mathcal{M})\times Act\rightarrow [0,1]
\]
such that for all $\xi\in Hists(\mathcal{M})$, the following conditions hold:
\begin{itemize}\itemsep1pt \parskip0pt \parsep0pt
\item $\sum_{a\in Act} D(\xi,a)=1$\enskip;
\item for all $a\in Act$, $D(\xi,a)>0$ implies $a\in\mathrm{En}({\xi}{\downarrow})$\enskip. 
\end{itemize}
$D$ is called \emph{measurable} if for all $a\in Act$ and $n\ge 0$, the function $D(\centerdot,a)$ is measurable w.r.t 
$(\Omega_{\mathcal{M}}^n,\mathcal{S}^n_\mathcal{M})$, provided that the domain of $D(\centerdot,a)$ is restricted to 
$Hists^n(\mathcal{M})$\enskip.
\end{definition}
Intuitively, an early scheduler $D$ chooses a distribution over actions once a state is entered. The decision $D(\xi,\centerdot)$ is based on the 
history $\xi$ which records all the information of the (finite) execution of the CTMDP so far. The measurability condition will be needed 
to define a probability measure for the measurable space $(\Omega_\mathcal{M},\mathcal{S}_\mathcal{M})$. 

Then we introduce the class of late schedulers on locally-uniform CTMDPs. The major difference between early and late schedulers is that 
late schedulers can make the decision when the time-delay at the current state is over.

\begin{definition}
Suppose $\mathcal{M}$ is locally-uniform. A \emph{late scheduler} $D$ is a function 
\[
D:Hists(\mathcal{M})\times \mathbb{R}_{\ge 0}\times Act\rightarrow [0,1]
\]
such that for each $\xi\in Hists(\mathcal{M})$ and $t\in\mathbb{R}_{\ge 0}$, the following conditions hold:
\begin{itemize}\itemsep1pt \parskip0pt \parsep0pt
\item $\sum_{a\in Act} D(\xi,t,a)=1$\enskip;
\item for all $a\in Act$, $D(\xi,t,a)>0$ implies $a\in\mathrm{En}({\xi}{\downarrow})$\enskip. 
\end{itemize}
$D$ is \emph{measurable} iff for all $n\ge 0$ and $a\in Act$, the function $D(\centerdot,\centerdot,a)$ is measurable w.r.t 
$(\Omega_{\mathcal{M}}^n\times \mathbb{R}_{\ge 0},~\mathcal{S}^n_\mathcal{M}\otimes\mathcal{B}(\mathbb{R}_{\ge 0}))$, provided that the 
domain of $D(\centerdot,\centerdot,a)$ is restricted to $Hists^n(\mathcal{M})\times \mathbb{R}_{\ge 0}$\enskip. 
\end{definition}
Intuitively, a late scheduler $D$ is more powerful than an early scheduler since it can make the decision $D(\xi,t,\centerdot)$ based on
the elapsed time $t$ at ${\xi}{\downarrow}$. The locally-uniformity allows a late scheduler to make such decision, without 
mathematical ambiguity on the accumulated cost and the probability density function for the time-delay. In general, late schedulers can 
achieve better objectives due to their extra ability against early schedulers (cf., e.g., maximal time-bounded reachability 
probability~\cite{DBLP:conf/fossacs/NeuhausserSK09}).

In the rest of the paper, when we consider late schedulers, we will assume that the underlying CTMDP is locally 
uniform. We will denote by $\mathcal{E}_\mathcal{M}$ (resp. $\mathcal{L}_\mathcal{M}$) the set of early schedulers 
(resp. late schedulers). For the sake of simplicity, we use `$D$' to denote either an early or late scheduler.

Each measurable early/late scheduler will induce a probability measure on combined actions, when applied to a specific history. Firstly, 
we introduce the probability measure under the setting of early schedulers.

\begin{definition}\label{def:earlysch} Let $\xi\in Hists(\mathcal{M})$ be a history and $D$ a measurable early scheduler. The 
\emph{probability measure} $\mu^D_\mathcal{M}(\xi,\centerdot)$ for the measurable space $(\Gamma_\mathcal{M},\mathcal{U}_\mathcal{M})$ is 
defined as follows: 
\[
\mu^D_\mathcal{M}(\xi,U):=\sum_{a\in\mathrm{En}(\xi\downarrow)}D(\xi,a)\cdot\int_{\mathbb{R}_{\ge 0}} {f}_{\mathbf{E}({\xi}{\downarrow},a)}(t)\cdot\left[\sum_{s\in L}\mathbf{1}_U(a,t,s)\cdot\mathbf{P}({\xi}{\downarrow},a,s)\right]\,\mathrm{d}t
\]
for each $U\in\mathcal{U}_\mathcal{M}$.
\end{definition}
Then we introduce the probability measure under the setting of late schedulers. For the sake of simplicity, we will use the same notation 
$\mu$ (in the case of early schedulers) for late schedulers.

\begin{definition}
Let $\xi\in Hists(\mathcal{M})$ be a history and $D$ a measurable late scheduler. The \emph{probability measure} 
$\mu^D_\mathcal{M}(\xi,\centerdot)$ for the measurable space $(\Gamma_\mathcal{M},\mathcal{U}_\mathcal{M})$ is defined as follows: 
\[
\mu^D_\mathcal{M}(\xi,U):=\int_{\mathbb{R}_{\ge 0}}f_{\mathbf{E}({\xi}{\downarrow})}(t)\cdot\left\{\sum_{a\in \mathrm{En}({\xi}{\downarrow})}D(\xi,t,a)\cdot\left[\sum_{s\in L}\mathbf{1}_{U}(a,t,s)\cdot\mathbf{P}({\xi}{\downarrow},a,s)\right]\right\}\,\mathrm{d}t
\]
for each $U\in\mathcal{U}_\mathcal{M}$.
\end{definition}
Now we define the probability spaces on histories and paths. Firstly, we define the probability space on histories. To this end, we introduce the notion of concatenation as follows.

\begin{definition}
Let $\xi\in Hists(\mathcal{M})$ be a history and $(a,t,s)\in\Gamma_\mathcal{M}$ be a combined action. We define 
$\xi\circ (a,t,s)\in Hists(\mathcal{M})$ to be the history obtained by concatenating $(a,t,s)$ to ${\xi}{\downarrow}$ (i.e. 
$\xi\circ (a,t,s)=\xi[0]\dots{\xi}{\downarrow}\xrightarrow{a,t}s$)\enskip.
\end{definition}
Then the probability space on histories of fixed length is given as follows. 

\begin{definition}
Suppose $D$ is a measurable early (late) scheduler and $\alpha$ is an initial distribution. The sequence $\left\{\mathrm{Pr}_{\mathcal{M},D,\alpha}^n:\mathcal{S}_{\mathcal{M}}^n\rightarrow [0,1]\right\}_{n\ge 0}$ of probability measures is inductively as follows:
\begin{eqnarray*}
\mathrm{Pr}_{\mathcal{M},D,\alpha}^0(\Xi)&:=&\sum_{s\in \Xi}\alpha(s)\enskip;\\
\mathrm{Pr}^{n+1}_{\mathcal{M},D,\alpha}(\Xi)&:=&\int_{\Omega^n_\mathcal{M}}\left[\int_{\Gamma_{\mathcal{M}}}\mathbf{1}_\Xi(\xi\circ\gamma)~\mu^D_\mathcal{M}(\xi,\mathrm{d}\gamma)\right]\mathrm{Pr}^{n}_{\mathcal{M},D,\alpha}(\mathrm{d}\xi)\enskip;
\end{eqnarray*}
for each $\Xi\in \mathcal{S}^{n}_\mathcal{M}$\enskip.
\end{definition}
Finally, the probability space on paths is given as follows. 

\begin{definition}\label{def:prs}
Let $D$ be a measurable early (late) scheduler and $\alpha$ be an initial distribution. The probability space $(\Omega_{\mathcal{M}}, \mathcal{S}_{\mathcal{M}}, \mathrm{Pr}_{\mathcal{M},D,\alpha})$ is defined as follows:
\begin{itemize}\itemsep1pt \parskip0pt \parsep0pt
\item $\Omega_{\mathcal{M}}$ and $\mathcal{S}_{\mathcal{M}}$ is defined as in Definition~\ref{def:msp};
\item $\mathrm{Pr}_{\mathcal{M},D,\alpha}$ is the unique probability measure such that 
\[
\mathrm{Pr}_{\mathcal{M},D,\alpha}(\mathrm{Cyl}(\Xi))=\mathrm{Pr}_{\mathcal{M},D,\alpha}^{n}\left(\Xi\right)
\]
for all $n\ge 0$ and $\Xi\in\mathcal{S}^n_\mathcal{M}$\enskip.
\end{itemize}
\end{definition}
We end this section with a fundamental property asserting that the role of initial distribution $\alpha$ can be decomposed into Dirac 
distributions on individual states. 

\begin{proposition}\label{prop:init}
For each early (late) scheduler $D$ and each initial distribution $\alpha$, 
$\mathrm{Pr}_{\mathcal{M},D,\alpha}(\Pi)=\sum_{s\in L}\alpha(s)\cdot\mathrm{Pr}_{\mathcal{M},D,\mathcal{D}[s]}(\Pi)$
for all $\Pi\in\mathcal{S}_\mathcal{M}$. 
\end{proposition}

\section{A General Integral Characterization}\label{sect:fixpoint}

In this section we derive a general integral characterization for the probability measure on paths. Below we fix a CTMDP 
$\mathcal{M}=\left(L,Act,\mathbf{R},\{\mathbf{w}_i\}_{1\le i\le k}\right)$\enskip. For the sake of simplicity, We will omit  
`$\mathcal{M}$'s which appear in the subscript of the notation `$\mathrm{Pr}$'. 

First, we define shifting functions on histories and paths which shifts each path/history by one transition step. 
\begin{definition}\label{def:shift:set}
Given $\Pi\in\mathcal{S}_\mathcal{M}$ and $(s,a)\in L\times Act$, we define the function 
$P^{s,a}_{\Pi}:\mathbb{R}_{\ge 0}\rightarrow 2^{Paths(\mathcal{M})}$ by: 
$P^{s,a}_\Pi(t):=\{\pi\in Paths(\mathcal{M})\mid s\xrightarrow{a,t}\pi\in\Pi\}$\enskip.

Analogously, given $\Xi\in\mathcal{S}^n_\mathcal{M}$ with $n\ge 1$ and $(s,a)\in L\times Act$, we define 
$H^{s,a}_{\Xi}:\mathbb{R}_{\ge 0}\rightarrow 2^{Hists^{n-1}(\mathcal{M})}$ by:
$H^{s,a}_\Xi(t):=\{\xi\in Hists^{n-1}(\mathcal{M})\mid s\xrightarrow{a,t}\xi\in\Xi\}$\enskip.
\end{definition}
We also define the shifted scheduler of a measurable early/late scheduler.

\begin{definition}\label{def:shift:scheduler}
Let $s\in L$, $a\in Act$ and $t\in\mathbb{R}_{\ge 0}$. For each measurable early (late) scheduler $D$, the early (late) scheduler 
$D[s\xrightarrow{a,t}]$ is defined as follows: 
\begin{itemize}\itemsep1pt \parskip0pt \parsep0pt
\item $D[s\xrightarrow{a,t}](\xi)=D(s\xrightarrow{a,t}\xi)$ for all $\xi\in Hists(\mathcal{M})$, provided that $D$ is an early scheduler;
\item $D[s\xrightarrow{a,t}](\xi,\tau)=D(s\xrightarrow{a,t}\xi,\tau)$ for all $(\xi,\tau)\in Hists(\mathcal{M})\times\mathbb{R}_{\ge 0}$, provided that $D$ is a late scheduler. 
\end{itemize}
\end{definition}
The following lemma states that each shifted set of paths/histories is measurable w.r.t corresponding measurable space.

\begin{lemma}\label{lemm:shift:set:measurability}
$P^{s,a}_\Pi(t)\in\mathcal{S}_\mathcal{M}$ for all $\Pi\in\mathcal{S}_\mathcal{M}$, $(s,a)\in L\times Act$ and 
$t\in\mathbb{R}_{\ge 0}$\enskip. Analogously, $H^{s,a}_\Xi(t)\in\mathcal{S}^{n-1}_\mathcal{M}$ for all $n\ge 1$, 
$\Xi\in\mathcal{S}^n_\mathcal{M}$ , $(s,a)\in L\times Act$ and $t\in\mathbb{R}_{\ge 0}$\enskip.
\end{lemma}
Moreover, each shifted scheduler is measurable, as is illustrated in the following lemma. 

\begin{lemma}\label{lemm:shift:schd:measurability}
Let $s\in L$, $a\in Act$ and $t\in\mathbb{R}_{\ge 0}$. For each measurable early (late) scheduler $D$, $D[s\xrightarrow{a,t}]$ is a 
measurable early (late) scheduler.
\end{lemma}
Based on Lemma~\ref{lemm:shift:set:measurability} and Lemma~\ref{lemm:shift:schd:measurability}, we define a shift probability function as 
follows.

\begin{definition}\label{def:shift:probability}
Let $s\in L$, $a\in Act$, $\Pi\in\mathcal{S}_\mathcal{M}$ and $D$ a measurable early (late) scheduler. Define 
$p_{\Pi,D}^{s,a}:\mathbb{R}_{\ge 0}\rightarrow [0,1]$ by: 
\[
p_{\Pi,D}^{s,a}(t)=\mathrm{Pr}_{D[s\xrightarrow{a,t}],\mathbf{P}(s,a,\centerdot)}\left(P_\Pi^{s,a}(t)\right)
\]
for all $t\ge 0$.
\end{definition}

The following proposition states that the shift probability function is measurable w.r.t $(\mathbb{R}_{\ge 0}, \mathcal{B}(\mathbb{R}_{\ge 0}))$.

\begin{proposition}\label{prop:shiftfunc:measurability}
$p_{\Pi,D}^{s,a}$ is a measurable function w.r.t $(\mathbb{R}_{\ge 0}, \mathcal{B}(\mathbb{R}_{\ge 0}))$ given any $\Pi\in\mathcal{S}_\mathcal{M}$, $s\in L$, $a\in Act$ and measurable early (late) scheduler $D$.
\end{proposition}

Below we present the integral characterization for measurable early schedulers.  

\begin{theorem}\label{thm:fix-early}
Let $D$ be a measurable early scheduler. For each $\Pi\in\mathcal{S}_\mathcal{M}$ and $s\in L$, we have
\[
\mathrm{Pr}_{D,\mathcal{D}[s]}(\Pi)=\sum_{a\in\mathrm{En}(s)}D(s,a)\cdot\int_{0}^\infty f_{\mathbf{E}(s,a)}(t)\cdot p^{s,a}_{\Pi,D}(t)\,\mathrm{d}t\enskip.
\]
\end{theorem}
Intuitively, the integrand $f_{\mathbf{E}(s,a)}(t)\cdot p^{s,a}_{\Pi,D}(t)$ can be viewed as certain ``probability density function''.  
Similarly, we can obtain the integral characterization for measurable late schedulers. 

\begin{theorem}\label{thm:fix-late}
Let $D$ be a measurable late scheduler. For each $\Pi\in\mathcal{S}_\mathcal{M}$ and $s\in S$, we have
\[
\mathrm{Pr}_{D,\mathcal{D}[s]}(\Pi)=\int_{0}^\infty f_{\mathbf{E}(s)}(t)\cdot\left[\sum_{a\in\mathrm{En}(s)} D(s,t,a)\cdot p^{s,a}_{\Pi,D}(t)\right]\,\mathrm{d}t\enskip.
\]
\end{theorem}

\section{Maximal Cost-Bounded Reachability \\Probability}\label{sec:mcbrp}

In this section, we consider maximal cost-bounded reachability probabilities. Below we fix a CTMDP 
$\mathcal{M}=\left(L,Act,\mathbf{R},\{\mathbf{w}_i\}_{1\le i\le k}\right)$ and a set $G\subseteq L$\enskip. 

\begin{definition}\label{def:cbrp}
Let $D$ be a measurable early (late) scheduler. Define the function 
$\mathrm{prob}^{D}_{G}: L\times\mathbb{R}^k\rightarrow [0,1]$ by: 
$\mathrm{prob}^{D}_{G}(s,\mathbf{c}):=\mathrm{Pr}_{D,\mathcal{D}[s]}\left(\Pi_G^\mathbf{c} \right)$ where
\[
\Pi_G^\mathbf{c}:=\{\pi\in Paths(\mathcal{M})\mid\mathbf{C}(\pi,G)\le \mathbf{c}\}\enskip.
\]
Define $\mathrm{prob}^{\mathrm{e},\max}_{G}:L\times\mathbb{R}^k\rightarrow [0,1]$ and $\mathrm{prob}^{\mathrm{l},\max}_{G}: L\times\mathbb{R}^k\rightarrow [0,1]$ by:
\begin{itemize}\itemsep1pt \parskip0pt \parsep0pt
\item $\mathrm{prob}^\mathrm{e,\max}_{G}(s,\mathbf{c}):=\sup_{D\in\mathcal{E}_\mathcal{M}}\mathrm{prob}^{D}_{G}(s,\mathbf{c})$\enskip;
\item $\mathrm{prob}^\mathrm{l,\max}_{G}(s,\mathbf{c}):=\sup_{D\in\mathcal{L}_\mathcal{M}}\mathrm{prob}^{D}_{G}(s,\mathbf{c})$\enskip;
\end{itemize}
for $s\in L$ and $\mathbf{c}\in\mathbb{R}^k$\enskip.
\end{definition}
From the definition, we can see that $\Pi_G^\mathbf{c}$ is the set of paths which can reach $G$ within cost $\mathbf{c}$, 
$\mathrm{prob}^\mathrm{e,\max}_{G}(s,\mathbf{c})$ (resp. $\mathrm{prob}^\mathrm{l,\max}_{G}(s,\mathbf{c})$) is the maximal 
probability of $\Pi_G^\mathbf{c}$ with initial distribution $\mathcal{D}(s)$ (i.e., fixed initial state $s$) ranging over all early (resp. 
late) schedulers. It is not hard to verify that $\Pi^\mathbf{c}_G$ is measurable, thus all functions in Definition~\ref{def:cbrp} are 
well-defined. It is worth noting that if $\mathbf{c}\not\ge \vec{0}$, then both $\mathrm{prob}^{D}_{G}(s,\mathbf{c})$, 
$\mathrm{prob}^\mathrm{e,\max}_{G}(s,\mathbf{c})$ and $\mathrm{prob}^\mathrm{l,\max}_{G}(s,\mathbf{c})$ is zero. 

From Theorem~\ref{thm:fix-early}, Theorem~\ref{thm:fix-late} and Proposition~\ref{prop:init}, we can directly obtain the following 
results. 

\begin{corollary}\label{crlly:int_early}
Let $D$ be a measurable early scheduler. The function $\mathrm{prob}^D_G$ satisfies the following conditions:
\begin{enumerate}\itemsep1pt \parskip0pt \parsep0pt
\item If $s\in G$ then $\mathrm{prob}^D_G(s,\mathbf{c})=\mathbf{1}_{\mathbb{R}^k_{\ge 0}}(\mathbf{c})$\enskip;
\item If $s\not\in G$ then 
\begin{align*}
& \mathrm{prob}^D_G(s,\mathbf{c})= \sum_{a\in\mathrm{En}(s)}D(s,a)~\cdot\\
& \quad\int_{0}^\infty f_{\mathbf{E}(s,a)}(t)\cdot\left[\sum_{s'\in L}\mathbf{P}(s,a,s')\cdot\mathrm{prob}^{D[s\xrightarrow{a,t}]}_G(s',\mathbf{c}-t\cdot\mathbf{w}(s,a))\right]\,\mathrm{d}t\enskip.
\end{align*}
\end{enumerate}
\end{corollary}

\begin{corollary}\label{crlly:int_late}
Let $D$ be a measurable late scheduler. The function $\mathrm{prob}^D_G$ satisfies the following conditions:
\begin{enumerate}\itemsep1pt \parskip0pt \parsep0pt
\item If $s\in G$ then $\mathrm{prob}^D_G(s,\mathbf{c})=\mathbf{1}_{\mathbb{R}^k_{\ge 0}}(\mathbf{c})$\enskip;
\item If $s\not\in G$ then 
\begin{align*}
& \mathrm{prob}^D_G(s,\mathbf{c})= \int_{0}^\infty f_{\mathbf{E}(s)}(t)\cdot\\
& \quad\left\{\sum_{a\in\mathrm{En}(s)}D(s,t,a)\cdot
\left[\sum_{s'\in L}\mathbf{P}(s,a,s')\cdot\mathrm{prob}^{D[s\xrightarrow{a,t}]}_G(s',\mathbf{c}-t\cdot\mathbf{w}(s))\right]\right\}\,\mathrm{d}t\enskip.
\end{align*}
\end{enumerate}
\end{corollary}
The following theorem mainly presents the fixed-point characterization for $\mathrm{prob}^\mathrm{e,\max}_{G}$, while it also states that $\mathrm{prob}^\mathrm{e,\max}_{G}$ is Lipschitz continuous.

\begin{theorem}\label{thm:minearlyf}
The function $\mathrm{prob}^{\mathrm{e},\max}_G(\centerdot,\centerdot)$ is the least fixed-point (w.r.t $\le$) of the high-order operator $\mathcal{T}^{\mathrm{e}}_G:\left[L\times\mathbb{R}^k\rightarrow [0,1]\right]\rightarrow\left[L\times\mathbb{R}^k\rightarrow [0,1]\right]$ defined by:  
\begin{itemize}\itemsep1pt \parskip0pt \parsep0pt
\item $\mathcal{T}^{\mathrm{e}}_G(h)(s,\mathbf{c}):=\mathbf{1}_{\mathbb{R}^k_{\ge 0}}(\mathbf{c})$ for all $s\in G$ and $\mathbf{c}\in\mathbb{R}^k$;
\item given any $s\in L-G$ and $\mathbf{c}\in\mathbb{R}^k$, 
\begin{align*}
& \mathcal{T}^\mathrm{e}_G(h)(s,\mathbf{c}):=\\
& \quad\max_{a\in\mathrm{En}(s)}\int_{0}^\infty f_{\mathbf{E}(s,a)}(t)\cdot\left[\sum_{s'\in L}\mathbf{P}(s,a,s')\cdot h(s',\mathbf{c}-t\cdot\mathbf{w}(s,a))\right]\,\mathrm{d}t\enskip;
\end{align*}
\end{itemize}
for each $h:L\times\mathbb{R}^k\rightarrow [0,1]$. Moreover, 
\[
\left|\mathrm{prob}^{\mathrm{e},\max}_G(s,\mathbf{c})-\mathrm{prob}^{\mathrm{e},\max}_G(s,\mathbf{c}')\right|\le \frac{\mathbf{E}_{\max}}{\mathbf{w}_{\min}}\cdot{\parallel}{\mathbf{c}-\mathbf{c}'}{\parallel}_\infty
\]
for all $\mathbf{c},\mathbf{c}'\ge \vec{0}$ and $s\in L$\enskip.
\end{theorem}
The counterpart for late schedulers is illustrated as follows.

\begin{theorem}\label{thm:minlatef}
The function $\mathrm{prob}^{\mathrm{l},\max}_G$ is the least fixed-point (w.r.t $\le$) of the high-order operator $\mathcal{T}^{\mathrm{l}}_G:\left[L\times\mathbb{R}^k\rightarrow [0,1]\right]\rightarrow\left[L\times\mathbb{R}^k\rightarrow[0,1]\right]$ defined as follows:  
\begin{itemize}\itemsep1pt \parskip0pt \parsep0pt
\item $\mathcal{T}^\mathrm{l}_G(h)(s,\mathbf{c}):=\mathbf{1}_{\mathbb{R}^k_{\ge 0}}(\mathbf{c})$ if $s\in G$;
\item If $s\not\in G$ then 
\begin{align*}
& \mathcal{T}^\mathrm{l}_G(h)(s,\mathbf{c}):=\\
& \quad\int_{0}^\infty f_{\mathbf{E}(s)}(t)\cdot\max_{a\in\mathrm{En}(s)}\left[\sum_{s'\in L}\mathbf{P}(s,a,s')\cdot h(s',\mathbf{c}-t\cdot\mathbf{w}(s))\right]\,\mathrm{d}t\enskip;
\end{align*}
\end{itemize}
for each $h:L\times\mathbb{R}^k\rightarrow [0,1]$. Moreover, 
\[
\left|\mathrm{prob}^{\mathrm{l},\max}_G(s,\mathbf{c})-\mathrm{prob}^{\mathrm{l},\max}_G(s,\mathbf{c}')\right|\le \frac{\mathbf{E}_{\max}}{\mathbf{w}_{\min}}\cdot{\parallel}{\mathbf{c}-\mathbf{c}'}{\parallel}_\infty
\]
for all $\mathbf{c},\mathbf{c}'\ge \vec{0}$ and $s\in L$\enskip.
\end{theorem}
The Lipschitz constant $\frac{\mathbf{E}_{\max}}{\mathbf{w}_{\min}}$ will be crucial to the error bound of our approximation algorithm.

Now we describe the proof error in~\cite{DBLP:phd/de/Neuhausser2010,DBLP:conf/qest/NeuhausserZ10}. The error lies in the proof of
~\cite[Lemma 5.1 on Pages 119]{DBLP:phd/de/Neuhausser2010} which tries to prove that the time-bounded reachability probability functions are 
continuous. In detail, the error is at the proof for right-continuity of the functions. Let us take the sentence 
``This implies ... for some $\xi\le\frac{\epsilon}{2}$.'' from line -3 to line -2 on page 119 as (*). (*) is wrong in general, 
as one can treat $D$'s as natural numbers, and define 
\[
\mathrm{Pr}_n(\mbox{``reach }G\mbox{ within }z\mbox{''}):=
\begin{cases}
n\cdot z & \mbox{if }z\in [0,\frac{1}{n}]\\
1        &  \mbox{if }z\in (\frac{1}{n},\infty)
\end{cases}\enskip.
\]
Then $\sup_n \mathrm{Pr}_n(\mbox{``reach }G\mbox{ within }z\mbox{''})$ equals 1 for $z>0$ and 0 for $z=0$. Thus 
$\sup_D \mathrm{Pr}_D(\mbox{``reach }G\mbox{ within }z\mbox{''})$ on $z\ge 0$ is right-discontinuous at $z=0$, which does not satisfy (*) 
(treat $D$ as a natural number). Note that a concrete counterexample does not exist as \cite[Lemma 5.1]{DBLP:phd/de/Neuhausser2010} is 
correct due to this paper; it is the proof that is flawed. Also note that Lemma 5.1 is important as the least fixed-point characterization 
\cite[Theorem 5.1 on Page 120]{DBLP:phd/de/Neuhausser2010} and the optimal scheduler 
\cite[Theorem 5.2 on page 124]{DBLP:phd/de/Neuhausser2010} directly rely on it. 
We fix the error in the more general setting of cost-bounded reachability probability by providing new proofs as illustrated in 
(the appendixes) of Section 4, Section 5 and Section 6. 

\section{Optimal Early/Late Schedulers}

In this section, we establish optimal early/late schedulers for maximal cost-bounded reachability probability. We show that there exists a 
deterministic cost-positional early/late scheduler that achieves the maximal cost-bounded reachability probability. 
Below we fix a CTMDP $\mathcal{M}=\left(L,Act,\mathbf{R},\{\mathbf{w}_i\}_{1\le i\le k}\right)$\enskip. We first 
introduce the notion of deterministic cost-positional schedulers.

\begin{definition}
A measurable early scheduler $D$ is called \emph{deterministic cost-positional} iff (i) $D(\xi,\centerdot)=D(\xi',\centerdot)$ whenever 
${\xi}{\downarrow}={\xi'}{\downarrow}$ and $\mathbf{C}({\xi})=\mathbf{C}({\xi'})$, and (ii) $D(\xi,\centerdot)$ is 
Dirac for all histories $\xi$. 

A measurable late sheduler $D$ is called \emph{deterministic cost-positional} iff (i) $D(\xi,t,\centerdot)=D(\xi',t',\centerdot)$ whenever 
${\xi}{\downarrow}={\xi'}{\downarrow}$ and 
$\mathbf{C}({\xi})+t\cdot\mathbf{w}({\xi}{\downarrow})=\mathbf{C}({\xi'})+t'\cdot\mathbf{w}({\xi'}{\downarrow})$, 
and (ii) $D(\xi,t,\centerdot)$ is Dirac for all histories $\xi$ and $t\ge 0$.
\end{definition}
Intuitively, a deterministic cost-positional scheduler makes its decision solely on the current state and the cost accumulated so far, and 
its decision is always Dirac. Below we show that such a scheduler suffices to achieve maximal cost-bounded reachability probability under 
the context of early schedulers. 

\begin{theorem}\label{thm:opt:earlyscheduler}
For all $\mathbf{c}\in\mathbb{R}^k$ and $G\subseteq L$, there exists a measurable deterministic cost-positional early scheduler 
$\mathtt{D}_\mathbf{c}^\mathrm{e}$ such that
$\mathrm{prob}^{\mathrm{e},\max}_G(s,\mathbf{c})=\mathrm{Pr}_{\mathtt{D}_\mathbf{c}^\mathrm{e},\mathcal{D}[s]}(\Pi_G^\mathbf{c})$ for all 
$s\in L$\enskip.
\end{theorem}
\begin{proof}
Let $G\subseteq L$. Fix an arbitrary linear order $\preceq$ on $Act$. We construct the measurable early scheduler 
$\mathtt{D}_\mathbf{c}^\mathrm{e}$ (for each $\mathbf{c}\in\mathbb{R}^k$) as follows. Fix a $\mathbf{c}\in\mathbb{R}^k$.
Define the function $g:L\times Act\times \mathbb{R}^k\rightarrow [0,1]$ by:
\begin{align*}
& g(s,a,\mathbf{x}):=\\
& \quad\int_0^\infty f_{\mathbf{E}({s},a)}(t)\cdot\left[\sum_{s'\in L}\mathbf{P}(s,a,s')\cdot\mathrm{prob}^{\mathrm{e},\max}_G(s', \mathbf{c}-\mathbf{x}-t\cdot\mathbf{w}(s,a))\right]\,\mathrm{d}t\enskip.
\end{align*}
Note that $\mathrm{prob}_G^{\mathrm{e},\max}(s,\mathbf{c}-\mathbf{x})=\max_{a\in \mathrm{En}(s)} g(s,a,\mathbf{x})$ if $s\not\in G$.
Consider an arbitrary $\xi\in Hists(\mathcal{M})$. Define
\begin{align*}
& L^{\xi}_1:=\Big\{s\in L\mid \exists a^*\in\mathrm{En}(s).\Big[g(s, a^*,\mathbf{C}(\xi))=\max_{a\in\mathrm{En}(s)}g(s,a,\mathbf{C}(\xi))\\
& \qquad\qquad\qquad\qquad\wedge \mathbf{w}(s,a^*)\ne\vec{0}\Big]\Big\}
\end{align*}
and $L^{\xi}_2:=\left\{s\in L\mid \mathrm{prob}_G^{\mathrm{e},\max}(s,\mathbf{c}-\mathbf{C}(\xi))=0\right\}$\enskip. The probability distribution 
$\mathtt{D}_\mathbf{c}^\mathrm{e}(\xi,\centerdot)$ is determined by the following procedure:
\begin{enumerate}\itemsep1pt \parskip0pt \parsep0pt
\item If ${\xi}{\downarrow}\in L^{\xi}_2$, then we set $\mathtt{D}_\mathbf{c}^\mathrm{e}(\xi,\centerdot)=\mathcal{D}[a^*]$ for which 
$a^*\in\mathrm{En}({\xi}{\downarrow})$ is arbitrarily fixed.
\item If ${\xi}{\downarrow}\in L^{\xi}_1\backslash L^{\xi}_2$, then we set 
$\mathtt{D}_\mathbf{c}^\mathrm{e}(\xi,\centerdot)=\mathcal{D}[a^*]$ for which $a^*\in\mathrm{En}({\xi}{\downarrow})$ satisfies that 
(i) $g({\xi}{\downarrow}, a^*,\mathbf{C}(\xi))=\max_{a\in\mathrm{En}(\xi\downarrow)}g({\xi}{\downarrow},a,\mathbf{C}(\xi))$ and 
(ii) $\mathbf{w}({\xi}{\downarrow}, a^*)\ne\vec{0}$; if there are multiple such $a^*$'s, we choose the least of them w.r.t $\preceq$.
\item If ${\xi}{\downarrow}\in L-(L^{\xi}_1\cup L^{\xi}_2)$, then we set $\mathtt{D}_\mathbf{c}^\mathrm{e}(\xi,\centerdot)$ to be an action 
$a^*$ which satisfies that (i) 
$g({\xi}{\downarrow},a^*,\mathbf{C}(\xi))=\max_{a\in\mathrm{En}({\xi}{\downarrow})}g({\xi}{\downarrow},a,\mathbf{C}(\xi))$ and (ii) there 
exists $s\in L$ such that $\mathbf{P}({\xi}{\downarrow},a^*,s)>0$ and the distance from $s$ to $L^\xi_1\backslash L^\xi_2$ is (one-step) smaller than that 
from ${\xi}{\downarrow}$ in the digraph 
\begin{align*}
& \mathcal{G}^\xi:=\Big(L,\Big\{(u,v)\in L\times L\mid \\
& \quad\exists b\in\mathrm{En}(u).\big[\mathbf{P}(u,b,v)>0\wedge g(u,b,\mathbf{C}(\xi))=\max_{a\in\mathrm{En}(u)}g(u,a,\mathbf{C}(\xi))\big]\Big\}\Big)\enskip. 
\end{align*}
If there are multiple such $a^*$'s, choose the least of them w.r.t $\preceq$.
\end{enumerate}
The legitimacy of the third step in the procedure above follows from Proposition~\ref{prop:approx:early:weightzero} to be proved later: 
the set $L''$ of states that cannot reach $L_1^\xi\backslash L_2^\xi$ should be empty, or otherwise one can reduce all values 
$\{\mathrm{prob}^{\mathrm{e},\max}_G(s,\mathbf{c}-\mathbf{C}(\xi))\}_{s\in L''}$ by a small amount to obtain a pre-fixed-point smaller than 
$\{\mathrm{prob}^{\mathrm{e},\max}_G(s,\mathbf{c}-\mathbf{C}(\xi))\}_{s\in L}$.
By definition, $\mathtt{D}_\mathbf{c}^\mathrm{e}$ is deterministic cost-positional. Note that there are finitely many triples 
$(L_1^\xi,L_2^\xi,\mathcal{G}^\xi)$ since $L$ is finite. Also, by Theorem~\ref{thm:minearlyf}, 
$\mathbf{x}\mapsto g(s,a,\mathbf{x})$ is separately continuous on $\{\mathbf{x}\mid \mathbf{x}\le\mathbf{c}\}$ and its complement set, 
for all $s\in L$ and $a\in\mathrm{En}(s)$. Thus, the set of all 
histories $\xi$ with length $n$ such that the triple $(L_1^\xi,L_2^\xi,\mathcal{G}^\xi)$ happens to be a specific one is measurable w.r.t 
$(\Omega_\mathcal{M}^n,\mathcal{S}^n_\mathcal{M})$. It follows that $\mathtt{D}_\mathbf{c}^\mathrm{e}$ is a measurable scheduler. We prove 
that $\mathrm{prob}^{\mathrm{e},\max}_G(s,\mathbf{c})=\mathrm{Pr}_{\mathtt{D}_\mathbf{c}^\mathrm{e},\mathcal{D}[s]}(\Pi_G^\mathbf{c})$
for all $s\in L$ and $\mathbf{c}\ge\vec{0}$. 

Define the function $h:L\times\mathbb{R}^k_{\ge 0}\rightarrow [0,1]$ by 
$h(s, \mathbf{c}):=\mathrm{Pr}_{\mathtt{D}_\mathbf{c}^\mathrm{e},\mathcal{D}[s]}(\Pi_G^\mathbf{c})$. 
Suppose that $\mathrm{prob}^{\mathrm{e},\max}_G\ne h$. Let $\overline{h}:=|\mathrm{prob}^{\mathrm{e},\max}_G-h|$. 
Then there exists $s\in L$ and $\mathbf{c}\ge\vec{0}$ such that $h(s,\mathbf{c})>0$. Define 
$T:=\max_{1\le i\le k}\{\frac{\mathbf{c}_i}{\mathbf{w}_{\min}}\}$\enskip. Let
\[
d:=\sup\left\{\overline{h}(s,\mathbf{c}')\mid s\in L,\vec{0}\le\mathbf{c}'\le\mathbf{c}\right\}
\]
and
\[
d':=\sup\left\{\overline{h}(s,\mathbf{c}')\mid s\in L,\vec{0}\le\mathbf{c}'\le\mathbf{c}\mbox{ and }\mathbf{w}(s,a^*)\ne\vec{0}\mbox{ where } 
\mathtt{D}_{\mathbf{c}'}^\mathrm{e}(s,\centerdot)=\mathcal{D}[a^*]\right\}\enskip.
\]
We first show that $d'<d$ by a (nested) contradiction proof.

Suppose $d'=d$. Choose $\epsilon>0$ such that $d>e^{\mathbf{E}_{\max}\cdot T}\cdot\epsilon$. Then choose $s\in L$ and 
$\vec{0}\le\mathbf{c}'\le\mathbf{c}$ such that $d-\epsilon<|\overline{h}(s,\mathbf{c}')|\le d$ and $\mathbf{w}(s,a^*)\ne\vec{0}$,
where $\mathcal{D}[a^*]=\mathtt{D}_\mathbf{c}^\mathrm{e}(s,\centerdot)$.
By Theorem~\ref{thm:fix-early}, we have
\[
h(s,\mathbf{c}')=\int_0^\infty f_{\mathbf{E}(s,a^*)}(t)\cdot \left[\sum_{s'\in L}\mathbf{P}(s,a^*,s')\cdot h(s',\mathbf{c}'-t\cdot\mathbf{w}(s,a^*))\right]\,\mathrm{d}t\enskip.
\]
Then with Theorem~\ref{thm:minearlyf}, we obtain
\[
\overline{h}(s,\mathbf{c}')\le\int_0^T f_{\mathbf{E}(s,a^*)}(t)\cdot\left[\sum_{s'\in L}\mathbf{P}(s,a^*,s')\cdot
\overline{h}(s',\mathbf{c}'-t\cdot\mathbf{w}(s,a^*))\right]\,\mathrm{d}t\enskip.
\]
This implies $d-\epsilon\le (1-e^{-\mathbf{E}_{\max}\cdot T})\cdot d$, which in turn implies 
$d\le e^{\mathbf{E}_{\max}\cdot T}\cdot\epsilon$. Contradiction to the choice of $\epsilon$.

Thus $d>d'$. Let $\delta:=d-d'$ and $\epsilon:=\mathbf{P}_{\min}^{|L|}\cdot\delta$. We inductively construct a finite sequence $s_0,s_1,\dots,s_l$ ($1\le l\le |L|$) which satisfies
\[
\overline{h}(s_i,\mathbf{c}')> d-\mathbf{P}^{-i}_{\min}\cdot\epsilon~(i=0,\dots,l)
\]
as follows. Note that the triple $(L_1^s,L_2^s,\mathcal{G}^s)$ for some $s\in L$ (w.r.t some $\mathbf{c}\in\mathbb{R}^k$) remains constant as $s$ varies, since $\mathbf{C}(s)=\vec{0}$. 
\begin{enumerate}\itemsep1pt \parskip0pt \parsep0pt
\item Initially, we set $i=0$ and choose $s_0\in L,\vec{0}\le\mathbf{c}'\le\mathbf{c}$ such that
\[
d-\epsilon<\overline{h}(s_0,\mathbf{c}')\le d\enskip.
\]
\item As long as $i\le l$ and $0\le d-\mathbf{P}_{\min}^{-i}\cdot\epsilon<\overline{h}(s_i,\mathbf{c}')\le d$\enskip, we have $s_i\in L\backslash(L^{s_i}_1\cup L^{s_i}_2)$ (w.r.t $\mathbf{c}'$) since $\mathbf{P}_{\min}^{-i}\cdot\epsilon\le\delta$. Thus
\[
h(s_i,\mathbf{c}')=\sum_{s'\in L}\mathbf{P}(s_i,a^*,s')\cdot h(s',\mathbf{c}')
\]
and there exists $u\in L$ such that (i) $\mathbf{P}(s_i,a^*,u)>0$ and (ii) via $u$ the distance to $L^{s_i}_1\backslash L^{s_i}_2$ in $\mathcal{G}^{s_i}$ (w.r.t $\mathbf{c}'$) is decreased by one, for which $\mathcal{D}[a^*]=\mathtt{D}_{\mathbf{c}'}^\mathrm{e}(s_i)$. Moreover, by
\[
\overline{h}(s_i,\mathbf{c}')\le\sum_{s'\in L}\mathbf{P}(s_i,a^*,s')\cdot\overline{h}(s',\mathbf{c}')
\]
we obtain
\[
d-\mathbf{P}^{-i}_{\min}\epsilon<(1-{\mathbf{P}}_{\min})\cdot d+{\mathbf{P}}_{\min}\cdot\overline{h}(u,\mathbf{c}')\enskip,
\]
which is further reduced to $\overline{h}(u,\mathbf{c}')> d-\mathbf{P}^{-(i+1)}_{\min}\cdot\epsilon$\enskip. We set $s_{i+1}=u$;
\item If $s_{i+1}\in L^{s_{i+1}}_1\backslash L^{s_{i+1}}_2$, then the construction is terminated. Otherwise, back to Step 2.
\end{enumerate}
The legitimacy and termination (within $|L|$ steps) of the inductive construction follows from the definition of 
$\mathtt{D}_{\mathbf{c}'}^\mathrm{e}$. By $s_l\in L^{s_l}_1\backslash L^{s_l}_2$, we obtain
\[
d-\delta\ge \overline{h}(s_l,\mathbf{c}')> d-\mathbf{P}^{-l}_{\min}\cdot\epsilon~~(l\le |L|)\enskip,
\]
which is a contradiction due to $\epsilon=\mathbf{P}^{|L|}_{\min}\cdot\delta$. Thus $\mathrm{prob}^{\mathrm{e},\max}_G= h$.
\end{proof}
We can obtain a similar result for late schedulers.

\begin{theorem}\label{thm:opt:latescheduler}
For all $\mathbf{c}\in\mathbb{R}^k$ and $G\subseteq L$, there exists a deterministic cost-positional measurable late scheduler $\mathtt{D}^\mathrm{l}_\mathbf{c}$ such that
$\mathrm{prob}^{\mathrm{l},\max}_G(s,\mathbf{c})=\mathrm{Pr}_{\mathtt{D}^\mathrm{l}_\mathbf{c},\mathcal{D}[s]}(\Pi_G^\mathbf{c})$ for all 
$s\in L$\enskip.
\end{theorem}
\begin{proof}
Let $G\subseteq L$. Fix an arbitrary linear order $\preceq$ on $Act$. We construct the measurable late scheduler 
$\mathtt{D}_\mathbf{c}^\mathrm{l}$ (for each $\mathbf{c}\in\mathbb{R}^k$) as follows. Fix a $\mathbf{c}\in\mathbb{R}^k$.
Define the function $g:L\times Act\times \mathbb{R}^k\times\mathbb{R}_{\ge 0}\rightarrow [0,1]$ by:
\begin{align*}
& g(s,a,\mathbf{x},t):=\\
& \quad\sum_{s'\in L}\mathbf{P}(s,a,s')\cdot\mathrm{prob}^{\mathrm{l},\max}_G(s', \mathbf{c}-\mathbf{x}-t\cdot\mathbf{w}(s))\enskip.
\end{align*}
Note that $\mathrm{prob}^{\mathrm{l},\max}_G(s,\mathbf{c}-\mathbf{x})=\int_0^\infty f_{\mathbf{E}(s)}(t)\cdot\max_{a\in\mathrm{En}(s)}g(s,a,\mathbf{x},t)\,\mathrm{d}t$ if $s\not\in G$.
Consider arbitrary $\xi\in Hists(\mathcal{M})$ and $t\in\mathbb{R}_{\ge 0}$. Define
\[
L_1:=\Big\{s\in L\mid \mathbf{w}(s)\ne\vec{0}\Big\}
\]
and $L^{\xi,t}_2:=\left\{s\in L\mid \mathrm{prob}_G^{\mathrm{l},\max}(s,\mathbf{c}-\mathbf{C}(\xi)-t\cdot\mathbf{w}(s))=0\right\}$\enskip. The probability distribution 
$\mathtt{D}_\mathbf{c}^\mathrm{l}(\xi,t,\centerdot)$ is determined by the following procedure:
\begin{enumerate}\itemsep1pt \parskip0pt \parsep0pt
\item If ${\xi}{\downarrow}\in L^{\xi,t}_2$, then we set $\mathtt{D}_\mathbf{c}^\mathrm{l}(\xi,t,\centerdot)=\mathcal{D}[a^*]$ for which 
$a^*\in\mathrm{En}({\xi}{\downarrow})$ is arbitrarily fixed.
\item If ${\xi}{\downarrow}\in L_1\backslash L^{\xi,t}_2$, then we set 
$\mathtt{D}_\mathbf{c}^\mathrm{l}(\xi,t,\centerdot)=\mathcal{D}[a^*]$ for which $a^*\in\mathrm{En}({\xi}{\downarrow})$ satisfies that 
$g({\xi}{\downarrow}, a^*,\mathbf{C}(\xi),t)=\max_{a\in\mathrm{En}(\xi\downarrow)}g({\xi}{\downarrow},a,\mathbf{C}(\xi),t)$; if there are multiple such $a^*$'s, we choose the least of them w.r.t $\preceq$.
\item If ${\xi}{\downarrow}\in L-(L_1\cup L^{\xi,t}_2)$, then we set $\mathtt{D}_\mathbf{c}^\mathrm{e}(\xi,t,\centerdot)$ to be an action 
$a^*$ which satisfies that (i) 
$g({\xi}{\downarrow},a^*,\mathbf{C}(\xi),t)=\max_{a\in\mathrm{En}({\xi}{\downarrow})}g({\xi}{\downarrow},a,\mathbf{C}(\xi),t)$ and (ii) 
there exists $s\in L$ such that $\mathbf{P}({\xi}{\downarrow},a^*,s)>0$ and the distance from $s$ to $L_1\backslash L^{\xi,t}_2$ is 
(one-step) smaller than that from ${\xi}{\downarrow}$ in the digraph 
\begin{align*}
& \mathcal{G}^{\xi,t}:=\Big(L,\Big\{(u,v)\in L\times L\mid \\
& ~\exists b\in\mathrm{En}(u).\big[\mathbf{P}(u,b,v)>0\wedge g(u,b,\mathbf{C}(\xi),t)=\max_{a\in\mathrm{En}(u)}g(u,a,\mathbf{C}(\xi),t)\big]\Big\}\Big)\enskip. 
\end{align*}
If there are multiple such $a^*$'s, choose the least of them w.r.t $\preceq$.
\end{enumerate}
The legitimacy of the third step in the procedure above follows from Proposition~\ref{prop:approx:late:weightzero} to be proved later:
the set $L''$ of states that cannot reach $L_1\backslash L_2^{\xi,t}$ should be empty, or otherwise one can reduce all values 
$\{\mathrm{prob}^{\mathrm{l},\max}_G(s,\mathbf{c}-\mathbf{C}(\xi)-t\cdot\mathbf{w}(s))\}_{s\in L''}$ by a small amount to obtain 
a pre-fixed-point smaller than 
$\{\mathrm{prob}^{\mathrm{l},\max}_G(s,\mathbf{c}-\mathbf{C}(\xi)-t\cdot\mathbf{w}(s))\}_{s\in L}$.
By definition, $\mathtt{D}_\mathbf{c}^\mathrm{l}$ is deterministic cost-positional. Note that there are finitely many triples 
$(L_1,L_2^{\xi,t},\mathcal{G}^{\xi,t})$ since $L$ is finite. Also, by Theorem~\ref{thm:minlatef}, 
$(\mathbf{x},t)\mapsto g(s,a,\mathbf{x},t)$ is separately continuous on 
\[
\{(\mathbf{x},t)\mid \mathbf{x}+t\cdot\mathbf{w}(s)\le\mathbf{c}\}
\]
and its complement set, for all $s\in L$ and $a\in\mathrm{En}(s)$. 
Thus, the set of all pairs $(\xi,t)$, where  
$\xi$ is a history of length $n$ and $t$ is a non-negative real number $t$, such that the triple $(L_1,L_2^{\xi,t},\mathcal{G}^{\xi,t})$ happens to be a specific one is measurable w.r.t 
$(\Omega_\mathcal{M}^n\times\mathbb{R}_{\ge 0},\mathcal{S}^n_\mathcal{M}\otimes \mathcal{B}(\mathbb{R}_{\ge 0}))$. It follows that $\mathtt{D}_\mathbf{c}^\mathrm{l}$ is a measurable scheduler. We prove 
that $\mathrm{prob}^{\mathrm{l},\max}_G(s,\mathbf{c})=\mathrm{Pr}_{\mathtt{D}_\mathbf{c}^\mathrm{l},\mathcal{D}[s]}(\Pi_G^\mathbf{c})$
for all $s\in L$ and $\mathbf{c}\ge\vec{0}$. 

Define the function $h:L\times\mathbb{R}^k_{\ge 0}\rightarrow [0,1]$ by 
$h(s, \mathbf{c}):=\mathrm{Pr}_{\mathtt{D}_\mathbf{c}^\mathrm{l},\mathcal{D}[s]}(\Pi_G^\mathbf{c})$. 
Suppose that $\mathrm{prob}^{\mathrm{l},\max}_G\ne h$. Let $\overline{h}:=|\mathrm{prob}^{\mathrm{l},\max}_G-h|$. 
Then there exists $s\in L$ and $\mathbf{c}\ge\vec{0}$ such that $\overline{h}(s,\mathbf{c})>0$. Define 
$T:=\max_{1\le i\le k}\{\frac{\mathbf{c}_i}{\mathbf{w}_{\min}}\}$\enskip. Let
\[
d:=\sup\left\{\overline{h}(s,\mathbf{c}')\mid s\in L,\vec{0}\le\mathbf{c}'\le\mathbf{c}\right\}
\]
and
\[
d':=\sup\left\{\overline{h}(s,\mathbf{c}')\mid s\in L,\vec{0}\le\mathbf{c}'\le\mathbf{c}\mbox{ and }\mathbf{w}(s)\ne\vec{0}\right\}\enskip.
\]
We first show that $d'<d$ by a (nested) contradiction proof.

Suppose $d'=d$. Choose $\epsilon>0$ such that $d>e^{\mathbf{E}_{\max}\cdot T}\cdot\epsilon$. Then choose $s\in L$ and 
$\vec{0}\le\mathbf{c}'\le\mathbf{c}$ such that $d-\epsilon<|\overline{h}(s,\mathbf{c}')|\le d$ and $\mathbf{w}(s)\ne\vec{0}$. 
By Theorem~\ref{thm:fix-late}, we have
\[
h(s,\mathbf{c}')=\int_0^\infty f_{\mathbf{E}(s)}(t)\cdot \left[\sum_{s'\in L}\mathbf{P}(s,a^*(t),s')\cdot h(s',\mathbf{c}'-t\cdot\mathbf{w}(s))\right]\,\mathrm{d}t\enskip.
\]
with $\mathcal{D}[a^*(t)]=\mathtt{D}_\mathbf{c}^\mathrm{l}(s,t,\centerdot)$.
Then with Theorem~\ref{thm:minlatef}, we obtain
\[
\overline{h}(s,\mathbf{c}')\le\int_0^T f_{\mathbf{E}(s)}(t)\cdot\left[\sum_{s'\in L}\mathbf{P}(s,a^*(t),s')\cdot
\overline{h}(s',\mathbf{c}'-t\cdot\mathbf{w}(s))\right]\,\mathrm{d}t\enskip.
\]
This implies $d-\epsilon\le (1-e^{-\mathbf{E}_{\max}\cdot T})\cdot d$, which in turn implies 
$d\le e^{\mathbf{E}_{\max}\cdot T}\cdot\epsilon$. Contradiction to the choice of $\epsilon$.

Thus $d>d'$. Let $\delta:=d-d'$ and $\epsilon:=\mathbf{P}_{\min}^{|L|}\cdot\delta$. We inductively construct a finite sequence $s_0,s_1,\dots,s_l$ ($1\le l\le |L|$) which satisfies
\[
\overline{h}(s_i,\mathbf{c}')> d-\mathbf{P}^{-i}_{\min}\cdot\epsilon~(i=0,\dots,l)
\]
as follows. Note that the triple $(L_1,L_2^{s,t},\mathcal{G}^{s,t})$ for $s\in L\backslash L_1$ and $t\in\mathbf{R}_{\ge 0}$ (w.r.t some $\mathbf{c}\in\mathbb{R}^k$) remains constant as $s,t$ varies, since $\mathbf{C}(s)=\mathbf{w}(s)=\vec{0}$. 
\begin{enumerate}\itemsep1pt \parskip0pt \parsep0pt
\item Initially, we set $i=0$ and choose $s_0\in L\backslash L_1,\vec{0}\le\mathbf{c}'\le\mathbf{c}$ such that
\[
d-\epsilon<\overline{h}(s_0,\mathbf{c}')\le d\enskip.
\]
\item As long as $i\le l$ and $0\le d-\mathbf{P}_{\min}^{-i}\cdot\epsilon<\overline{h}(s_i,\mathbf{c}')\le d$\enskip, we have 
$s_i\in L\backslash L_1$ (w.r.t $\mathbf{c}'$) since $\mathbf{P}_{\min}^{-i}\cdot\epsilon\le\delta$. By $\mathbf{w}(s_i)=\vec{0}$, $\mathtt{D}_{\mathbf{c}'}^{\mathrm{l}}(s_i,t,\centerdot)$
remains constant when $t$ varies. It follows that $s_i\in L\backslash (L_1\cup L_2^{\xi,0})$. Let $a^*$ be the action such that
$\mathcal{D}[a^*]=\mathtt{D}_{\mathbf{c}'}^{\mathrm{l}}(s_i,0,\centerdot)$. Then 
\[
h(s_i,\mathbf{c}')=\sum_{s'\in L}\mathbf{P}(s_i,a^*,s')\cdot h(s',\mathbf{c}')
\]
and there exists $u\in L$ such that $\mathbf{P}(s_i,a^*,u)>0$ and via $u$ the distance to $L_1\backslash L^{s_i,0}_2$ in $\mathcal{G}^{s_i,0}$ (w.r.t $\mathbf{c}'$) is decreased by one. Moreover, by
\[
\overline{h}(s_i,\mathbf{c}')\le\sum_{s'\in L}\mathbf{P}(s_i,a^*,s')\cdot\overline{h}(s',\mathbf{c}')
\]
we obtain
\[
d-\mathbf{P}^{-i}_{\min}\epsilon<(1-{\mathbf{P}}_{\min})\cdot d+{\mathbf{P}}_{\min}\cdot\overline{h}(u,\mathbf{c}')\enskip,
\]
which is further reduced to $\overline{h}(u,\mathbf{c}')> d-\mathbf{P}^{-(i+1)}_{\min}\cdot\epsilon$\enskip. We set $s_{i+1}=u$;
\item If $s_{i+1}\in L_1\backslash L^{s_{i+1},0}_2$, then the construction is terminated. Otherwise, back to Step 2.
\end{enumerate}
The legitimacy and termination (within $|L|$ steps) of the inductive construction follows from the definition of 
$\mathtt{D}_{\mathbf{c}'}^\mathrm{l}$. By $s_l\in L_1\backslash L^{s_l,0}_2$, we obtain
\[
d-\delta\ge \overline{h}(s_l,\mathbf{c}')> d-\mathbf{P}^{-l}_{\min}\cdot\epsilon~~(l\le |L|)\enskip,
\]
which is a contradiction due to $\epsilon=\mathbf{P}^{|L|}_{\min}\cdot\delta$. Thus $\mathrm{prob}^{\mathrm{l},\max}_G= h$.
\end{proof}

\section{Differential Characterizations for \\Maximal Reachability Probabilities}

In this section, we derive differential characterizations for the functions $\mathrm{prob}^{\mathrm{e},\max}_G$ and $\mathrm{prob}^{\mathrm{l},\max}_G$. These differential characterizations will be fundamental to our approximation algorithms. Below we fix a CTMDP $\left(L,Act,\mathbf{R},\{\mathbf{w}_i\}_{1\le i\le k}\right)$ and a set $G\subseteq L$.

\subsection{Early Schedulers}

Below we derive the differential characterization for $\mathrm{prob}^{\mathrm{e},\max}_G$. To ease the notation, we abbreviate 
$\mathrm{prob}^{\mathrm{e},\max}_G$ as $\mathrm{prob}^{\max}_G$ in this part. 

To derive the differential characterization for early schedulers, we first extend $\mathrm{prob}^{\max}_G$ in the following way.

\begin{definition}\label{def:approx:early:aux}
Let $Z_G:=\{(s,a)\in (L-G)\times Act\mid a\in\mathrm{En}(s)\}$\enskip. 
Define $\mathrm{prob}^{\max}_G:Z_G\times\mathbb{R}^k\rightarrow [0,1]$ by
\[
\mathrm{prob}^{\max}_G((s,a),\mathbf{c}):=\int_{0}^\infty f_{\mathbf{E}(s,a)}(t)\cdot\left[\sum_{s'\in L}\mathbf{P}(s,a,s')\cdot \mathrm{prob}^{\max}_G(s',\mathbf{c}-t\cdot\mathbf{w}(s,a))\right]\,\mathrm{d}t
\]
for $(s,a)\in Z_G$ and $\mathbf{c}\in\mathbb{R}^k$.
\end{definition}
By Definition~\ref{def:approx:early:aux} and Theorem~\ref{thm:minearlyf}, one easily sees that 
\[
\mathrm{prob}^{\max}_G(s,\mathbf{c})=\max_{a\in\mathrm{En}(s)}\mathrm{prob}^{\max}_G((s,a),\mathbf{c})\enskip
\]
for all $s\in L-G$ and $\mathbf{c}\in\mathbb{R}^k$.

The following definition introduces a sort of directional derivative which will be crucial in our approximation algorithm. 

\begin{definition}\label{def:approx:early:diderivative}
Let $z\in Z_G$ and $\mathbf{c}\ge\vec{0}$. Define
\[
\nabla^+\mathrm{prob}^{\max}_G(z,\mathbf{c}):=\lim\limits_{t\rightarrow 0^+}\frac{\mathrm{prob}^{\max}_G(z,\mathbf{c}+t\cdot\mathbf{w}(z))-\mathrm{prob}^{\max}_G(z,\mathbf{c})}{t}\enskip.
\]
If $\mathbf{c}_i>0$ whenever $\mathbf{w}_i(z)>0$ ($1\le i\le k$), define
\[
\nabla^-\mathrm{prob}^{\max}_G(z,\mathbf{c}):=\lim\limits_{t\rightarrow 0^-}\frac{\mathrm{prob}^{\max}_G(z,\mathbf{c}+t\cdot \mathbf{w}(z))-\mathrm{prob}^{\max}_G(z,\mathbf{c})}{t}\enskip;
\]
Otherwise, let $\nabla^-\mathrm{prob}^{\max}_G(z,\mathbf{c}):=\nabla^+\mathrm{prob}^{\max}_G(z,\mathbf{c})$\enskip.
\end{definition}
Thus $\nabla^+\mathrm{prob}^{\max}_G(z,\mathbf{c})$ (resp. $\nabla^-\mathrm{prob}^{\max}_G(z,\mathbf{c})$) is the right (resp. left) 
directional derivative along the vector $\mathbf{w}(z)$. 
The following theorem gives a characterization for $\nabla^+\mathrm{prob}^{\max}_G((s,a),\mathbf{c})$ and 
$\nabla^-\mathrm{prob}^{\max}_G((s,a),\mathbf{c})$\enskip.

\begin{theorem}\label{thm:approx:early:diderivative}
For all $z\in Z_G$ and $\mathbf{c}\ge \vec{0}$, 
$\nabla^+\mathrm{prob}^{\max}_G(z,\mathbf{c})=\nabla^-\mathrm{prob}^{\max}_G(z,\mathbf{c})$. Moreover, 
\[
\nabla^+\mathrm{prob}^{\max}_G((s,a),\mathbf{c})=\sum_{s'\in L}\mathbf{R}(s,a,s')\cdot\left(\mathrm{prob}^{\max}_{G}(s',\mathbf{c})-\mathrm{prob}^{\max}_G((s,a),\mathbf{c})\right)
\]
for all $(s,a)\in Z_G$ and $\mathbf{c}\ge\vec{0}$\enskip.
\end{theorem}
\begin{proof}
Let $z=(s,a)$ and $\lambda:=\mathbf{E}(s,a)$. We first consider $\nabla^+\mathrm{prob}^{\max}_G$.
By Theorem~\ref{thm:minearlyf}, for $t>0$, we have
\begin{eqnarray*}
& & \mathrm{prob}^{\max}_G((s,a),\mathbf{c}+t\cdot\mathbf{w}(s,a))\\
&=& \int_{0}^\infty f_{\lambda}(\tau)\cdot\left[\sum_{s'\in L}\mathbf{P}(s,a,s')\cdot\mathrm{prob}^{\max}_{G}(s',\mathbf{c}+(t-\tau)\cdot\mathbf{w}(s,a))\right]\,\mathrm{d}\tau \\
&=& \int_{0}^t f_{\lambda}(\tau)\cdot\left[\sum_{s'\in L}\mathbf{P}(s,a,s')\cdot\mathrm{prob}^{\max}_{G}(s',\mathbf{c}+(t-\tau)\cdot\mathbf{w}(s,a))\right]\,\mathrm{d}\tau \\
& & + \int_{t}^\infty f_{\lambda}(\tau)\cdot\left[\sum_{s'\in L}\mathbf{P}(s,a,s')\cdot\mathrm{prob}^{\max}_{G}(s',\mathbf{c}+(t-\tau)\cdot\mathbf{w}(s,a))\right]\,\mathrm{d}\tau \\
&=& e^{-\lambda\cdot t}\cdot \int_{0}^t \lambda\cdot e^{\lambda\cdot\tau}\cdot\left[\sum_{s'\in L}\mathbf{P}(s,a,s')\cdot\mathrm{prob}^{\max}_{G}(s',\mathbf{c}+\tau\cdot\mathbf{w}(s,a))\right]\,\mathrm{d}\tau \\
& & { } + e^{-\lambda\cdot t}\cdot\int_{0}^\infty f_{\lambda}(\tau)\cdot\left[\sum_{s'\in L}\mathbf{P}(s,a,s')\cdot\mathrm{prob}^{\max}_{G}(s',\mathbf{c}-\tau\cdot\mathbf{w}(s,a))\right]\,\mathrm{d}\tau \\
&=& e^{-\lambda\cdot t}\cdot \int_{0}^t \lambda\cdot e^{\lambda\cdot \tau}\cdot\left[\sum_{s'\in L}\mathbf{P}(s,a,s')\cdot\mathrm{prob}^{\max}_{G}(s',\mathbf{c}+\tau\cdot\mathbf{w}(s,a))\right]\,\mathrm{d}\tau \\
& & { } + e^{-\lambda\cdot t}\cdot\mathrm{prob}^{\max}_G((s,a),\mathbf{c})\enskip\\
\end{eqnarray*}
where the third equality is obtained by the variable substitution $\tau'=t-\tau$ for the first integral, and $\tau'=\tau-t$ in the second 
integral. The legitimacy of the variable substitution follows from the fact that the integrand is piecewise continuous 
(cf. Theorem~\ref{thm:minearlyf}). Thus by the continuity of $\mathrm{prob}^{\max}_G$ (Theorem~\ref{thm:minearlyf}) and an application of 
L'Hospital's Rule to Definition~\ref{def:approx:early:diderivative}, we obtain
\[
\nabla^+\mathrm{prob}^{\max}_G((s,a),\mathbf{c})=
\lambda\cdot\left[\sum_{s'\in L}\mathbf{P}(s,a,s')\cdot\mathrm{prob}^{\max}_{G}(s',\mathbf{c})\right]-\lambda\cdot\mathrm{prob}^{\max}_G((s,a),\mathbf{c})\enskip.
\]
Then the result follows. 

The proof for $\nabla^+\mathrm{prob}^{\max}_G((s,a),\mathbf{c})=\nabla^-\mathrm{prob}^{\max}_G((s,a),\mathbf{c})$ follows a similar 
argument. By Theorem~\ref{thm:minearlyf} and the previous derivation, for adequate $t>0$, we have
\begin{eqnarray*}
& & \mathrm{prob}^{\max}_G((s,a),\mathbf{c})\\
&=& \mathrm{prob}^{\max}_G((s,a),(\mathbf{c}-t\cdot\mathbf{w}(s,a))+t\cdot\mathbf{w}(s,a))\\
&=& e^{-\lambda\cdot t}\cdot \int_{0}^t \lambda\cdot e^{\lambda\cdot \tau}\cdot\left[\sum_{s'\in L}\mathbf{P}(s,a,s')\cdot\mathrm{prob}^{\max}_{G}(s',\mathbf{c}-(t-\tau)\cdot\mathbf{w}(s,a))\right]\,\mathrm{d}\tau \\
& & { } + e^{-\lambda\cdot t}\cdot\mathrm{prob}^{\max}_G((s,a),\mathbf{c}-t\cdot\mathbf{w}(s,a))\\
&=& \int_{0}^t \lambda\cdot e^{-\lambda\cdot \tau}\cdot\left[\sum_{s'\in L}\mathbf{P}(s,a,s')\cdot\mathrm{prob}^{\max}_{G}(s',\mathbf{c}-\tau\cdot\mathbf{w}(s,a))\right]\,\mathrm{d}\tau \\
& & { } + e^{-\lambda\cdot t}\cdot\mathrm{prob}^{\max}_G((s,a),\mathbf{c}-t\cdot\mathbf{w}(s,a))\enskip.\\
\end{eqnarray*}
where the last equality is obtained through the variable substitution $\tau'=t-\tau$. Thus by continuity and  
L'H{\^o}spital's rule, we obtain 
\[
\nabla^-\mathrm{prob}^{\max}_G((s,a),\mathbf{c})=
\lambda\cdot\left[\sum_{s'\in L}\mathbf{P}(s,a,s')\cdot\mathrm{prob}^{\max}_{G}(s',\mathbf{c})\right]-\lambda\cdot\mathrm{prob}^{\max}_G((s,a),\mathbf{c})\enskip.
\]
It follows that $\nabla^+\mathrm{prob}^{\max}_G((s,a),\mathbf{c})=\nabla^-\mathrm{prob}^{\max}_G((s,a),\mathbf{c})$\enskip.
\end{proof}
Theorem~\ref{thm:approx:early:diderivative} gives a differential characterization for $\mathrm{prob}^{\max}_G(z,\centerdot)$, 
$z\in Z_G$. Since $\nabla^+\mathrm{prob}^{\max}_G(z,\mathbf{c})=\nabla^-\mathrm{prob}^{\max}_G(z,\mathbf{c})$, we will solely use 
$\nabla\mathrm{prob}^{\max}_G(z,\mathbf{c})$ to denote both of them. 

Theorem~\ref{thm:approx:early:diderivative} allows one to approximate $\mathrm{prob}^{\max}_G(z,\mathbf{c}+t\cdot\mathbf{w}(z))$ by 
$\mathrm{prob}^{\max}_G(z,\mathbf{c})$ and $\nabla\mathrm{prob}^{\max}_G(z,\mathbf{c})$. This suggests an approximation algorithm which 
approximates $\mathrm{prob}^{\max}_G(z,\mathbf{c})$ from $\{\mathrm{prob}^{\max}_G(z,\mathbf{c}')\mid\mathbf{c}'\lneq\mathbf{c}\}$\enskip. 
An exception is the case when $\mathbf{w}(z)=\vec{0}$. Below we tackle this situation. 

\begin{proposition}\label{prop:approx:early:weightzero}
Let 
\[
Y^{\mathrm{e}}_G:=\{z\in Z_G\mid \mathbf{w}(z)=\vec{0}\}\cup\{s\in L-G\mid \exists a\in\mathrm{En}(s).\mathbf{w}(s,a)=\vec{0}\}\enskip.
\]
For all $\mathbf{c}\ge \vec{0}$, the function $y\mapsto\mathrm{prob}^{\max}_G(y,\mathbf{c})$ is the least fixed-point of the high order
operator $\mathcal{Y}^\mathrm{e}_{\mathbf{c}, G}:\left[Y^\mathrm{e}_G\rightarrow [0,1]\right]\rightarrow\left[Y^\mathrm{e}_G\rightarrow [0,1]\right]$ defined as follows:
\begin{align*}
&\mathcal{Y}^\mathrm{e}_{\mathbf{c}, G}(h)(s,a):=\\
&\quad\sum_{s'\in Y^\mathrm{e}_G}\mathbf{P}(s,a,s')\cdot h(s')+\sum_{s'\in L-Y^\mathrm{e}_G}\mathbf{P}(s,a,s')\cdot\mathrm{prob}_G^{\max}(s',\mathbf{c})
\end{align*}
for $(s,a)\in Y^\mathrm{e}_G$, and 
\begin{align*}
&\mathcal{Y}^\mathrm{e}_{\mathbf{c}, G}(h)(s):=\\
&\quad\max\{\max\{h(s,a)\mid (s,a)\in Y^\mathrm{e}_G\},\max\{\mathrm{prob}^{\max}_G((s,a),\mathbf{c})\mid (s,a)\in Z_G-Y^\mathrm{e}_G\}\}
\end{align*}
for $s\in Y^\mathrm{e}_G$\enskip.
\end{proposition}
\begin{proof}
By Theorem~\ref{thm:minearlyf}, one can easily see that $\{\mathrm{prob}^{\max}_G(y,\mathbf{c})\}_{y\in Y^\mathrm{e}_G}$ is a fixed-point 
of $\mathcal{Y}^\mathrm{e}_\mathbf{c}$. Suppose that it is not the least fixed-point of $\mathcal{Y}^\mathrm{e}_{\mathbf{c}, G}$. Let the 
least fixed point be $h^*$. Define 
\begin{itemize}\itemsep1pt \parskip0pt \parsep0pt
\item $\delta:=\max\{\mathrm{prob}^{\max}_G(y,\mathbf{c})-h^*(y)\mid y\in Y^\mathrm{e}_G\}$, and 
\item $Y':=\{y\in Y^\mathrm{e}_G\mid \mathrm{prob}^{\max}_G(y,\mathbf{c})-h^*(y)=\delta\}$\enskip.
\end{itemize}
Since $\{\mathrm{prob}^{\max}_G(y,\mathbf{c})\}_{y\in Y^\mathrm{e}_G}\ne h^*$, we have $\delta>0$. Consider an arbitrary $y\in Y'$. By the 
maximal choice of $\delta$ and $Y'$, one can obtain that
\begin{enumerate}\itemsep1pt\parskip0pt\parsep0pt
\item for all $(s,a)\in Y'$, $s'\in Y'$ whenever $s'\in L$ and $\mathbf{P}(s,a,s')>0$;
\item for all $s\in Y'$, $(s,a)\in Y'$ whenever $a\in\mathrm{En}(s)$ and $\mathrm{prob}^{\max}_G(s,\mathbf{c})=\mathrm{prob}^{\max}_G((s,a),\mathbf{c})$\enskip.
\end{enumerate}
Intuitively, one can decrease every coordinate in $Y'$ on $\mathrm{prob}^{\max}_G$ by a same amount so that a certain ``balance'' still holds. 
Choose $\delta'\in (0,\delta)$ such that 
\[
\mathrm{prob}^{\max}_G(s,\mathbf{c})-\delta'\ge\max\{\mathrm{prob}_{G}^{\max}(s,a,\mathbf{c})-\delta'\cdot\mathbf{1}_{Y'}(s,a)\mid a\in\mathrm{En}(s)\}
\]
for all $s\in Y'$. Define $h: L\times\mathbb{R}^k\rightarrow [0,1]$ by: 
$h(s,\mathbf{c}')=\mathrm{prob}^{\max}_G(s,\mathbf{c}')-\delta'$ if $\mathbf{c}'=\mathbf{c}$ and $s\in Y'$, and 
$h(s,\mathbf{c}')=\mathrm{prob}^{\max}_G(s,\mathbf{c}')$ otherwise. Then $h$ is a pre-fixed-point of $\mathcal{T}^\mathrm{e}_G$ which 
satisfies that $h\lneq\mathrm{prob}^{\max}_G$. Contradiction to Theorem~\ref{thm:minearlyf}.
\end{proof}

\subsection{Late Schedulers}

Below we derive the differential characterization for $\mathrm{prob}^{\mathrm{l},\max}_G$. The development is much similar to the one of 
early schedulers. To ease the notation, we abbreviate $\mathrm{prob}^{\mathrm{l},\max}_G$ as $\mathrm{prob}^{\max}_G$ in this part. 

\begin{definition}\label{def:approx:late:diderivative}
Let $s\in L-G$ and $\mathbf{c}\ge\vec{0}$. Define
\[
\nabla^+\mathrm{prob}^{\max}_G(s,\mathbf{c}):=\lim\limits_{t\rightarrow 0^+}\frac{\mathrm{prob}^{\max}_G(s,\mathbf{c}+t\cdot\mathbf{w}(s))-\mathrm{prob}^{\max}_G(s,\mathbf{c})}{t}\enskip.
\]
If $\mathbf{c}_i>0$ whenever $\mathbf{w}_i(s)>0$ ($1\le i\le k$), define
\[
\nabla^-\mathrm{prob}^{\max}_G(s,\mathbf{c}):=\lim\limits_{t\rightarrow 0^-}\frac{\mathrm{prob}^{\max}_G(s,\mathbf{c}+t\cdot \mathbf{w}(s))-\mathrm{prob}^{\max}_G(s,\mathbf{c})}{t}\enskip;
\]
Otherwise, let $\nabla^-\mathrm{prob}^{\max}_G(s,\mathbf{c})=\nabla^+\mathrm{prob}^{\max}_G(s,\mathbf{c})$\enskip.
\end{definition}

\begin{theorem}\label{thm:approx:late:diderivative}
For all $s\in L-G$ and $\mathbf{c}\ge\vec{0}$, $\nabla^+\mathrm{prob}^{\max}_G(s,\mathbf{c})=\nabla^-\mathrm{prob}^{\max}_G(s,\mathbf{c})$. Moreover, 
\[
\nabla^+\mathrm{prob}^{\max}_G(s,\mathbf{c})=\max_{a\in\mathrm{En}(s)}
\sum_{s'\in L}\mathbf{R}(s,a,s')\cdot\left(\mathrm{prob}^{\max}_{G}(s',\mathbf{c})-\mathrm{prob}^{\max}_G(s,\mathbf{c})\right)\enskip.
\]
\end{theorem}
\begin{proof}
Let $\lambda:=\mathbf{E}(s)$. We first consider $\nabla^+\mathrm{prob}^{\max}_G$.
By Theorem~\ref{thm:minlatef}, for $t>0$, we have

\begin{eqnarray*}
& & \mathrm{prob}^{\max}_G(s,\mathbf{c}+t\cdot\mathbf{w}(s))\\
&=& \int_{0}^\infty f_{\lambda}(\tau)\cdot\max_{a\in\mathrm{En}(s)}\left[\sum_{s'\in L}\mathbf{P}(s,a,s')\cdot\mathrm{prob}^{\max}_{G}(s',\mathbf{c}+(t-\tau)\cdot\mathbf{w}(s))\right]\,\mathrm{d}\tau \\
&=& \int_{0}^t f_{\lambda}(\tau)\cdot\max_{a\in\mathrm{En}(s)}\left[\sum_{s'\in L}\mathbf{P}(s,a,s')\cdot\mathrm{prob}^{\max}_{G}(s',\mathbf{c}+(t-\tau)\cdot\mathbf{w}(s))\right]\,\mathrm{d}\tau \\
& & + \int_{t}^\infty f_{\lambda}(\tau)\cdot\max_{a\in\mathrm{En}(s)}\left[\sum_{s'\in L}\mathbf{P}(s,a,s')\cdot\mathrm{prob}^{\max}_{G}(s',\mathbf{c}+(t-\tau)\cdot\mathbf{w}(s))\right]\,\mathrm{d}\tau \\
&=& e^{-\lambda\cdot t}\cdot \int_{0}^t \lambda\cdot e^{\lambda\cdot\tau}\cdot\max_{a\in\mathrm{En}(s)}\left[\sum_{s'\in L}\mathbf{P}(s,a,s')\cdot\mathrm{prob}^{\max}_{G}(s',\mathbf{c}+\tau\cdot\mathbf{w}(s))\right]\,\mathrm{d}\tau \\
& & { }+e^{-\lambda\cdot t}\int_{0}^\infty f_{\lambda}(\tau)\cdot\max_{a\in\mathrm{En}(s)}\left[\sum_{s'\in L}\mathbf{P}(s,a,s')\cdot\mathrm{prob}^{\max}_{G}(s',\mathbf{c}-\tau\cdot\mathbf{w}(s))\right]\,\mathrm{d}\tau \\
&=& e^{-\lambda\cdot t}\cdot \int_{0}^t \lambda\cdot e^{\lambda\cdot t}\cdot\max_{a\in\mathrm{En}(s)}\left[\sum_{s'\in L}\mathbf{P}(s,a,s')\cdot\mathrm{prob}^{\max}_{G}(s',\mathbf{c}+\tau\cdot\mathbf{w}(s))\right]\,\mathrm{d}\tau \\
& & {}+e^{-\lambda\cdot t}\cdot\mathrm{prob}^{\max}_G(s,\mathbf{c})\\
\end{eqnarray*}
where the third equality is obtained by the variable substitution $\tau'=t-\tau$ for the first integral, and $\tau'=\tau-t$ in the second 
integral. The legitimacy of the variable substitution follows from the fact that the integrand is piecewise continuous 
(cf. Theorem~\ref{thm:minlatef}). Thus by applying L'Hospital's rule to Definition~\ref{def:approx:late:diderivative}, we obtain that
\[
\nabla^+\mathrm{prob}^{\max}_G(s,\mathbf{c})=
\max_{a\in\mathrm{En}(s)}\lambda\cdot\left[\sum_{s'\in L}\mathbf{P}(s,a,s')\cdot\left(\mathrm{prob}^{\max}_{G}(s',\mathbf{c})-\mathrm{prob}^{\max}_G(s,\mathbf{c})\right)\right]
\]
which implies the result. The proof for 
$\nabla^-\mathrm{prob}^{\max}_G((s,a),\mathbf{c})$ follows a similar argument. By 
Theorem~\ref{thm:minlatef} and the previous derivation, we have for adequate $t>0$, 
\begin{eqnarray*}
& & \mathrm{prob}^{\max}_G(s,\mathbf{c})\\
&=& \mathrm{prob}^{\max}_G(s,(\mathbf{c}-t\cdot\mathbf{w}(s))+t\cdot\mathbf{w}(s))\\
&=& e^{-\lambda\cdot t}\cdot \int_{0}^t \lambda\cdot e^{\lambda\cdot \tau}\cdot\max_{a\in\mathrm{En}(s)}\left[\sum_{s'\in L}\mathbf{P}(s,a,s')\cdot\mathrm{prob}^{\max}_{G}(s',\mathbf{c}-(t-\tau)\cdot\mathbf{w}(s))\right]\,\mathrm{d}\tau \\
& & { } + e^{-\lambda\cdot t}\cdot\mathrm{prob}^{\max}_G(s,\mathbf{c}-t\cdot\mathbf{w}(s))\\
&=& \int_{0}^t \lambda\cdot e^{-\lambda\cdot \tau}\cdot\max_{a\in\mathrm{En}(s)}\left[\sum_{s'\in L}\mathbf{P}(s,a,s')\cdot\mathrm{prob}^{\max}_{G}(s',\mathbf{c}-\tau\cdot\mathbf{w}(s))\right]\,\mathrm{d}\tau \\
& & { } + e^{-\lambda\cdot t}\cdot\mathrm{prob}^{\max}_G(s,\mathbf{c}-t\cdot\mathbf{w}(s))\enskip.\\
\end{eqnarray*}
where the last equality is obtained through the variable substitution $\tau'=t-\tau$. Thus by applying L'Hospital's Rule to 
Definition~\ref{def:approx:late:diderivative}, we obtain 
\[
\nabla^-\mathrm{prob}^{\max}_G(s,\mathbf{c})=
\max_{a\in\mathrm{En}(s)}\lambda\cdot\left[\sum_{s'\in L}\mathbf{P}(s,a,s')\cdot\left(\mathrm{prob}^{\max}_{G}(s',\mathbf{c})-\mathrm{prob}^{\max}_G(s,\mathbf{c})\right)\right]
\]
which directly shows that $\nabla^+\mathrm{prob}^{\max}_G((s,a),\mathbf{c})=\nabla^-\mathrm{prob}^{\max}_G((s,a),\mathbf{c})$. 
\end{proof}
As in the case of early schedulers, there is a special case when $\mathbf{w}(s)=\vec{0}$. The following proposition is the counterpart of 
Proposition~\ref{prop:approx:early:weightzero}.

\begin{proposition}\label{prop:approx:late:weightzero}
Let $Y^\mathrm{l}_G:=\{s\in L\mid \mathbf{w}(s)=\vec{0}\}$\enskip. For each $\mathbf{c}\ge \vec{0}$, the function 
$s\mapsto\mathrm{prob}^{\max}_G(s,\mathbf{c})$ is the least fixed-point (w.r.t $\le$) of the high-order operator 
$\mathcal{Y}^\mathrm{l}_{\mathbf{c}, G}:\left[Y^{\mathrm{l}}_G\rightarrow [0,1]\right]\rightarrow\left[Y^{\mathrm{l}}_G\rightarrow [0,1]\right]$ defined as follows:
\[
\mathcal{Y}^\mathrm{l}_{\mathbf{c}, G}(h)(s):=\max_{a\in\mathrm{En}(s)}\left[\sum_{s'\in Y^\mathrm{l}_G}\mathbf{P}(s,a,s')\cdot h(s')+\sum_{s'\in L-Y^\mathrm{l}_G}\mathbf{P}(s,a,s')\cdot\mathrm{prob}_G^{\max}(s',\mathbf{c})\right]\enskip.
\]
\end{proposition}
\begin{proof}
By Theorem~\ref{thm:minlatef}, one can easily see that $\{\mathrm{prob}^{\max}_G(s,\mathbf{c})\}_{s\in Y^\mathrm{l}_G}$ is a fixed point 
of $\mathcal{Y}^\mathrm{l}_{\mathbf{c}, G}$. Suppose that it is not the least fixed-point of $\mathcal{Y}^\mathrm{l}_{\mathbf{c}, G}$. Let the least 
fixed-point be $h^*$. Define 
\begin{itemize}\itemsep1pt \parskip0pt \parsep0pt
\item $\delta:=\max\{\mathrm{prob}^{\max}_G(s,\mathbf{c})-h^*(s)\mid s\in Y^\mathrm{l}_G\}$, and 
\item $Y':=\{s\in Y^\mathrm{l}_G\mid \mathrm{prob}^{\max}_G(s,\mathbf{c})-h^*(s)=\delta\}$\enskip.
\end{itemize}
Since $s\mapsto\mathrm{prob}^{\max}_G(s,\mathbf{c})\ne h^*$, we have $\delta>0$. By the maximal choice of $\delta$ and $Y'$, one can
obtain that for all $s\in Y'$, $s'\in Y'$ whenever there exists an $a\in\mathrm{En}(s)$ such that 
$\mathrm{prob}_G^{\max}(s,\mathbf{c})=\sum_{s''\in L}\mathbf{P}(s,a,s'')\cdot \mathrm{prob}_G^{\max}(s'',\mathbf{c})$ and $\mathbf{P}(s,a,s')>0$\enskip.
Intuitively, one can decrease every coordinate in $Y'$ on $\mathrm{prob}^{\max}_G$ by a same amount so that the ``balance'' still holds. 
Choose $\delta'\in (0,\delta)$ such that 
\[
\mathrm{prob}^{\max}_G(s,\mathbf{c})-\delta'\ge\max_{a\in\mathrm{En}(s)}\left[\sum_{s'\in L}\mathbf{P}(s,a,s')\cdot (\mathrm{prob}^{\max}_G(s',\mathbf{c})-\delta'\cdot\mathbf{1}_{Y'}(s'))\right]
\]
for all $s\in Y'$. Define $h:L\times\mathbb{R}^k\rightarrow [0,1]$ by: 
$h(s,\mathbf{c}')=\mathrm{prob}^{\max}_G(s,\mathbf{c}')-\delta'$ if $\mathbf{c}'=\mathbf{c}$ and $s\in Y'$, and 
$h(s,\mathbf{c}')=\mathrm{prob}^{\max}_G(s,\mathbf{c}')$ otherwise. Then $h$ is a pre-fixed-point of $\mathcal{T}^\mathrm{l}_G$ which 
satisfies that $h\lneq\mathrm{prob}^{\max}_G$. Contradiction to Theorem~\ref{thm:minlatef}.
\end{proof}

\section{Numerical Approximation Algorithms}

In this section, we develop approximation algorithms to compute the maximal cost-bounded reachability probability under  
both early and late schedulers. In the following we fix a CTMDP $\mathcal{M}=\left(L,Act,\mathbf{R},\{\mathbf{w}_i\}_{1\le i\le k}\right)$. 
Our numerical algorithms will achieve the following tasks:

~\\
\noindent\textbf{Input:} a set $G\subseteq L$, a state $s\in L$, a vector $\mathbf{c}\in\mathbb{N}^k_{0}$ and an error bound 
$\epsilon>0$;

\noindent\textbf{Output:} a value $x\in [0,1]$ such that $|\mathrm{prob}^{\mathrm{o},\max}_G(s,\mathbf{c})-x|\le\epsilon$ for 
$\mathrm{o}\in\{\mathrm{e},\mathrm{l}\}$\enskip.

~\\
For computational purposes, we assume that each $\mathbf{w}_i(s,a)$ is an integer; rational numbers (and simultaneously the input
cost bound vector) can adjusted to integers by multiplying a common multiplier of the denominators, without changing the maximal 
probability value to be approximated.  

\subsection{Early Schedulers}

In this part, we demonstrate the approximation algorithm under early schedulers. We will omit `$\mathrm{e}$' on the superscript of 
$\mathrm{prob}_G^{\mathrm{e},\max}$, $Y_G^\mathrm{e}$ and $\mathcal{Y}_{c,G}^\mathrm{e}$, etc. Below we fix a set $G\subseteq L$. 

Based on Theorem~\ref{thm:approx:early:diderivative} and Proposition~\ref{prop:approx:early:weightzero}, we design our approximation
scheme as follows. First we introduce our discretization for a given $\mathbf{c}\in\mathbb{N}_{0}^k$ and a discretization step
$\frac{1}{N}$ ($N\in\mathbb{N}$). Note that $\mathrm{prob}^{\max}_G(s,\mathbf{c})=1$ whenever $s\in L$ and $\mathbf{c}\ge \vec{0}$.
Thus we do not need to incorporate those points into discretization. 

\begin{definition}\label{def:numerical:early:disc}
Let $\mathbf{c}\in\mathbb{N}_{0}^k$ and $N\in\mathbb{N}$. Define 
\[
\mathrm{Disc}(\mathbf{c},N):=\{\mathbf{d}\in\mathbb{R}^k\mid \vec{0}\le\mathbf{d}\le\mathbf{c}\mbox{ and }N\cdot\mathbf{d}_i\in\mathbb{N}_0\mbox{ for all }1\le i\le k\}\enskip.
\]
The set $\mathsf{D}^\mathbf{c}_N$ of \emph{discretized grid points} is defined as follows:
\[
\mathsf{D}^\mathbf{c}_N:=((L-G)\cup Z_G)\times \mathrm{Disc}(\mathbf{c},N)\enskip.
\]
\end{definition}
Thus $\mathsf{D}^\mathbf{c}_N$ is the set of ``grid points'' that are within the scope of $\mathbf{c}$ and are discretized with
discretization step $\frac{1}{N}$. The following definition presents the approximation scheme on $\mathsf{D}^\mathbf{c}_N$.
Intuitively, the approximation scheme describes how those ``points'' are related.

\begin{definition}
Define $X_G:=((L-G)\cup Z_G)-Y_G$. The \emph{approximation scheme} $\Upsilon^{G}_{\mathbf{c},N}$ on
$\mathsf{D}^\mathbf{c}_N$ consists of the following items:
\begin{itemize}\itemsep1pt\parskip0pt\parsep0pt
\item exactly one \emph{rounding argument} for each element of $\mathsf{D}^\mathbf{c}_N$\enskip;
\item a system of equations for elements in $X_G\times\mathrm{Disc}(\mathbf{c},N)$\enskip;
\item a linear program on $Y_G$ for each $\mathbf{d}\in\mathrm{Disc}(\mathbf{c},N)$ \enskip.
\end{itemize}
\noindent\textbf{Rounding Arguments:} For each element $y\in\mathsf{D}^\mathbf{c}_N$, the rounding argument for $y$ is as follows:
\[
\overline{\mathrm{prob}}_G(y):=\frac{K}{N^2}\mbox{ if }\mathrm{prob}(y)\in \left[\frac{K}{N^2},\frac{K+1}{N^2}\right)\mbox{ for some integer }0\le K\le N^2\enskip.
\]

~\\
\noindent\textbf{Equations:} The system of equations is described as follows. For all 
$((s,a),\mathbf{d})\in\mathsf{D}^\mathbf{c}_N$ with $\mathbf{w}(s,a)\ne\vec{0}$ and 
$\mathbf{d}-\frac{1}{N}\cdot\mathbf{w}(s,a)\ge \vec{0}$, there is a linear equation
\begin{align*}
& \frac{\mathrm{prob}_G((s,a),\mathbf{d})-\overline{\mathrm{prob}}_G((s,a),\mathrm{pre}(\mathbf{d},(s,a)))}{\frac{1}{N}}=\tag{E1}\\
& \quad\sum_{s'\in L}\mathbf{R}(s,a,s')\cdot\left(\overline{\mathrm{prob}}_G(s',\mathrm{pre}(\mathbf{d},(s,a)))-\overline{\mathrm{prob}}_G((s,a),\mathrm{pre}(\mathbf{d},(s,a)))\right)
\end{align*}
where $\mathrm{pre}(\mathbf{d},z):=\mathbf{d}-\frac{1}{N}\cdot\mathbf{w}(z)$. For all  
$((s,a),\mathbf{d})\in\mathsf{D}^\mathbf{c}_N$ with $\mathbf{w}(s,a)\ne\vec{0}$ and 
$\mathbf{d}-\frac{1}{N}\cdot\mathbf{w}(s,a)\not\ge \vec{0}$, there is a linear equation
\begin{equation*}
\mathrm{prob}_G((s,a),\mathbf{d})=0\enskip.\tag{E2}
\end{equation*}
For all $(s,\mathbf{d})\in\mathsf{D}^\mathbf{c}_N$ such that $\mathbf{w}(s,a)\ne\vec{0}$ for all $a\in\mathrm{En}(s)$, there is an equation
\begin{equation*}
\mathrm{prob}_G(s,\mathbf{d})=\max_{a\in\mathrm{En}(s)}\mathrm{prob}_G((s,a),\mathbf{d})\enskip.\tag{E3}
\end{equation*}

~\\
\noindent\textbf{Linear Programs:} For each $\mathbf{d}\in\mathrm{Disc}(\mathbf{c},N)$, the collection 
$\{\mathrm{prob}_G(y,\mathbf{d})\}_{y\in Y_G}$ of values is the unique optimum solution of the linear program as follows:
\begin{itemize}\itemsep1pt \parskip0pt \parsep0pt
\item[] $\min\sum_{y\in Y_G}\mathrm{prob}_G(y,\mathbf{d})$, subject to:
\item   $\mathrm{prob}_G((s,a),\mathbf{d})\ge\sum_{s'\in L}\mathbf{P}(s,a,s')\cdot\mathrm{prob}_G(s',\mathbf{d})$ for all $(s,a)\in Y_G$;
\item   $\mathrm{prob}_G((s,a),\mathbf{d})\le \mathrm{prob}_G(s,\mathbf{d})$ for all $(s,a)\in Y_G$;
\item   $\mathrm{prob}_G(y,\mathbf{d})\in [0,1]$ for all $y\in Y_G$;
\end{itemize}
where the values $\{\mathrm{prob}_G(y,\mathbf{d})\}_{y\in X_G}$ are assumed to be known. 

In all of the statements above, both $\mathrm{prob}_G(s,\mathbf{d})$ and $\overline{\mathrm{prob}}_G(s,\mathbf{d})$ represents $1$ for $s\in G$. 
\end{definition}
Generally, $\mathrm{prob}_G(y,\mathbf{d})$ approximates $\mathrm{prob}^{\max}_G(y,\mathbf{d})$ and 
$\overline{\mathrm{prob}}_G(y,\mathbf{d})$ approximates the same value with a rounding operation. A detailed computational sequence of the 
approximation scheme is described in Algorithm~\ref{algo:approx:early}.

\begin{algorithm}                      
\caption{The Computation of $\Upsilon_{\mathbf{c},N}^G$ (for early schedulers)}          
\label{algo:approx:early}                           
\begin{algorithmic}[1]                    
    \STATE Set all grid points in $\{((s,a),\mathbf{d})\in\mathsf{D}^\mathbf{c}_N\mid\mathbf{d}-\frac{1}{N}\cdot\mathbf{w}(s,a)\not\ge 0\}$ to zero by (E2);
    \STATE Compute all $\mathrm{prob}_G(s,\mathbf{d})$ that can be directly obtained through (E3);
    \STATE Compute all $\mathrm{prob}_G(y,\mathbf{d})$ that can be directly obtained through the linear program;
    \STATE Compute all $\overline{\mathrm{prob}}_G(y,\mathbf{d})$ that can be directly obtained by the rounding argument;
    \STATE Compute all $\mathrm{prob}_G((s,a),\mathbf{d})$ that can be directly obtained through (E1), and back to Step 2. until all grid points in $\mathsf{D}^\mathbf{c}_N$ are computed; 
\end{algorithmic}
\end{algorithm}
In principle, we compute the ``higher'' grid point $\mathrm{prob}_G((s,a),\mathbf{d}+\frac{1}{N}\cdot\mathbf{w}(s,a))$ by 
$\mathrm{prob}_G((s,a),\mathbf{d})$ and (E1), and then update other ``higher'' points through (E3) and the linear program. The rounding 
argument is incorporated to avoid precision explosion caused by linear programming. The following proposition shows that Algorithm~\ref{algo:approx:early} indeed terminates after a finite number of steps. 

\begin{proposition}\label{prop:numerical:early:termination}
Algorithm~\ref{algo:approx:early} terminates after a finite number of steps for all $\mathbf{c}\in\mathbb{N}^k_{0}$ and $N\in\mathbb{N}$.
\end{proposition}
\begin{proof}
Let $\mathbf{c}\ge\vec{0}$ and $N\in\mathbb{N}$. The proposition can be proved through a straightforward induction on 
$\sum_{i=1}^k\mathbf{d}_i$ that both $\mathrm{prob}_G(y,\mathbf{d})$ and $\overline{\mathrm{prob}}_G(y,\mathbf{d})$ can be computed after a finite number of steps for all 
$(y,\mathbf{d})\in\mathsf{D}_N^\mathbf{c}$. The base step where $\mathrm{prob}_G(y,\mathbf{d})$ is computed through (E2) is easy. For the 
inductive step, suppose that for all $(y,\mathbf{d}')\in\mathsf{D}_N^\mathbf{c}$ with 
$\sum_{i=1}^k\mathbf{d}'_i<\sum_{i=1}^k\mathbf{d}_i$, both $\mathrm{prob}_G(y,\mathbf{d})$ and $\overline{\mathrm{prob}}_G(y,\mathbf{d})$
can be computed after a finite number of steps by Algorithm~\ref{algo:approx:early}. Then the inductive step can be sequentially justified 
by (E1), (E3), the linear program and the rounding argument. 
\end{proof}

Below we prove that the approximation scheme really approximates $\mathrm{prob}^{\max}_G$. To ease the notation, we shall use 
$\mathrm{prob}_G(y,\mathbf{d})$ or $\overline{\mathrm{prob}}_G(y,\mathbf{d})$ to denote both the variable at the grid point and the value 
it holds under the approximation scheme. 

\begin{theorem}\label{thm:numerical:early:approxscheme}
Let $\mathbf{c}\in\mathbb{N}^k_{0}$ and $N\in\mathbb{N}$ with $N\ge\mathbf{E}_{\max}$. For each $(y,\mathbf{d})\in\mathsf{D}_N^\mathbf{c}$,
\[
\left|\mathrm{prob}_G(y,\mathbf{d})-\mathrm{prob}_G^{\max}(y,\mathbf{d})\right|\le \left(\frac{2\cdot\mathbf{E}^2_{\max}\cdot\mathbf{w}_{\max}}{N\cdot\mathbf{w}_{\min}}+\frac{1}{N}\right)\cdot\left[\sum_{i=1}^k\mathbf{d}_i\right]+\frac{\mathbf{E}_{\max}}{N}
\]
and 
\begin{align*}
& \left|\overline{\mathrm{prob}}_G(y,\mathbf{d})-\mathrm{prob}_G^{\max}(y,\mathbf{d})\right|\le \\
& \qquad \left(\frac{2\cdot\mathbf{E}^2_{\max}\cdot\mathbf{w}_{\max}}{N\cdot\mathbf{w}_{\min}}+\frac{1}{N}\right)\cdot\left[\sum_{i=1}^k\mathbf{d}_i\right]+\frac{\mathbf{E}_{\max}}{N}+\frac{1}{N^2}\enskip.
\end{align*}
\end{theorem}
\begin{proof}
We proceed by induction on the number of computation steps illustrated by Algorithm~\ref{algo:approx:early}.

~\\
\noindent\textbf{Base Step:} $(y,\mathbf{d})=((s,a),\mathbf{d})$ with $\mathbf{d}-\frac{1}{N}\cdot\mathbf{w}(s,a)\not\ge\vec{0}$\enskip. Then 
we know that $\mathrm{prob}^{\max}_G((s,a),\mathbf{d}-x\cdot\mathbf{w}(s,a))=0$ where $x\in [0, \frac{1}{N}]$ is the largest real number 
such that $\mathbf{d}-x\cdot\mathbf{w}(s,a)\ge\vec{0}$. Then by Theorem~\ref{thm:approx:early:diderivative} and Lagrange's Mean-Value Theorem, 
we obtain that 
\begin{align*}
& \mathrm{prob}^{\max}_G((s,a),\mathbf{d})-\mathrm{prob}^{\max}_G((s,a),\mathbf{d}-x\cdot\mathbf{w}(s,a))=\\
& \qquad x\cdot\nabla\mathrm{prob}_G^{\max}((s,a),\mathbf{d}-x'\cdot\mathbf{w}(s,a))
\end{align*}
for some $x'\in (0,x)$. It follows that $\mathrm{prob}^{\max}_G((s,a),\mathbf{d})\le\frac{1}{N}\cdot\mathbf{E}(s,a)$. Thus we have,
\[
\left|\mathrm{prob}^{\max}_G((s,a),\mathbf{d})-\mathrm{prob}_G((s,a),\mathbf{d})\right|\le\frac{1}{N}\cdot\mathbf{E}_{\max}\enskip.
\]
\noindent\textbf{Inductive Step.} The inductive step can be classified into the following  cases:

~\\
\noindent{Case 1:} $(y,\mathbf{d})=(s,\mathbf{d})$ and $\mathrm{prob}_G(s,\mathbf{d})$ is computed through (E3). Then the result follows 
directly from the fact that 
\[
\left|\mathrm{prob}^{\max}_G(s,\mathbf{d})-\mathrm{prob}_G(s,\mathbf{d})\right|\le\max_{a\in\mathrm{En}(s)}\left|\mathrm{prob}^{\max}_G((s,a),\mathbf{d})-\mathrm{prob}_G((s,a),\mathbf{d})\right|~~.
\]
\noindent{Case 2:} $(y,\mathbf{d})$ is computed through the linear program for $\mathbf{d}$. Note that the linear program indeed computes 
the least fixed-point of $\mathcal{Y}_\mathbf{d,G}$. By induction hypothesis, for all $y'\in X_G$,
\begin{align*}
& \left|\mathrm{prob}_G(y',\mathbf{d})-\mathrm{prob}_G^{\max}(y',\mathbf{d})\right|\le \\
&\qquad \left(\frac{2\cdot\mathbf{E}^2_{\max}\cdot\mathbf{w}_{\max}}{N\cdot\mathbf{w}_{\min}}+\frac{1}{N}\right)\cdot\left[\sum_{i=1}^k\mathbf{d}_i\right]+\frac{\mathbf{E}_{\max}}{N}\enskip.
\end{align*}
Thus by induction on $n$, one can prove that 
\[
\left|\mathcal{Y}_{\mathbf{d},n}(\vec{0})-\mathcal{Y}'_{\mathbf{d},n}(\vec{0})\right|\le\left(\frac{2\cdot\mathbf{E}^2_{\max}\cdot\mathbf{w}_{\max}}{N\cdot\mathbf{w}_{\min}}+\frac{1}{N}\right)\cdot\left[\sum_{i=1}^k\mathbf{d}_i\right]+\frac{\mathbf{E}_{\max}}{N}
\]
where $\mathcal{Y}'_{\mathbf{d}}$ is the high order operator obtained by replacing $\{\mathrm{prob}^{\max}_G(y,\mathbf{c})\}_{y\in X_G}$ 
in the definition of $\mathcal{Y}_{\mathbf{d},G}$ with $\{\mathrm{prob}_G(y,\mathbf{c})\}_{y\in X_G}$, and $\mathcal{Y}_{\mathbf{d},n}$ 
(resp. $\mathcal{Y}'_{\mathbf{d},n}$) is the $n$-th iteration of $\mathcal{Y}_{\mathbf{d},G}$ (resp. $\mathcal{Y}'_{\mathbf{d}}$). It 
follows that 
\[
\left|\mathrm{prob}_G(y,\mathbf{d})-\mathrm{prob}_G^{\max}(y,\mathbf{d})\right|\le \left(\frac{2\cdot\mathbf{E}^2_{\max}\cdot\mathbf{w}_{\max}}{N\cdot\mathbf{w}_{\min}}+\frac{1}{N}\right)\cdot\left[\sum_{i=1}^k\mathbf{d}_i\right]+\frac{\mathbf{E}_{\max}}{N}\enskip.
\]
\noindent{Case 3:} $(y,\mathbf{d})=((s,a),\mathbf{d})$ and $\mathrm{prob}(y,\mathbf{d})$ is computed through (E1). 
By Lagrange's Mean-Value Theorem and Theorem~\ref{thm:approx:early:diderivative}, we have 
\[
\mathrm{prob}^{\max}_G(y,\mathbf{d})-\mathrm{prob}_G^{\max}(y,\mathrm{pre}(\mathbf{d},y))=\frac{1}{N}\cdot\nabla\mathrm{prob}^{\max}_G(y,\mathbf{d}-x\cdot\mathbf{w}(y))
\]
for some $x\in (0,\frac{1}{N})$. By Theorem~\ref{thm:approx:early:diderivative} and Theorem~\ref{thm:minearlyf}, we can obtain that
\begin{align*}
& \mathrm{prob}^{\max}_G((s,a),\mathbf{d})=\mathrm{prob}^{\max}_G((s,a),\mathrm{pre}(\mathbf{d},(s,a)))+\delta+{} \tag{*}\\
& \quad\frac{1}{N}\cdot\sum_{s'\in L}\mathbf{R}(s,a,s')\cdot\left[\mathrm{prob}^{\max}_G(s',\mathrm{pre}(\mathbf{d},y))-\mathrm{prob}^{\max}_G(y,\mathrm{pre}(\mathbf{d},y))\right]
\end{align*}
for some 
$\delta\in [-\frac{2}{N^2}\cdot\frac{\mathbf{E}^2_{\max}\cdot\mathbf{w}_{\max}}{\mathbf{w}_{\min}}, \frac{2}{N^2}\cdot\frac{\mathbf{E}^2_{\max}\cdot\mathbf{w}_{\max}}{\mathbf{w}_{\min}}]$\enskip. 
Rewriting (*) and (E1), we obtain that 
\begin{align*}
& \mathrm{prob}^{\max}_G(y,\mathbf{d})=\frac{1}{N}\cdot\sum_{s'\in L}\mathbf{R}(s,a,s')\cdot\mathrm{prob}^{\max}_G(s',\mathrm{pre}(\mathbf{d},y))+{}\\
& \qquad\left(1-\frac{\mathbf{E}(s,a)}{N}\right)\cdot \mathrm{prob}^{\max}_G(y,\mathrm{pre}(\mathbf{d},y))+\delta
\end{align*}
and 
\begin{align*}
& \mathrm{prob}_G(y,\mathbf{d})=\frac{1}{N}\cdot\sum_{s'\in L}\mathbf{R}(s,a,s')\cdot\overline{\mathrm{prob}}_G(s',\mathrm{pre}(\mathbf{d},y))+{}\\
& \qquad\left(1-\frac{\mathbf{E}(s,a)}{N}\right)\cdot \overline{\mathrm{prob}}_G(y,\mathrm{pre}(\mathbf{d},y))\enskip.
\end{align*}
By induction hypothesis, we have 
\begin{align*}
& \left|\mathrm{prob}^{\max}_G(y,\mathbf{d})-\mathrm{prob}_G(y,\mathbf{d})\right|\le \\
& \qquad\left(\frac{2\cdot\mathbf{E}^2_{\max}\cdot\mathbf{w}_{\max}}{N\cdot\mathbf{w}_{\min}}+\frac{1}{N}\right)\cdot\left[\sum_{i=1}^k\left(\mathrm{pre}(\mathrm{d},y)\right)_i\right]+\delta+\frac{1}{N^2}+\frac{\mathbf{E}_{\max}}{N}
\end{align*}
from which the induction step can be obtained. 

~\\
\noindent{Case 4:} $\overline{\mathrm{prob}}(y,\mathbf{d})$ is computed through rounding. The induction step for this case is straightforward. 
\end{proof}
From Theorem~\ref{thm:numerical:early:approxscheme}, we derive our approximation algorithm for early schedulers as follows.

\begin{corollary}\label{crlly:approx:algorithm:final:early}
There exists an algorithm such that given any $\epsilon>0$, $s\in L$, $G\subseteq L$ and $\mathbf{c}\in\mathbb{N}_{0}^k$, the algorithm 
outputs a $d\in [0,1]$ which satisfies that $\left|d-\mathrm{prob}_G^{\max}(s,\mathbf{c})\right|\le\epsilon$. Moreover, the algorithm runs 
in 
\[
\mathcal{O}((\max\{\mathbf{E}_{\max},\frac{M}{\epsilon}\})^k\cdot(\Pi_{i=1}^k\mathbf{c}_i)\cdot (|\mathcal{M}|+\log\frac{M}{\epsilon})^8) 
\]
time, where 
$M:=(2\cdot\mathbf{E}_{\max}^2\cdot\frac{\mathbf{w}_{\max}}{\mathbf{w}_{\min}}+1)\cdot\left[\sum_{i=1}^k\mathbf{c}_i\right]+\mathbf{E}_{\max}+1$ and $|\mathcal{M}|$ is the size of $\mathcal{M}$\enskip.
\end{corollary}
\begin{proof}
The algorithm is a simple application of Theorem~\ref{thm:numerical:early:approxscheme}. If $s\in G$, the algorithm just returns $1$. 
Otherwise, the algorithm just calls Algorithm~\ref{algo:approx:early} with 
$N:=\lfloor\max\{\mathbf{E}_{\max},\frac{M}{\epsilon}\}\rfloor+1$ and set $d=\mathrm{prob}_G(s,\mathbf{c})$; 
By Theorem~\ref{thm:numerical:early:approxscheme}, we directly obtain that 
$\left|d-\mathrm{prob}_G^{\max}(s,\mathbf{c})\right|\le M\cdot\frac{1}{N}$\enskip. For each 
$\mathbf{d}\in\mathrm{Disc}(\mathbf{c},N)$, the total computation of $\{\mathrm{prob}_G(y,\mathbf{d})\}_{y\in X_G\cup Y_G}$ and 
$\{\overline{\mathrm{prob}}_G(y,\mathbf{d})\}_{y\in X_G\cup Y_G}$ takes $\mathcal{O}((|\mathcal{M}|+\log\frac{M}{\epsilon})^8)$ time since the most 
time consuming part is the linear program which takes $\mathcal{O}((|\mathcal{M}|+\log N)^8)$ time (cf. \cite{Schrijver:1986:TLI:17634}). Thus the total 
running time of the 
algorithm is $\mathcal{O}((\max\{\mathbf{E}_{\max},\frac{M}{\epsilon}\})^k\cdot(\Pi_{i=1}^k\mathbf{c}_i)\cdot (|\mathcal{M}|+\log\frac{M}{\epsilon})^8)$ since the size of 
$\mathrm{Disc}(\mathbf{c},N)$ is $\mathcal{O}(N^k\cdot(\Pi_{i=1}^k\mathbf{c}_i))$\enskip. 
\end{proof}

\subsection{Late Schedulers}

The development of a numerical approximation algorithm for late schedulers is much the same as the one for early schedulers. We base our
approximation scheme on Theorem~\ref{thm:approx:late:diderivative} and Proposition~\ref{prop:approx:late:weightzero}. 
In the following we fix a set $G\subseteq L$.  We will omit `$\mathrm{l}$' on the superscript of 
$\mathrm{prob}_G^{\mathrm{l},\max}$, $Y_G^\mathrm{l}$ and $\mathcal{Y}_{c,G}^\mathrm{l}$, etc.

Below we illustrate the discretization and the approximation scheme for late schedulers. Note that
$\mathrm{prob}^{\max}_G(s,\mathbf{c})=1$ whenever $s\in L$ and $\mathbf{c}\ge\vec{0}$. 
Thus we do not need to incorporate those points into the discretization. 

\begin{definition}
Let $\mathbf{c}\ge\vec{0}$ and $N\in\mathbb{N}$. 
The set $\mathsf{D}^\mathbf{c}_N$ of \emph{discretized grid points} is defined as follows:
\[
\mathsf{D}^\mathbf{c}_N:=(L-G)\times \mathrm{Disc}(\mathbf{c},N)\enskip,
\]
where $\mathrm{Disc}(\mathbf{c},N)$ is defined in Definition~\ref{def:numerical:early:disc}.
\end{definition}
The following definition presents the approximation scheme on $\mathsf{D}^\mathbf{c}_N$.

\begin{definition}
Define $X_G:=(L-G)-Y_G$. The \emph{approximation scheme} $\Upsilon^{G}_{\mathbf{c},N}$ on
$\mathsf{D}^\mathbf{c}_N$ consists of the following items:
\begin{itemize}\itemsep1pt\parskip0pt\parsep0pt
\item exactly one \emph{rounding argument} for each element of $\mathsf{D}^\mathbf{c}_N$\enskip;
\item a system of equations for elements in $X_G\times\mathrm{Disc}(\mathbf{c},N)$\enskip;
\item a linear program on $Y_G$ for each $\mathbf{d}\in\mathrm{Disc}(\mathbf{c},N)$ \enskip.
\end{itemize}
\noindent\textbf{Rounding Arguments:} For each element $y\in\mathsf{D}^\mathbf{c}_N$, the rounding argument for $y$ is as follows:
\[
\overline{\mathrm{prob}}_G(y):=\frac{K}{N^2}\mbox{ if }\mathrm{prob}_G(y)\in \left[\frac{K}{N^2},\frac{K+1}{N^2}\right)\mbox{ for some integer }0\le K\le N^2\enskip.
\]

~\\
\noindent\textbf{Equations:} The system of equations is described as follows. For all 
$(s,\mathbf{d})\in\mathsf{D}^\mathbf{c}_N$ with $\mathbf{w}(s)\ne\vec{0}$ and 
$\mathbf{d}-\frac{1}{N}\cdot\mathbf{w}(s)\ge\vec{0}$, there is a linear equation
\begin{align*}
& \frac{\mathrm{prob}_G(s,\mathbf{d})-\overline{\mathrm{prob}}_G(s,\mathrm{pre}(\mathbf{d},s))}{\frac{1}{N}}=\tag{E4}\\
& \quad\max_{a\in\mathrm{En}(s)}\sum_{s'\in L}\mathbf{R}(s,a,s')\cdot\left(\overline{\mathrm{prob}}_G(s',\mathrm{pre}(\mathbf{d},s))-\overline{\mathrm{prob}}_G(s,\mathrm{pre}(\mathbf{d},s))\right)
\end{align*}
where $\mathrm{pre}(\mathbf{d},s):=\mathbf{d}-\frac{1}{N}\cdot\mathbf{w}(s)$. For all  
$(s,\mathbf{d})\in\mathsf{D}^\mathbf{c}_N$ with $\mathbf{w}(s)\ne\vec{0}$ and 
$\mathbf{d}-\frac{1}{N}\cdot\mathbf{w}(s)\not\ge \vec{0}$, there is a linear equation
\begin{equation*}
\mathrm{prob}_G(s,\mathbf{d})=0\enskip.\tag{E5}
\end{equation*}

\noindent\textbf{Linear Programs:} For each $\mathbf{d}\in\mathrm{Disc}(\mathbf{c},N)$, the collection 
$\{\mathrm{prob}_G(s,\mathbf{d})\}_{s\in Y_G}$ of values is the unique optimum solution of the linear program as follows:
\begin{itemize}\itemsep1pt \parskip0pt \parsep0pt
\item[] $\min\sum_{s\in Y_G}\mathrm{prob}_G(s,\mathbf{d})$, subject to:
\item   $\mathrm{prob}_G(s,\mathbf{d})\ge\sum_{s'\in L}\mathbf{P}(s,a,s')\cdot\mathrm{prob}_G(s',\mathbf{d})$ for all $s\in Y_G$ and $a\in\mathrm{En}(s)$;
\item $\mathrm{prob}_G(s,\mathbf{d})\in [0,1]$ for all $s\in Y_G$;
\end{itemize}
where the values $\{\mathrm{prob}_G(s,\mathbf{d})\}_{s\in X_G}$ are assumed to be known. 

In all of the statements above, both $\mathrm{prob}_G(s,\mathbf{d})$ and $\overline{\mathrm{prob}}_G(s,\mathbf{d})$ represents $1$ for $s\in G$. 
\end{definition}
Similar to the case of early schedulers, $\mathrm{prob}_G(s,\mathbf{d})$ approximates $\mathrm{prob}^{\max}_G(s,\mathbf{d})$ and 
$\overline{\mathrm{prob}}_G(s,\mathbf{d})$ approximates the same value with a rounding operation. The detailed  computation  
of the approximation scheme is described in Algorithm~\ref{algo:approx:late}.

By a proof similar to the one of Proposition~\ref{prop:numerical:early:termination}, we can obtain the following proposition.

\begin{proposition}\label{prop:numerical:late:termination}
Algorithm~\ref{algo:approx:late} terminates after a finite number of steps for all $\mathbf{c}\in\mathbb{N}^k_{0}$ and $N\in\mathbb{N}$.
\end{proposition}

\begin{algorithm}                      
\caption{The Computation of $\Upsilon_{\mathbf{c},N}^G$ (for late schedulers)}          
\label{algo:approx:late}                           
\begin{algorithmic}[1]                    
    \STATE Set all grid points in $\{(s,\mathbf{d})\in\mathsf{D}^\mathbf{c}_N\mid\mathbf{d}-\frac{1}{N}\cdot\mathbf{w}(s)\not\ge \vec{0}\}$ to zero by (E5);
    \STATE Compute all $\mathrm{prob}_G(s,\mathbf{d})$ that can be directly obtained through the linear program;
    \STATE Compute all $\overline{\mathrm{prob}}_G(s,\mathbf{d})$ that can be directly obtained by the rounding argument;
    \STATE Compute all $\mathrm{prob}_G((s,a),\mathbf{d})$ that can be directly obtained through (E4), and back to Step 2. until all grid points in $\mathsf{D}^\mathbf{c}_N$ are computed; 
\end{algorithmic}
\end{algorithm}
Below we prove that the approximation scheme really approximates $\mathrm{prob}^{\max}_G$. To ease the notation, we shall use 
$\mathrm{prob}_G(s,\mathbf{d})$ or $\overline{\mathrm{prob}}_G(s,\mathbf{d})$ to denote both the variable at the grid point and the value it holds under the approximation scheme. 

\begin{theorem}\label{thm:numerical:late:algorithm}
Let $\mathbf{c}\in\mathbb{N}^k_{0}$ and $N\in\mathbb{N}$ with $N\ge\mathbf{E}_{\max}$. For each $(s,\mathbf{d})\in\mathsf{D}_N^\mathbf{c}$,
\[
\left|\mathrm{prob}_G(s,\mathbf{d})-\mathrm{prob}_G^{\max}(s,\mathbf{d})\right|\le \left(\frac{2\cdot\mathbf{E}^2_{\max}\cdot\mathbf{w}_{\max}}{N\cdot\mathbf{w}_{\min}}+\frac{1}{N}\right)\cdot\left[\sum_{i=1}^k\mathbf{d}_i\right]+\frac{\mathbf{E}_{\max}}{N}
\]
and 
\begin{align*}
& \left|\overline{\mathrm{prob}}_G(s,\mathbf{d})-\mathrm{prob}_G^{\max}(s,\mathbf{d})\right|\le \\ 
& \qquad\left(\frac{2\cdot\mathbf{E}^2_{\max}\cdot\mathbf{w}_{\max}}{N\cdot\mathbf{w}_{\min}}+\frac{1}{N}\right)\cdot\left[\sum_{i=1}^k\mathbf{d}_i\right]+\frac{\mathbf{E}_{\max}}{N}+\frac{1}{N^2}
\end{align*}
\end{theorem}
\begin{proof}
We proceed by induction on the number of computation steps illustrated by Algorithm~\ref{algo:approx:late}.

~\\
\noindent\textbf{Base Step:} $(s,\mathbf{d})$ satisfies that $\mathbf{d}-\frac{1}{N}\cdot\mathbf{w}(s)\not\ge \vec{0}$\enskip. Then we know 
that $\mathrm{prob}^{\max}_G(s,\mathbf{d}-x\cdot\mathbf{w}(s))=0$ where $x\in [0, \frac{1}{N}]$ is the largest real number such that 
$\mathbf{d}-x\cdot\mathbf{w}(s)\ge\vec{0}$. Then by Theorem~\ref{thm:approx:late:diderivative} and Lagrange's Mean-Value Theorem, we obtain that 
\begin{align*}
& \mathrm{prob}^{\max}_G(s,\mathbf{d})-\mathrm{prob}^{\max}_G(s,\mathbf{d}-x\cdot\mathbf{w}(s))=\\
& \qquad x\cdot\nabla\mathrm{prob}_G^{\max}(s,\mathbf{d}-x'\cdot\mathbf{w}(s))
\end{align*}
for some $x'\in (0,x)$. It follows that $\mathrm{prob}^{\max}_G(s,\mathbf{d})\le\frac{1}{N}\cdot\mathbf{E}(s)$. Thus we have,
\[
\left|\mathrm{prob}^{\max}_G(s,\mathbf{d})-\mathrm{prob}_G(s,\mathbf{d})\right|\le\frac{1}{N}\cdot\mathbf{E}_{\max}\enskip.
\]
\noindent\textbf{Inductive Step.} The inductive step can be classified into the following  cases:

~\\
\noindent{Case 1:} $(s,\mathbf{d})$ is computed through the linear program for $\mathbf{d}$. Note that the linear program indeed computes 
the least fixed point of $\mathcal{Y}_\mathbf{d}$. By induction hypothesis, for all $y'\in X_G$,
\[
\left|\mathrm{prob}_G(y',\mathbf{d})-\mathrm{prob}_G^{\max}(y',\mathbf{d})\right|\le \left(\frac{2\cdot\mathbf{E}^2_{\max}\cdot\mathbf{w}_{\max}}{N\cdot\mathbf{w}_{\min}}+\frac{1}{N}\right)\cdot\left[\sum_{i=1}^k\mathbf{d}_i\right]+\frac{\mathbf{E}_{\max}}{N}\enskip.
\]
Thus by induction on $n$, one can prove that 
\[
\left|\mathcal{Y}_{\mathbf{d},n}(\vec{0})-\mathcal{Y}'_{\mathbf{d},n}(\vec{0})\right|\le \left(\frac{2\cdot\mathbf{E}^2_{\max}\cdot\mathbf{w}_{\max}}{N\cdot\mathbf{w}_{\min}}+\frac{1}{N}\right)\cdot\left[\sum_{i=1}^k\mathbf{d}_i\right]+\frac{\mathbf{E}_{\max}}{N}
\]
where $\mathcal{Y}'_{\mathbf{d}}$ is the high-order operator obtained by replacing $\{\mathrm{prob}^{\max}_G(y,\mathbf{c})\}_{y\in X_G}$ 
in the definition of $\mathcal{Y}_{\mathbf{d},G}$ with $\{\mathrm{prob}_G(y,\mathbf{c})\}_{y\in X_G}$, and $\mathcal{Y}_{\mathbf{d},n}$ 
(resp. $\mathcal{Y}'_{\mathbf{d},n}$) is the $n$-th iteration of $\mathcal{Y}_{\mathbf{d},G}$ (resp. $\mathcal{Y}'_{\mathbf{d}}$). It 
follows that 
\[
\left|\mathrm{prob}_G(y,\mathbf{d})-\mathrm{prob}_G^{\max}(y,\mathbf{d})\right|\le \left(\frac{2\cdot\mathbf{E}^2_{\max}\cdot\mathbf{w}_{\max}}{N\cdot\mathbf{w}_{\min}}+\frac{1}{N}\right)\cdot\left[\sum_{i=1}^k\mathbf{d}_i\right]+\frac{\mathbf{E}_{\max}}{N}\enskip.
\]
\noindent{Case 2:} $\mathrm{prob}_G(s,\mathbf{d})$ is computed through (E4). 
By Lagrange's Mean-Value Theorem and Theorem~\ref{thm:approx:early:diderivative}, we have 
\[
\mathrm{prob}^{\max}_G(s,\mathbf{d})-\mathrm{prob}_G^{\max}(s,\mathrm{pre}(\mathbf{d},y))=\frac{1}{N}\cdot\nabla\mathrm{prob}^{\max}_G(s,\mathbf{d}-x\cdot\mathbf{w}(s))
\]
for some $x\in (0,\frac{1}{N})$. By Theorem~\ref{thm:approx:late:diderivative} and Theorem~\ref{thm:minlatef}, we can obtain that
\begin{align*}
& \mathrm{prob}^{\max}_G(s,\mathbf{d})=\mathrm{prob}^{\max}_G(s,\mathrm{pre}(\mathbf{d},s))+\delta+{} \tag{**}\\
& \quad\frac{1}{N}\cdot\max_{a\in\mathrm{En}(s)}\sum_{s'\in L}\mathbf{R}(s,a,s')\cdot\left[\mathrm{prob}^{\max}_G(s',\mathrm{pre}(\mathbf{d},s))-\mathrm{prob}^{\max}_G(s,\mathrm{pre}(\mathbf{d},s))\right]
\end{align*}
for some $\delta\in [-\frac{2\cdot\mathbf{E}^2_{\max}\cdot\mathbf{w}_{\max}}{N^2\cdot\mathbf{w}_{\min}}, \frac{2\cdot\mathbf{E}^2_{\max}\cdot\mathbf{w}_{\max}}{N^2\cdot\mathbf{w}_{\min}}]$\enskip. Rewriting (**) and (E4), we obtain 
\begin{align*}
& \mathrm{prob}^{\max}_G(s,\mathbf{d})=\frac{1}{N}\cdot\max_{a\in\mathrm{En}(s)}\left[\sum_{s'\in L}\mathbf{R}(s,a,s')\cdot\mathrm{prob}^{\max}_G(s',\mathrm{pre}(\mathbf{d},s))\right]+{}\\
& \qquad\left(1-\frac{\mathbf{E}(s)}{N}\right)\cdot \mathrm{prob}^{\max}_G(y,\mathrm{pre}(\mathbf{d},y))+\delta
\end{align*}
and 
\begin{align*}
& \mathrm{prob}_G(s,\mathbf{d})=\frac{1}{N}\cdot\max_{a\in\mathrm{En}(s)}\sum_{s'\in L}\mathbf{R}(s,a,s')\cdot\overline{\mathrm{prob}}_G(s',\mathrm{pre}(\mathbf{d},s))+{}\\
& \qquad\left(1-\frac{\mathbf{E}(s)}{N}\right)\cdot \overline{\mathrm{prob}}_G(s,\mathrm{pre}(\mathbf{d},s))\enskip.
\end{align*}
By induction hypothesis, we have 
\begin{align*}
& \left|\mathrm{prob}^{\max}_G(s,\mathbf{d})-\mathrm{prob}_G(s,\mathbf{d})\right|\le \\ 
& \qquad\left(\frac{2\cdot\mathbf{E}^2_{\max}\cdot\mathbf{w}_{\max}}{N\cdot\mathbf{w}_{\min}}+\frac{1}{N}\right)\cdot\left[\sum_{i=1}^k\left(\mathrm{pre}(\mathbf{d},s)\right)_i\right]+\frac{\mathbf{E}_{\max}}{N}+\delta+\frac{1}{N^2}
\end{align*}
from which the induction step can be obtained. 

~\\
\noindent{Case 3:} $\overline{\mathrm{prob}}(y,\mathbf{d})$ is computed through rounding. The induction step for this case is straightforward. 
\end{proof}

\begin{corollary}
There exists an algorithm such that given any $\epsilon>0$, $s\in L$, $G\subseteq L$ and $\mathbf{c}\in\mathbb{N}_{0}^k$, the algorithm 
outputs a $d\in [0,1]$ which satisfies that $\left|d-\mathrm{prob}_G^{\max}(s,\mathbf{c})\right|\le\epsilon$. Moreover, the algorithm runs 
in 
\[
\mathcal{O}((\max\{\mathbf{E}_{\max},\frac{M}{\epsilon}\})^k\cdot(\Pi_{i=1}^k\mathbf{c}_i)\cdot (|\mathcal{M}|+\log\frac{M}{\epsilon})^8) 
\]
time, where 
$M$ is defined as in Corollary~\ref{crlly:approx:algorithm:final:early}\enskip.
\end{corollary}
\begin{proof}
The algorithm is an simple application of Theorem~\ref{thm:numerical:late:algorithm}. If $s\in G$, the algorithm just returns $1$. 
Otherwise, the algorithm just calls Algorithm~\ref{algo:approx:late} with 
$N:=\lfloor\max\{\mathbf{E}_{\max},\frac{M}{\epsilon}\}\rfloor+1$ and set $d=\mathrm{prob}_G(s,\mathbf{c})$; 
By Theorem~\ref{thm:numerical:late:algorithm}, we directly obtain that $\left|d-\mathrm{prob}_G^{\max}(s,\mathbf{c})\right|\le M\cdot\frac{1}{N}$\enskip. 
For each $\mathbf{d}\in\mathrm{Disc}(\mathbf{c},N)$, the total computation of $\{\mathrm{prob}_G(s,\mathbf{d})\}_{s\in X_G\cup Y_G}$ and 
$\{\overline{\mathrm{prob}}_G(s,\mathbf{d})\}_{s\in X_G\cup Y_G}$ takes $\mathcal{O}(|\mathcal{M}|+\log\frac{M}{\epsilon})^8)$ time since the most 
time consuming part is the linear program which takes $\mathcal{O}((|\mathcal{M}|+\log N)^8)$ time (cf. \cite{Schrijver:1986:TLI:17634}). Thus the total 
running time of the 
algorithm is $\mathcal{O}((\max\{\mathbf{E}_{\max},\frac{M}{\epsilon}\})^k\cdot(\Pi_{i=1}^k\mathbf{c}_i)\cdot (|\mathcal{M}|+\log\frac{M}{\epsilon})^8)$ since the size of 
$\mathrm{Disc}(\mathbf{c},N)$ is $\mathcal{O}(N^k\cdot(\Pi_{i=1}^k\mathbf{c}_i))$\enskip. 
\end{proof}

\section{Conclusion}

In this paper, we established an integral characterization for multi-dimensional maximal cost-bounded reachability probabilities 
in continuous-time Markov decision processes, the existence of deterministic cost-positional optimal scheduler and an algorithm to approximate the cost-bounded 
reachability probability with an error bound, under the setting of both early and late schedulers. The approximation algorithm is 
based on a differential characterization of cost-bounded reachability probability which in turn is derived from the integral 
characterization. The error bound is obtained through the differential characterization and the Lipschitz property described. 
Moreover, the approximation algorithm runs in polynomial time in the size of the CTMDP and the reciprocal of the error bound, and exponential in
the dimension of the unit-cost vector. 
An important missing part is the generation of an $\epsilon$-optimal scheduler. However, we conjecture that an $\epsilon$-optimal 
scheduler is not difficult to obtain given that the approximation scheme has been established. 

A future direction is to determine an $\epsilon$-optimal scheduler, under both early and late schedulers. Besides, we believe that the 
paradigms developed in this paper can also be applied to minimum cost-bounded reachability probability and even stochastic games~\cite{controlledCTMG} with 
multi-dimensional cost-bounded reachability objective. 

\subsubsection*{Acknowledgement}

I thank Prof. Joost-Pieter Katoen for his valuable advices on the writing of the paper, especially for the Introduction part. 

\bibliographystyle{plain}
\bibliography{references}

\appendix

\section{Proofs for Section~\ref{sect:scheduler}}

\noindent\textbf{Proposition~\ref{prop:init}.}~~For each early (late) scheduler $D$ and each initial distribution $\alpha$, 
$\mathrm{Pr}_{\mathcal{M},D,\alpha}(\Pi)=\sum_{s\in L}\alpha(s)\cdot\mathrm{Pr}_{\mathcal{M},D,\mathcal{D}[s]}(\Pi)$
for all $\Pi\in\mathcal{S}_\mathcal{M}$. 
\begin{proof}
Define $\mathrm{Pr}'(\Pi):=\sum_{s\in L}\alpha(s)\cdot\mathrm{Pr}_{\mathcal{M},D,\mathcal{D}[s]}(\Pi)$ for $\Pi\in\mathcal{S}_\mathcal{M}$. We prove that $\mathrm{Pr'}$ coincides with $\mathrm{Pr}_{\mathcal{M},D,\alpha}$. It is clear that $\mathrm{Pr}'$ is a probability measure since each $\mathrm{Pr}_{\mathcal{M},D,\mathcal{D}[s]}$ is a probability measure. So it suffices to prove that $\mathrm{Pr}'$ coincides with $\mathrm{Pr}_{\mathcal{M},D,\alpha}$ on $\bigcup_{n\ge 0}\{\mathrm{Cyl}(\Xi)\mid\Xi\in\mathcal{S}^n_\mathcal{M}\}$. To this end, we proceed by induction on $n$. 

~\\
\noindent\textbf{Base Step:} $n=0$ and $\Xi\in\mathcal{S}^n_\mathcal{M}$. By definition, we have
\[
\mathrm{Pr}_{\mathcal{M},D,\alpha}(\mathrm{Cyl}(\Xi))=\mathrm{Pr}^{0}_{\mathcal{M},D,\alpha}(\Xi)=\sum_{s\in L}\alpha(s)\cdot\mathbf{1}_\Xi(s)=\mathrm{Pr}'(\mathrm{Cyl}(\Xi))\enskip. 
\]
\textbf{Inductive Step:} Suppose $\Xi\in\mathcal{S}^{n+1}_\mathcal{M}$. Define
\[
g(\xi):=\int_{\Gamma_{\mathcal{M}}}\mathbf{1}_\Xi(\xi\circ\gamma)~\mu^D_\mathcal{M}(\xi,\mathrm{d}\gamma)\enskip.
\]
Let $\{g_m\}_{m\ge 0}$ be a sequence of simple functions that converges to $g$, which are denoted by $g_m=\sum_{i=1}^{l_m}d_m^i\cdot\mathbf{1}_{\Xi_i}$ for which $l_m\ge 1$, $d_m^i\ge 0$ and $\Xi_i\in\mathcal{S}^n_\mathcal{M}$ for all $1\le i\le l_m$. 
Then
\[
\mathrm{Pr}^{n+1}_{\mathcal{M},D,\alpha}(\Xi)=\int_{\Omega^n_\mathcal{M}} g(\xi)~\mathrm{Pr}^n_{\mathcal{M},\rho,\alpha}(\mathrm{d}\xi)=\lim\limits_{m\rightarrow\infty}\sum_{i=1}^{l_m}d_m^i\cdot\mathrm{Pr}^n_{\mathcal{M},D,\alpha}(\Xi_i)\enskip.
\]
By induction hypothesis, 
\[
\mathrm{Pr}_{\mathcal{M},D,\alpha}\left(\mathrm{Cyl}(\Xi_i)\right)=\mathrm{Pr}^n_{\mathcal{M},D,\alpha}(\Xi_i)=\sum_{s\in L}\alpha(s)\cdot\mathrm{Pr}_{\mathcal{M},D,\mathcal{D}[s]}\left(\mathrm{Cyl}(\Xi_i)\right)\enskip.
\]
Thus we have
\begin{eqnarray*}
& &\mathrm{Pr}_{\mathcal{M},D,\alpha}\left(\mathrm{Cyl}(\Xi)\right)\\
&=&\mathrm{Pr}^{n+1}_{\mathcal{M},D,\alpha}(\Xi)\\
&=&\lim\limits_{m\rightarrow\infty}\sum_{i=1}^{l_m}d_i\cdot \sum_{s\in L}\alpha(s)\cdot\mathrm{Pr}_{\mathcal{M},D,\mathcal{D}[s]}\left(\mathrm{Cyl}(\Xi_i)\right)\\
&=&\lim\limits_{m\rightarrow\infty}\sum_{s\in L}\alpha(s)\cdot\sum_{i=1}^{l_m}d_i\cdot\mathrm{Pr}_{\mathcal{M},D,\mathcal{D}[s]}\left(\mathrm{Cyl}(\Xi_i)\right)\\
&=& \sum_{s\in L}\alpha(s)\cdot\lim\limits_{m\rightarrow\infty}\sum_{i=1}^{l_m}d_i\cdot\mathrm{Pr}^n_{\mathcal{M},D,\mathcal{D}[s]}(\Xi_i))\\
&=& \sum_{s\in L}\alpha(s)\cdot\int_{\Omega^n_\mathcal{M}} g(\xi)~\mathrm{Pr}^n_{\mathcal{M},D,\mathcal{D}[s]}(\mathrm{d}\xi)\\
&=& \sum_{s\in L}\alpha(s)\cdot\mathrm{Pr}^{n+1}_{\mathcal{M},D,\mathcal{D}[s]}(\Xi)\\
&=& \sum_{s\in L}\alpha(s)\cdot\mathrm{Pr}_{\mathcal{M},D,\mathcal{D}[s]}\left(\mathrm{Cyl}(\Xi)\right)\enskip,
\end{eqnarray*}
which justifies the induction hypothesis.
\end{proof}

\section{Proofs for Section~\ref{sect:fixpoint}}

\noindent\textbf{Lemma~\ref{lemm:shift:set:measurability}.}
$P^{s,a}_\Pi(t)\in\mathcal{S}_\mathcal{M}$ for all $\Pi\in\mathcal{S}_\mathcal{M}$, $(s,a)\in L\times Act$ and 
$t\in\mathbb{R}_{\ge 0}$\enskip. Analogously, $H^{s,a}_\Xi(t)\in\mathcal{S}^{n-1}_\mathcal{M}$ for all $n\ge 1$, 
$\Xi\in\mathcal{S}^n_\mathcal{M}$ , $(s,a)\in L\times Act$ and $t\in\mathbb{R}_{\ge 0}$\enskip.
\begin{proof}
We first prove the case for paths. Fix some $s\in L$, $a\in Act$ and $t\ge 0$. Define the set $\mathcal{S}'$ by:
$\mathcal{S}':=\{\Pi\in\mathcal{S}_\mathcal{M}\mid P^{s,a}_{\Pi}(t)\in\mathcal{S}_\mathcal{M}\}$\enskip. We prove that 
$\mathcal{S'}=\mathcal{S}_\mathcal{M}$. By Remark~\ref{rmk:msp}, it suffices to prove that 
\begin{enumerate}\itemsep1pt \parskip0pt \parsep0pt
\item $\{\mathrm{Cyl}(\mathrm{Hists}(\theta))\mid \theta\mbox{ is a template}\}\subseteq \mathcal{S}'$, and
\item $\mathcal{S}'$ is a $\sigma$-algebra.
\end{enumerate}
The first point follows directly from the definition of templates. To see the second point, one can verify that (i) 
$P^{s,a}_{\Omega_\mathcal{M}}(t)=\Omega_\mathcal{M}$, (ii) $P^{s,a}_{\Pi^\mathrm{c}}(t)=\left[P^{s,a}_\Pi(t)\right]^\mathrm{c}$ and (iii) 
$P^{s,a}_{\bigcup_{n\ge 0}\Pi_n}(t)=\bigcup_{n\ge 0}P^{s,a}_{\Pi_n}(t)$\enskip.

The proof for histories can be similarly obtained by proving the following fact:
$\{\Xi\in\mathcal{S}^n_\mathcal{M}\mid H^{s,a}_\Xi(t)\in\mathcal{S}^{n-1}_\mathcal{M}\}=\mathcal{S}^n_\mathcal{M}$ for $n\ge 1$.
\end{proof}

\noindent\textbf{Lemma~\ref{lemm:shift:schd:measurability}.}
Let $s\in L$, $a\in Act$ and $t\in\mathbb{R}_{\ge 0}$. For each measurable early (late) scheduler $D$, $D[s\xrightarrow{a,t}]$ is a 
measurable early (late) scheduler.
\begin{proof}
Fix some arbitrary $b\in Act$, $\epsilon\in [0,1]$ and $n\ge 0$. If $D$ is a measurable early scheduler, then we have
\[
\{\xi\in Hists^n_{\mathcal{M}}\mid D[s\xrightarrow{a,t}](\xi,b)\le\epsilon\}=P^{s,a}_H(t)
\]
where $H:=\{\xi\in Hists^{n+1}(\mathcal{M})\mid D(\xi,b)\le\epsilon\}$\enskip. By Lemma~\ref{lemm:shift:set:measurability}, we obtain the 
desired result. Suppose now that $D$ is a measurable late scheduler. Define 
\[
\mathrm{shift}(X):=\{(\xi,\tau)\in Hists^n(\mathcal{M})\times\mathbb{R}_{\ge 0}\mid (s\xrightarrow{a,t}\xi,\tau)\in X\} 
\]
for each  $X\in\mathcal{S}^{n+1}_\mathcal{M}(\mathcal{M})\otimes\mathcal{B}(\mathbb{R}_{\ge 0})$\enskip. Let 
\[
\mathcal{X}':=\{X\in\mathcal{S}^{n+1}_\mathcal{M}\otimes\mathcal{B}(\mathbb{R}_{\ge 0})\mid \mathrm{shift}(X)\in \mathcal{S}^{n}_\mathcal{M}\otimes\mathcal{B}(\mathbb{R}_{\ge 0})\}\enskip. 
\]
Then one can prove $\mathcal{X}'=\mathcal{S}^{n+1}\otimes\mathcal{B}(\mathbb{R}_{\ge 0})$ similar to the proof of 
Lemma~\ref{lemm:shift:set:measurability}. In detail, one can proceed by proving that
\begin{enumerate}\itemsep1pt \parskip0pt \parsep0pt
\item $\{\mathrm{Hists}(\theta)\times I\mid|\theta|=n+1\mbox{ and }I\mbox{ is an interval of }\mathbb{R}_{\ge 0}\}\subseteq\mathcal{X}'$, 
and
\item $\mathcal{X}'$ is a $\sigma$-algebra.
\end{enumerate}
From
\begin{align*}
& \{(\xi,\tau)\in Hists^{n}(\mathcal{M})\times\mathbb{R}_{\ge 0}\mid D[s\xrightarrow{a,t}](\xi,\tau,b)\le\epsilon\}=\\
& \qquad\mathrm{shift}\left(\{(\xi,\tau)\in Hists^{n+1}(\mathcal{M})\times\mathbb{R}_{\ge 0}\mid D(\xi,\tau,b)\le\epsilon\}\right)\enskip,
\end{align*}
we obtain $\{(\xi,\tau)\in Hists^{n}(\mathcal{M})\times\mathbb{R}_{\ge 0}\mid D[s\xrightarrow{a,t}](\xi,\tau,b)\le\epsilon\}$ is 
measurable.
\end{proof}

\noindent\textbf{Proposition~\ref{prop:shiftfunc:measurability}.}~~
$p_{\Pi,D}^{s,a}$ is a measurable function w.r.t $(\mathbb{R}_{\ge 0}, \mathcal{B}(\mathbb{R}_{\ge 0}))$ given any 
$\Pi\in\mathcal{S}_\mathcal{M}$, $s\in L$, $a\in Act$ and measurable early (late) scheduler $D$.
\begin{proof}
Fix some $s\in L$, $a\in Act$ and measurable scheduler $D$. Define the set 
\[
\mathcal{S}':=\{\Pi\in\mathcal{S}_\mathcal{M}\mid p^{s,a}_{\Pi,D}\mbox{ is a measurable function}\}\enskip.
\] 
We show that 
\begin{enumerate}\itemsep1pt \parskip0pt \parsep0pt
\item $\mathcal{S}'$ is a $\lambda$-system (Dynkin system), and
\item the $\pi$-system $\{\mathrm{Cyl}(\mathrm{Hists}(\theta))\mid \theta\mbox{ is a template}\}\subseteq\mathcal{S}'$\enskip.
\end{enumerate}
The second point follows from the fact that $p_{\mathrm{Cyl}(\mathrm{Hists}(\theta)),D}^{s,a}=\sum_{s\in L}d_s\cdot\mathbf{1}_{X_s}$ where each $d_s$ is 
some real number in $[0,1]$ and each $X_s$ is some Borel set on $\mathbb{R}_{\ge 0}$, which can be observed from the definition of 
templates. The first point follows from the following facts:
\begin{itemize} \itemsep1pt \parskip0pt \parsep0pt
\item $p^{s,a}_{\Omega_\mathcal{M},D}$ is measurable since $p^{s,a}_{\Omega_\mathcal{M},D}(t)=1$ for all $t\ge 0$;
\item Whenever $\Pi_1,\Pi_2\in\mathcal{S}'$ and $\Pi_1\subseteq\Pi_2$, we have $\Pi_2\backslash\Pi_1\in\mathcal{S}'$ since 
$p_{\Pi_2\backslash\Pi_1,D}^{s,a}=p_{\Pi_2,D}^{s,a}-p_{\Pi_1,D}^{s,a}$; 
\item For any sequence $\{\Pi_n\}_{n\ge 0}$ such that $\Pi_n\in\mathcal{S}'$ and $\Pi_n\subseteq\Pi_{n+1}$ for all $n\ge 0$, we have 
$\bigcup_{n\ge 0}\Pi_n\in\mathcal{S'}$ since $p_{\bigcup_{n\ge 0}\Pi_n,D}^{s,a}=\lim\limits_{n\rightarrow\infty}p_{\Pi_n,D}^{s,a}$\enskip.
\end{itemize}
By applying Dynkin's $\pi$-$\lambda$ Theorem, we obtain that  $\mathcal{S}_\mathcal{M}\subseteq\mathcal{S}'$, which implies the result. 
\end{proof}

\noindent\textbf{Theorem~\ref{thm:fix-early}}.
Let $D$ be a measurable early scheduler. For each $\Pi\in\mathcal{S}_\mathcal{M}$ and $s\in L$, we have
\[
\mathrm{Pr}_{D,\mathcal{D}[s]}(\Pi)=\sum_{a\in\mathrm{En}(s)}D(s,a)\cdot\int_{0}^\infty f_{\mathbf{E}(s,a)}(t)\cdot p^{s,a}_{\Pi,D}(t)\,\mathrm{d}t\enskip.
\]
\begin{proof}
Define $\mathrm{Pr}':\mathcal{S}_\mathcal{M}\rightarrow [0,1]$ by:
\[
\mathrm{Pr}'(\Pi)=\sum_{a\in\mathrm{En}(s)}D(s,a)\cdot\int_{0}^\infty f_{\mathbf{E}(s,a)}(t)\cdot p^{s,a}_{\Pi,D}(t)\,\mathrm{d}t\enskip.
\]
We prove that $\mathrm{Pr}'$ coincides with $\mathrm{Pr}_{D,\mathcal{D}[s]}$\enskip. It suffices to prove that $\mathrm{Pr}'$ and 
$\mathrm{Pr}_{D,\mathcal{D}[s]}$ coincide on $\bigcup_{n\ge 0}\{\mathrm{Paths}(\Xi)\mid \Xi\in\Omega^n_\mathcal{M}\}$. The proof proceeds 
through induction on $n$.

\textbf{Base Step:} $n\in\{0,1\}$ and $\Xi\in\Omega^n_\mathcal{M}$. If $n=0$, we have 
\[
\mathrm{Pr}_{D,\mathcal{D}[s]}(\mathrm{Cyl}(\Xi))=\mathrm{Pr}_{D,\mathcal{D}[s]}^0(\Xi)=\mathbf{1}_{\Xi}(s)=\mathrm{Pr}'(\mathrm{Cyl}(\Xi))\enskip.
\]
Otherwise (i.e., $n=1$), we have
\begin{eqnarray*}
& & \mathrm{Pr}_{D,\mathcal{D}[s]}(\mathrm{Cyl}(\Xi)) \\
&=& \mathrm{Pr}_{D,\mathcal{D}[s]}^1(\Xi) \\
&=& \int_{\Omega^0_\mathcal{M}}\left[\int_{\Gamma_\mathcal{M}}\mathbf{1}_\Xi(\xi\circ\gamma)~\mu^D_\mathcal{M}(\xi,\mathrm{d}\gamma)\right]\mathrm{Pr}^0_{D,\mathcal{D}[s]}(\mathrm{d}\xi)\\
&=& \int_{\Gamma_\mathcal{M}}\mathbf{1}_\Xi(s\circ\gamma)~\mu^D_\mathcal{M}(s,\mathrm{d}\gamma)
\end{eqnarray*}
Let $U:=\{\gamma\in\Gamma_\mathcal{M}\mid s\circ\gamma\in\Xi\}$\enskip. Then $U\in\mathcal{U}_\mathcal{M}$. Thus, 
\begin{eqnarray*}
& & \int_{\Gamma_\mathcal{M}}\mathbf{1}_\Xi(s\circ\gamma)~\mu^D_\mathcal{M}(s,\mathrm{d}\gamma) \\
&=& \mu^D_\mathcal{M}(s, U) \\
&=& \sum_{a\in\mathrm{En}(s)}D(s,a)\cdot\int_{\mathbb{R}_{\ge 0}} {f}_{\mathbf{E}(s,a)}(t)\cdot\left[\sum_{s'\in L}\mathbf{1}_U(a,t,s')\cdot\mathbf{P}(s,a,s')\right]\,\mathrm{d}t\\
&=& \sum_{a\in\mathrm{En}(s)}D(s,a)\cdot\int_{\mathbb{R}_{\ge 0}} {f}_{\mathbf{E}(s,a)}(t)\cdot\left[\sum_{s'\in L}\mathbf{1}_\Xi(s
\xrightarrow{a,t}s')\cdot\mathbf{P}(s,a,s')\right]\,\mathrm{d}t\\
&=& \sum_{a\in\mathrm{En}(s)}D(s,a)\cdot\int_{\mathbb{R}_{\ge 0}} {f}_{\mathbf{E}(s,a)}(t)\cdot p^{s,a}_{\mathrm{Cyl}(\Xi),D}(t)\,\mathrm{d}t
\end{eqnarray*}
where the last equality follows from Proposition~\ref{prop:init}.

\textbf{Inductive Step:} Suppose $\Xi\in\mathcal{S}^{n+1}_\mathcal{M}$ with $n\ge 1$. Denote
\[ 
g_{\Xi,D}(\xi):=\int_{\Gamma_{\mathcal{M}}}\mathbf{1}_\Xi(\xi\circ\gamma)~\mu^D_\mathcal{M}(\xi,\mathrm{d}\gamma)\enskip.
\]
Let $\{g_m:\Omega^n_{\mathcal{M}}\rightarrow [0,1]\}_{m\ge 0}$ be a sequence of simple functions that converges to $g_{\Xi,D}$. Denote $g_m=\sum_{i=1}^{l_m}d_m^i\cdot\mathbf{1}_{\Xi_m^i}$\enskip.   Then, we have:
\begin{eqnarray*}
& & \mathrm{Pr}_{D,\mathcal{D}[s]}\left(\mathrm{Cyl}(\Xi)\right)\\
&=& \mathrm{Pr}^{n+1}_{D,\mathcal{D}[s]}\left(\Xi\right)\\
&=& \int_{\Omega^n_\mathcal{M}}g_{\Xi,D}(\xi)~\mathrm{Pr}^{n}_{D,\mathcal{D}[s]}(\mathrm{d}\xi)  \\
&=& \lim\limits_{m\rightarrow\infty}\sum_{i=1}^{l_m} d_m^i\cdot\mathrm{Pr}^n_{D,\mathcal{D}[s]}\left(\Xi_m^i\right) \\
&=& \lim\limits_{m\rightarrow\infty}\sum_{i=1}^{l_m} d_m^i\cdot\left[\sum_{a\in Act}D(s,a)\cdot\int_0^\infty f_{\mathbf{E}(s,a)}(t)\cdot p^{s,a}_{\mathrm{Cyl}(\Xi_m^i),D}(t)\,\mathrm{d}t\right] \\
&=& \lim\limits_{m\rightarrow\infty}\sum_{a\in\mathrm{En}(s)}D(s,a)\cdot\left[\int_0^\infty f_{\mathbf{E}(s,a)}(t)\cdot\sum_{i=1}^{l_m} d_m^i\cdot\cdot p^{s,a}_{\mathrm{Cyl}(\Xi_m^i),D}(t)\,\mathrm{d}t\right]  \\
&=& \sum_{a\in\mathrm{En}(s)}D(s,a)\cdot\left[\int_0^\infty f_{\mathbf{E}(s,a)}(t)\cdot\lim\limits_{m\rightarrow\infty}\sum_{i=1}^{l_m} d_m^i\cdot p^{s,a}_{\mathrm{Cyl}(\Xi_m^i),D}(t)\,\mathrm{d}t\right]  
\end{eqnarray*}
where the fourth equality follows from the induction hypothesis. Note that 
\begin{eqnarray*}
& & \lim\limits_{m\rightarrow\infty}\sum_{i=1}^{l_m} d_m^i\cdot p^{s,a}_{\mathrm{Cyl}(\Xi_m^i),D}(t) \\
&=& \lim\limits_{m\rightarrow\infty}\sum_{i=1}^{l_m} d_m^i\cdot \mathrm{Pr}_{D[s\xrightarrow{a,t}],\mathbf{P}(s,a,\centerdot)}\left(P_{\mathrm{Cyl}(\Xi_m^i)}^{s,a}(t)\right)\\
&=& \lim\limits_{m\rightarrow\infty}\sum_{i=1}^{l_m}d_m^i\cdot\mathrm{Pr}^{n-1}_{D[s\xrightarrow{a,t}],\mathbf{P}(s,a,\centerdot)}\left(H^{s,a}_{\Xi^i_m}(t)\right)\enskip.
\end{eqnarray*}
Denote $g'_m(\xi):=\sum_{i=1}^{l_m}d_m^i\cdot\mathbf{1}_{H^{s,a}_{\Xi^i_m}(t)}(\xi)$\enskip. 
Then  
$g'_m(\xi)=g_m(s\xrightarrow{a,t}\xi)$. It follows that $\lim\limits_{m\rightarrow\infty} g'_m(\xi)= g_{\Xi,D}(s\xrightarrow{a,t}\xi)$. By definition, 
\begin{enumerate} \itemsep1pt \parskip0pt \parsep0pt
\item $\mathbf{1}_{\Xi}\left((s\xrightarrow{a,t}\xi)\circ\gamma\right)=\mathbf{1}_{H^{s,a}_{\Xi}(t)}\left(\xi\circ\gamma\right)$ for all combined actions $\gamma$;
\item $\mu^D_\mathcal{M}(s\xrightarrow{a,t}\xi,U)=\mu^{D[s\xrightarrow{a,t}]}_\mathcal{M}(\xi,U)$ for all $U\in\Gamma_\mathcal{M}$.
\end{enumerate}
Thus $g_{\Xi,D}(s\xrightarrow{a,t}\xi)=g_{H^{s,a}_{\Xi}(t),~D[s\xrightarrow{a,t}]}(\xi)=\lim\limits_{m\rightarrow\infty} g'_m(\xi)$. It follows that 
\begin{eqnarray*}
& &\lim\limits_{m\rightarrow\infty}\sum_{i=1}^{l_m}d_m^i\cdot\mathrm{Pr}^{n-1}_{D[s\xrightarrow{a,t}],\mathbf{P}(s,a,\centerdot)}\left(H^{s,a}_{\Xi^i_m}(t)\right)\\
&=&\int_{\Omega_\mathcal{M}^{n-1}}g_{H^{s,a}_{\Xi}(t),D[s\xrightarrow{a,t}]}(\xi)~\mathrm{Pr}^{n-1}_{D[s\xrightarrow{a,t}],\mathbf{P}(s,a,\centerdot)}(\mathrm{d}\xi) \\
&=&\mathrm{Pr}^n_{D[s\xrightarrow{a,t}],\mathbf{P}(s,a,\centerdot)}\left(H^{s,a}_{\Xi}(t)\right)\enskip.
\end{eqnarray*}
Then we have
\begin{eqnarray*}
& &\mathrm{Pr}_{D,\mathcal{D}[s]}\left(\mathrm{Cyl}(\Xi)\right) \\
&=&\sum_{a\in Act}D(s,a)\cdot\left[\int_0^\infty f_{\mathbf{E}(s,a)}(t)\cdot\mathrm{Pr}^n_{D[s\xrightarrow{a,t}],\mathbf{P}(s,a,\centerdot)}\left(H^{s,a}_{\Xi}(t)\right)\,\mathrm{d}t\right] \\
&=&\sum_{a\in Act}D(s,a)\cdot\left[\int_0^\infty f_{\mathbf{E}(s,a)}(t)\cdot p_{\mathrm{Cyl}(\Xi),D}^{s,a}(t)\,\mathrm{d}t\right]
\end{eqnarray*}
which completes the inductive step.
\end{proof}

\noindent\textbf{Theorem~\ref{thm:fix-late}.}
Let $D$ be a measurable late scheduler. For each $\Pi\in\mathcal{S}_\mathcal{M}$ and $s\in L$, we have
\[
\mathrm{Pr}_{D,\mathcal{D}[s]}(\Pi)=\int_{0}^\infty f_{\mathbf{E}(s)}(t)\cdot\left[\sum_{a\in\mathrm{En}(s)} D(s,t,a)\cdot p^{s,a}_{\Pi,D}(t)\right]\,\mathrm{d}t\enskip.
\]
\begin{proof}
The proof follows the lines of Theorem~\ref{thm:fix-early}. Define $\mathrm{Pr}':\mathcal{S}_\mathcal{M}\rightarrow [0,1]$ by:
\[
\mathrm{Pr}'(\Pi)=\int_{0}^\infty f_{\mathbf{E}(s)}(t)\cdot\left[\sum_{a\in\mathrm{En}(s)} D(s,t,a)\cdot p^{s,a}_{\Pi,D}(t)\right]\,\mathrm{d}t\enskip.
\]
We prove that $\mathrm{Pr}'$ coincides with $\mathrm{Pr}_{D,\mathcal{D}[s]}$\enskip. It suffices to prove that $\mathrm{Pr}'$ and 
$\mathrm{Pr}_{D,\mathcal{D}[s]}$ coincide on $\bigcup_{n\ge 0}\{\mathrm{Cyl}(\Xi)\mid \Xi\in\Omega^n_\mathcal{M}\}$. The proof proceeds by induction on $n$.

\textbf{Base Step:} $n\in\{0,1\}$ and $\Xi\in\Omega^n_\mathcal{M}$. If $n=0$, we have 
\[
\mathrm{Pr}_{D,\mathcal{D}[s]}(\mathrm{Cyl}(\Xi))=\mathrm{Pr}_{D,\mathcal{D}[s]}^0(\Xi)=\mathbf{1}_{\Xi}(s)=\mathrm{Pr}'(\mathrm{Cyl}(\Xi))\enskip.
\]
Otherwise (i.e.  $n=1$), we have
\begin{eqnarray*}
& & \mathrm{Pr}_{D,\mathcal{D}[s]}(\mathrm{Cyl}(\Xi)) \\
&=& \mathrm{Pr}_{D,\mathcal{D}[s]}^1(\Xi) \\
&=& \int_{\Omega^0_\mathcal{M}}\left[\int_{\Gamma_\mathcal{M}}\mathbf{1}_\Xi(\xi\circ\gamma)~\mu^D_\mathcal{M}(\xi,\mathrm{d}\gamma)\right]\mathrm{Pr}^0_{D,\mathcal{D}[s]}(\mathrm{d}\xi)\\
&=& \int_{\Gamma_\mathcal{M}}\mathbf{1}_\Xi(s\circ\gamma)~\mu^D_\mathcal{M}(s,\mathrm{d}\gamma)\enskip.
\end{eqnarray*}
Let $U:=\{\gamma\in\Gamma_\mathcal{M}\mid s\circ\gamma\in\Xi\}$\enskip. Then $U\in\mathcal{U}_\mathcal{M}$. Thus, 
\begin{eqnarray*}
& & \int_{\Gamma_\mathcal{M}}\mathbf{1}_\Xi(s\circ\gamma)~\mu^D_\mathcal{M}(s,\mathrm{d}\gamma) \\
&=& \mu^D_\mathcal{M}(s, U) \\
&=& \int_{\mathbb{R}_{\ge 0}}f_{\mathbf{E}(s)}(t)\cdot\left\{\sum_{a\in \mathrm{En}(s)}D(s,t,a)\cdot\left[\sum_{s'\in L}\mathbf{1}_{U}(a,t,s')\cdot\mathbf{P}(s,a,s')\right]\right\}\,\mathrm{d}t\\
&=& \int_{\mathbb{R}_{\ge 0}}f_{\mathbf{E}(s)}(t)\cdot\left\{\sum_{a\in \mathrm{En}(s)}D(s,t,a)\cdot\left[\sum_{s'\in L}\mathbf{1}_{\Xi}(s\xrightarrow{a,t}s')\cdot\mathbf{P}(s,a,s')\right]\right\}\,\mathrm{d}t\\
&=& \int_{0}^\infty f_{\mathbf{E}(s)}(t)\cdot\left[\sum_{a\in\mathrm{En}(s)} D(s,t,a)\cdot p^{s,a}_{\Pi,D}(t)\right]\,\mathrm{d}t\enskip.
\end{eqnarray*}
where the last step follows from Propostion~\ref{prop:init}.

\textbf{Inductive Step:} Suppose $\Xi\in\mathcal{S}^{n+1}_\mathcal{M}$ with $n\ge 1$. Denote
\[ 
g_{\Xi,D}(\xi):=\int_{\Gamma_{\mathcal{M}}}\mathbf{1}_\Xi(\xi\circ\gamma)~\mu^D_\mathcal{M}(\xi,\mathrm{d}\gamma)\enskip.
\]
Let $\{g_m:\Omega^n_{\mathcal{M}}\rightarrow [0,1]\}_{m\ge 0}$ be a sequence of simple functions that converges to $g_{\Xi,D}$. Denote $g_m=\sum_{i=1}^{l_m}d_m^i\cdot\mathbf{1}_{\Xi_m^i}$\enskip.   Then, we have:
\begin{eqnarray*}
& & \mathrm{Pr}_{D,\mathcal{D}[s]}\left(\mathrm{Cyl}(\Xi)\right)\\
&=& \mathrm{Pr}^{n+1}_{D,\mathcal{D}[s]}\left(\Xi\right)\\
&=& \int_{\Omega^n_\mathcal{M}}g_{\Xi,D}(\xi)~\mathrm{Pr}^{n}_{D,\mathcal{D}[s]}(\mathrm{d}\xi)  \\
&=& \lim\limits_{m\rightarrow\infty}\sum_{i=1}^{l_m} d_m^i\cdot\mathrm{Pr}^n_{D,\mathcal{D}[s]}\left(\Xi_m^i\right) \\
&=& \lim\limits_{m\rightarrow\infty}\sum_{i=1}^{l_m} d_m^i\cdot\left\{\int_0^\infty f_{\mathbf{E}(s)}(t)\cdot\left[\sum_{a\in\mathrm{En}(s) }D(s,t,a)\cdot p^{s,a}_{\mathrm{Cyl}(\Xi_m^i),D}(t)\right]\,\mathrm{d}t\right\} \\
&=& \lim\limits_{m\rightarrow\infty}\int_0^\infty f_{\mathbf{E}(s)}(t)\cdot\left\{\sum_{a\in\mathrm{En}(s) }D(s,t,a)\cdot\left[\sum_{i=1}^{l_m} d_m^i\cdot p^{s,a}_{\mathrm{Cyl}(\Xi_m^i),D}(t)\right]\right\}\,\mathrm{d}t \\
&=& \int_0^\infty f_{\mathbf{E}(s)}(t)\cdot\left\{\sum_{a\in\mathrm{En}(s) }D(s,t,a)\cdot\left[\lim\limits_{m\rightarrow\infty}\sum_{i=1}^{l_m} d_m^i\cdot p^{s,a}_{\mathrm{Cyl}(\Xi_m^i),D}(t)\right]\right\}\,\mathrm{d}t
\end{eqnarray*}
where the fourth equality is from the induction hypothesis. Note that 
\begin{eqnarray*}
& & \lim\limits_{m\rightarrow\infty}\sum_{i=1}^{l_m} d_m^i\cdot p^{s,a}_{\mathrm{Cyl}(\Xi_m^i),D}(t) \\
&=& \lim\limits_{m\rightarrow\infty}\sum_{i=1}^{l_m} d_m^i\cdot \mathrm{Pr}_{D[s\xrightarrow{a,t}],\mathbf{P}(s,a,\centerdot)}\left(P_{\mathrm{Cyl}(\Xi_m^i)}^{s,a}(t)\right) \\
&=&\lim\limits_{m\rightarrow\infty}\sum_{i=1}^{l_m}d_m^i\cdot\mathrm{Pr}^{n-1}_{D[s\xrightarrow{a,t}],\mathbf{P}(s,a,\centerdot)}\left(H^{s,a}_{\Xi^i_m}(t)\right)
\end{eqnarray*}
Denote $g'_m(\xi):=\sum_{i=1}^{l_m}d_m^i\cdot\mathbf{1}_{H^{s,a}_{\Xi^i_m}(t)}(\xi)$\enskip. 
Then  
$g'_m(\xi)=g_m(s\xrightarrow{a,t}\xi)$. It follows that $\lim\limits_{m\rightarrow\infty} g'_m(\xi)= g_{\Xi,D}(s\xrightarrow{a,t}\xi)$. By definition, 
\begin{enumerate} \itemsep1pt \parskip0pt \parsep0pt
\item $\mathbf{1}_{\Xi}\left((s\xrightarrow{a,t}\xi)\circ\gamma\right)=\mathbf{1}_{H^{s,a}_{\Xi}(t)}\left(\xi\circ\gamma\right)$ for all combined action $\gamma$;
\item $\mu^D_\mathcal{M}(s\xrightarrow{a,t}\xi,U)=\mu^{D[s\xrightarrow{a,t}]}_\mathcal{M}(\xi,U)$ for all $U\in\Gamma_\mathcal{M}$.
\end{enumerate}
Thus $g_{\Xi,D}(s\xrightarrow{a,t}\xi)=g_{H^{s,a}_{\Xi}(t),~D[s\xrightarrow{a,t}]}(\xi)=\lim\limits_{m\rightarrow\infty} g'_m(\xi)$. It follows that 
\begin{eqnarray*}
& &\lim\limits_{m\rightarrow\infty}\sum_{i=1}^{l_m}d_m^i\cdot\mathrm{Pr}^{n-1}_{D[s\xrightarrow{a,t}],\mathbf{P}(s,a,\centerdot)}\left(H^{s,a}_{\Xi^i_m}(t)\right)\\
&=&\int_{\Omega_\mathcal{M}^{n-1}}g_{H^{s,a}_{\Xi}(t),D[s\xrightarrow{a,t}]}(\xi)~\mathrm{Pr}^{n-1}_{D[s\xrightarrow{a,t}],\mathbf{P}(s,a,\centerdot)}(\mathrm{d}\xi) \\
&=&\mathrm{Pr}^n_{D[s\xrightarrow{a,t}],\mathbf{P}(s,a,\centerdot)}\left(H^{s,a}_{\Xi}(t)\right)\enskip.
\end{eqnarray*}
Then we have
\begin{eqnarray*}
& &\mathrm{Pr}_{D,\mathcal{D}[s]}\left(\mathrm{Cyl}(\Xi)\right) \\
&=&\int_0^\infty f_{\mathbf{E}(s)}(t)\cdot\left[\sum_{a\in\mathrm{En}(s)}D(s,t,a)\cdot\mathrm{Pr}^n_{D[s\xrightarrow{a,t}],\mathbf{P}(s,a,\centerdot)}\left(H^{s,a}_{\Xi}(t)\right)\right]\,\mathrm{d}t \\
&=&\int_0^\infty f_{\mathbf{E}(s)}(t)\cdot\left[\sum_{a\in\mathrm{En}(s)}D(s,t,a)\cdot p_{\mathrm{Cyl}(\Xi),D}^{s,a}(t)\right]\,\mathrm{d}t
\end{eqnarray*}
which completes the inductive step.
\end{proof}

\section{Proofs for Section~\ref{sec:mcbrp}}

\noindent\textbf{Theorem~\ref{thm:minearlyf}.}
The function $\mathrm{prob}^{\mathrm{e},\max}_G(\centerdot,\centerdot)$ is the least fixed-point (w.r.t $\le$) of the high-order operator 
$\mathcal{T}^{\mathrm{e}}_G:\left[L\times\mathbb{R}_{\ge 0}\rightarrow [0,1]\right]\rightarrow\left[L\times\mathbb{R}_{\ge 0}\rightarrow [0,1]\right]$ defined by:  
\begin{itemize}\itemsep1pt \parskip0pt \parsep0pt
\item $\mathcal{T}^{\mathrm{e}}_G(h)(s,\mathbf{c}):=\mathbf{1}_{\mathbb{R}^k_{\ge 0}}(\mathbf{c})$ for all $s\in G$ and $\mathbf{c}\in\mathbb{R}^k$;
\item for all $s\in L-G$ and $\mathbf{c}\in\mathbb{R}^k$, 
\begin{align*}
& \mathcal{T}^\mathrm{e}_G(h)(s,\mathbf{c}):=\\
& \quad\max_{a\in\mathrm{En}(s)}\int_{0}^\infty f_{\mathbf{E}(s,a)}(t)\cdot\left[\sum_{s'\in L}\mathbf{P}(s,a,s')\cdot h(s',\mathbf{c}-t\cdot\mathbf{w}(s,a))\right]\,\mathrm{d}t\enskip.
\end{align*}
\end{itemize}
Moreover, 
\[
\left|\mathrm{prob}^{\mathrm{e},\max}_G(s,\mathbf{c})-\mathrm{prob}^{\mathrm{e},\max}_G(s,\mathbf{c}')\right|\le \frac{\mathbf{E}_{\max}}{\mathbf{w}_{\min}}\cdot{\parallel}{\mathbf{c}-\mathbf{c}'}{\parallel}_\infty
\]
given any $\mathbf{c},\mathbf{c}'\ge \vec{0}$ and $s\in L$\enskip.
\begin{proof}
Define the function $\mathrm{prob}^{\mathrm{e},\max}_{n,G}:L\times\mathbb{R}^k_{\ge 0}\rightarrow [0,1]$ by 
\[
\mathrm{prob}^{\mathrm{e},\max}_{n,G}(s,\mathbf{c}):=\sup_{D\in\mathcal{E}(\mathcal{M})}\mathrm{Pr}_{D,\mathcal{D}[s]}\left(\Pi_{n,G}^\mathbf{c}\right)
\]
where 
\[
\Pi_{n,G}^\mathbf{c}:=\{\pi\in\mathrm{Paths}(\mathcal{M})\mid\mathbf{C}(\pi,G)\le\mathbf{c}\mbox{ and }\pi[m]\in G\mbox{ for some }0\le m\le n\}~~.
\]
Intuitively, $\mathrm{prob}^{\mathrm{e},\max}_{n,G}$ is the maximal cost-bounded reachability probability function within $n$ steps. 
For $n\in\mathbb{N}_{\ge 0}$ and $\delta>0$, define
\begin{align*}
& \epsilon(n,\delta):=\sup\bigg\{\left|\mathrm{prob}^{\mathrm{e},\max}_{n,G}(s,\mathbf{c})-\mathrm{prob}^{\mathrm{e},\max}_{n,G}(s,\mathbf{c}')\right|\mid \\
& \qquad\qquad s\in L, \mathbf{c},\mathbf{c}'\ge \vec{0}\mbox{ and }\parallel\mathbf{c}-\mathbf{c'}\parallel_\infty\le\delta\bigg\}\enskip.
\end{align*}
Note that $\epsilon(0,\delta)=0$ for all $\delta$. Firstly, we prove by induction on $n\ge 0$ that the following assertions hold: 
\begin{enumerate}\itemsep1pt \parskip0pt \parsep0pt
\item[(a)] $\mathrm{prob}^{\mathrm{e},\max}_{n,G}(s,\mathbf{c})=\mathbf{1}_{\mathbb{R}^k_{\ge 0}}(\mathbf{c})$ if $s\in G$\enskip;
\item[(b)] Given any $\delta>0$, $\epsilon(n,\delta)\le \frac{\mathbf{E}_{\max}}{\mathbf{w}_{\min}}\cdot\delta$\enskip;
\item[(c)] If $n>0$ and $s\in L-G$ then $\mathrm{prob}^{\mathrm{e},\max}_{n+1,G}=\mathcal{T}^\mathrm{e}_G(\mathrm{prob}^{\mathrm{e},\max}_{n,G})$\enskip.
\end{enumerate}
The base step when $n=0$ is easy. It is clear that 
$\mathrm{Pr}_{D, \mathcal{D}[s]}(\Pi^\mathbf{c}_{0,G})=\mathbf{1}_G(s)\cdot\mathbf{1}_{\mathbb{R}^k_{\ge 0}}(\mathbf{c})$ for all 
measurable early scheduler $D$. Thus (a), (b) holds and (c) is vacuum true. For the inductive step, let $n=m+1$ with $m\ge 0$. It is easy 
to see that (a) holds. Below we prove (b) and (c). Let $\mathbf{c}\ge \vec{0}$ and $s\in L-G$. By Proposition~\ref{thm:fix-early} and 
Corollary~\ref{crlly:int_early}, for all measurable early scheduler $D$, 
\begin{align*}
& \mathrm{Pr}_{D,\mathcal{D}[s]}(\Pi^\mathbf{c}_{m+1,G})=\sum_{a\in\mathrm{En}(s)}D(s,a)\cdot\\
& ~~~~\int_{0}^\infty f_{\mathbf{E}(s,a)}(t)\cdot\left[\sum_{s'\in L}\mathbf{P}(s,a,s')\cdot\mathrm{Pr}_{D[s\xrightarrow{a,t}],\mathcal{D}[s']}\left(\Pi_{m,G}^{\mathbf{c}-t\cdot\mathbf{w}(s,a)}\right)\right]\,\mathrm{d}t\enskip.
\end{align*}
If we modify $D$ to $D'$ by setting $D'(s,\centerdot)$ to
\[
\mathcal{D}\left[\argmax_{a\in\mathrm{En}(s)}\int_0^\infty f_{\mathbf{E}(s,a)}(t)\cdot\left(\sum_{s'\in L}\mathbf{P}(s,a,s')\cdot\mathrm{Pr}_{D[s\xrightarrow{a,t}],\mathcal{D}[s']}\left(\Pi_{m,G}^{\mathbf{c}-t\cdot\mathbf{w}(s,a)}\right)\right)\,\mathrm{d}t\right]
\]
and $D'(\xi,\centerdot)=D(\xi,\centerdot)$ for $\xi\ne s$, then $D'$ is still a measurable scheduler which satisfies that 
$\mathrm{Pr}_{D,\mathcal{D}[s]}(\Pi^\mathbf{c}_{m+1,G})\le \mathrm{Pr}_{D',\mathcal{D}[s]}(\Pi^\mathbf{c}_{m+1,G})$\enskip. Thus we have
\begin{align*}
& \mathrm{prob}^{\max}_{m+1,G}(s,\mathbf{c})=\\
& ~\sup_{D\in\mathcal{E}_\mathcal{M}}\max_{a\in\mathrm{En}(s)}\int_{0}^\infty f_{\mathbf{E}(s,a)}(t)\cdot\left[\sum_{s'\in L}\mathbf{P}(s,a,s')\cdot\mathrm{Pr}_{D[s\xrightarrow{a,t}],\mathcal{D}[s']}\left(\Pi_{m,G}^{\mathbf{c}-t\cdot\mathbf{w}(s,a)}\right)\right]\,\mathrm{d}t~.
\end{align*}
Below we prove the inductive step for (c). In detail, we prove that the value $\mathrm{prob}^{\max}_{m+1,G}(s,\mathbf{c})=\sup_{D\in\mathcal{E}_\mathcal{M}}\max_{a\in\mathrm{En}(s)}\mathcal{F}(D,a,\mathbf{c})$ with 
\[
\mathcal{F}(D,a,\mathbf{c}):=\int_{0}^\infty f_{\mathbf{E}(s,a)}(t)\cdot\left[\sum_{s'\in L}\mathbf{P}(s,a,s')\cdot\mathrm{Pr}_{D[s\xrightarrow{a,t}],\mathcal{D}[s']}\left(\Pi_{m,G}^{\mathbf{c}-t\cdot\mathbf{w}(s,a)}\right)\right]\,\mathrm{d}t
\]
equals $\max_{a\in\mathrm{En}(s)}\mathcal{G}(a,\mathbf{c})$
\[
\mathcal{G}(a,\mathbf{c}):=\int_{0}^\infty f_{\mathbf{E}(s,a)}(t)\cdot\left[\sum_{s'\in L}\mathbf{P}(s,a,s')\cdot\mathrm{prob}^{\mathrm{e},\max}_{m,G}(s',\mathbf{c}-t\cdot\mathbf{w}(s,a))\right]\,\mathrm{d}t~.
\]
It is not difficult to see that the former is no greater than the latter. Below we prove the reverse direction. Denote 
$a^*:=\argmax_{a\in\mathrm{En}(s)}\mathcal{G}(a,\mathbf{c})$\enskip. We clarity two cases below.

~\\
\noindent\textbf{Case 1: } $\mathbf{w}(s,a^*)=\vec{0}$. For all $\epsilon>0$, we can choose the measurable scheduler $D^\epsilon$ such 
that (i) $D^\epsilon(s,\centerdot)=\mathcal{D}[a^*]$ and (ii) 
$D^\epsilon(s\xrightarrow{a^*,t}\xi,\centerdot)=D_{\xi[0]}^{m,\epsilon}(\xi,\centerdot)$ for any 
$s\xrightarrow{a^*,t}\xi\in Hists(\mathcal{M})$, where $D_{\xi[0]}^{m,\epsilon}$ is a measurable scheduler such that 
\[
\mathrm{Pr}_{D_{\xi[0]}^{m,\epsilon},\mathcal{D}[\xi[0]]}\left(\Pi^\mathbf{c}_{m,G}\right)\ge\mathrm{prob}^{\mathrm{e},\max}_{m,G}(\xi[0],\mathbf{c})-\epsilon\enskip.
\]
The probability distribution $D(\xi,\centerdot)$ for all other $\xi\in Hists(\mathcal{M})$ is irrelevant and can be set to an arbitrary 
canonical distribution. It is not hard to verify that $D^\epsilon$ satisfies that 
$\max_{a\in\mathrm{En}(s)}\mathcal{F}(D^\epsilon,a,\mathbf{c})\ge \mathcal{G}(a^*)-\epsilon$. Thus 
$\mathrm{prob}^{\max}_{m+1,G}(s,\mathbf{c})=\mathcal{G}(a^*,\mathbf{c})$ by the arbitrary choice of $\epsilon$.

~\\
\noindent\textbf{Case 2: } $\mathbf{w}(s,a^*)\ne\vec{0}$. Then the integrand function of $\mathcal{G}(a^*,\mathbf{c})$ takes non-zero value only on $[0,T_{s,a^*}^\mathbf{c}]$ with 
\[
T_{s,a^*}^\mathbf{c}:=\min\left\{\frac{\mathbf{c}_i}{\mathbf{w}_i(s,a^*)}\mid 1\le i\le k, \mathbf{w}_i(s,a^*)>0)\right\}\enskip. 
\]
By induction hypothesis (b), the integrand of $\mathcal{G}(a^*,\mathbf{c})$ is Lipschitz continuous on $[0,T_{s,a^*}^\mathbf{c}]$. For 
each $N\in\mathbb{N}$, we divide the interval $[0,T_{s,a^*}^\mathbf{c})$ into $N$ equal pieces $I_1,\dots, I_N$ with 
$I_j=[\frac{j-1}{N}\cdot T_{s,a^*}^\mathbf{c}, \frac{j}{N}\cdot T_{s,a^*}^\mathbf{c})$. Moreover, we choose some $t_j\in I_j$ arbitrarily. 
Then we construct the measurable scheduler $D^{N,\epsilon}$ as follows: 
\[
D^{N,\epsilon}(s,\centerdot)=\mathcal{D}[a^*]\mbox{ and }D^{N,\epsilon}(s\xrightarrow{a,t}\xi,\centerdot)=D^{m,\epsilon}_{\xi[0],j}(\xi,\centerdot)\mbox{ when }t\in I_j\enskip, 
\]
where $D^{m,\epsilon}_{\xi[0],j}$ is a measurable early scheduler such that 
\[
\mathrm{Pr}_{D^{m,\epsilon}_{\xi[0],j},\mathcal{D}[\xi[0]]}\left(\Pi^{\mathbf{c}-t_j\cdot\mathbf{w}(s,a^*)}_{m,G}\right)\ge \mathrm{prob}^{\mathrm{e},\max}_{m,G}(\xi[0],\mathbf{c}-t_j\cdot\mathbf{w}(s,a^*))-\epsilon\enskip.
\]
$D^{N,\epsilon}(\xi,\centerdot)$ for all other $\xi$ is irrelevant. By induction hypothesis (b), we have 
$\lim_{N\rightarrow\infty,\epsilon\rightarrow 0^+}\mathcal{F}(D^{N,\epsilon},a^*,\mathbf{c})=G(a^*,\mathbf{c})$. Then 
$\mathrm{prob}^{\max}_{m+1,G}(s,\mathbf{c})=\mathcal{G}(a^*,\mathbf{c})$ since $\epsilon$ can be arbitrarily chosen.

It remains to show that the inductive step for (b) holds. Let $\mathbf{c},\mathbf{c}'\ge\vec{0}$ and $s\in L-G$. Denote 
$\delta:={\parallel}{\mathbf{c}-\mathbf{c}'}{\parallel}_\infty$\enskip. Consider an arbitrary $a\in\mathrm{En}(s)$. If 
$\mathbf{w}(s,a)=\vec{0}$, then clearly $\left|\mathcal{G}(a,\mathbf{c})-\mathcal{G}(a,\mathbf{c}')\right|\le\epsilon(m,\delta)$\enskip.
Otherwise, we have 
\begin{eqnarray*}
& & \left|\mathcal{G}(a,\mathbf{c})-\mathcal{G}(a,\mathbf{c}')\right|\\
& \le & \int_0^{T}f_{\mathbf{E}(s,a)}(t)\cdot\epsilon(m,\delta)\,\mathrm{d}t+f_{\mathbf{E}(s,a)}(T)\cdot\left|T_{s,a}^{\mathbf{c}}-T_{s,a}^{\mathbf{c}'}\right|\\
& \le & (1-e^{-\mathbf{E}(s,a)\cdot T})\cdot \epsilon(m,\delta)+ e^{-\mathbf{E}(s,a)\cdot T}\cdot \frac{\mathbf{E}_{\max}}{\mathbf{w}_{\min}}\cdot\delta\\
&\le &\frac{\mathbf{E}_{\max}}{\mathbf{w}_{\min}}\cdot\delta\enskip
\end{eqnarray*}
where $T:=\min\{T_{s,a}^{\mathbf{c}},T_{s,a}^{\mathbf{c}'}\}$ and the last step is obtained through induction hypothesis. It follows that
\begin{eqnarray*}
& & \left|\mathrm{prob}^{\mathrm{e},\max}_{m+1,G}(s,\mathbf{c})-\mathrm{prob}^{\mathrm{e},\max}_{m+1,G}(s,\mathbf{c}')\right| \\
& \le & \max_{a\in\mathrm{En}(s)}\left|\mathcal{G}(a,\mathbf{c}-\mathcal{G}(a,\mathbf{c}')\right|\\
& \le & \frac{\mathbf{E}_{\max}}{\mathbf{w}_{\min}}\cdot\delta\enskip.
\end{eqnarray*}
Then the inductive step for (b) is completed. 

Secondly, we prove that $\lim\limits_{n\rightarrow\infty}\mathrm{prob}^{\mathrm{e},\max}_{n,G}=\mathrm{prob}^{\mathrm{e},\max}_G$. Let 
$s\in L$ and $\mathbf{c}\in\mathbb{R}^k$. By definition, $\mathrm{prob}^{\mathrm{e},\max}_{n,G}\le\mathrm{prob}^{\mathrm{e},\max}_{n+1,G}$ 
and $\mathrm{prob}^{\mathrm{e},\max}_{n,G}\le\mathrm{prob}^{\mathrm{e},\max}_G$ for all $n\ge 0$. Thus 
$\lim\limits_{n\rightarrow\infty}\mathrm{prob}^{\mathrm{e},\max}_{n,G}$ exists and is no greater than $\mathrm{prob}^{\mathrm{e},\max}_G$. 
For the reverse direction, fix an arbitrary $\epsilon>0$. Let $D$ be a measurable early scheduler such that 
$\mathrm{Pr}_{D,\mathcal{D}[s]}(\Pi_{G}^\mathbf{c})\ge\mathrm{prob}^{\mathrm{e},\max}_G(s,\mathbf{c})-\epsilon$. By definition, 
$\mathrm{prob}_{n,G}^{\mathrm{e},\max}(s,\mathbf{c})\ge\mathrm{Pr}_{D,\mathcal{D}[s]}(\Pi_{n,G}^\mathbf{c})$. It follows that
\[ 
\lim\limits_{n\rightarrow\infty}\mathrm{prob}_{n,G}^{\mathrm{e},\max}(s,\mathbf{c})\ge\lim\limits_{n\rightarrow\infty}\mathrm{Pr}_{D,\mathcal{D}[s]}(\Pi_{n,G}^\mathbf{c})=\mathrm{Pr}_{D,\mathcal{D}[s]}(\Pi_{G}^\mathbf{c})\ge \mathrm{prob}^{\mathrm{e},\max}_G(s,\mathbf{c})-\epsilon\enskip.
\]
Thus $\lim\limits_{n\rightarrow\infty}\mathrm{prob}_{n,G}^{\mathrm{e},\max}(s,\mathbf{c})=\mathrm{prob}^{\mathrm{e},\max}_G$ since 
$\epsilon$ is arbitrarily chosen.

Thirdly, we prove that $\mathrm{prob}^{\mathrm{e},\max}_G$ is the least fixed point of $\mathcal{T}^{\mathrm{e}}_G$. It is clear that 
$\mathrm{prob}^{\mathrm{e},\max}_G(s,\mathbf{c})=\mathbf{1}_{\mathbb{R}_{\ge 0}^k}(\mathbf{c})$ if $s\in G$. By applying Monotone 
Convergence Theorem on (c), we obtain 
$\mathrm{prob}^{\mathrm{e},\max}_G(s,\mathbf{c})=\mathcal{T}^\mathrm{e}_G(\mathrm{prob}^{\mathrm{e},\max}_G)(s,\mathbf{c})$. Thus, 
$\mathrm{prob}^{\mathrm{e},\max}_G$ is a fixed-point of $\mathcal{T}^\mathrm{e}_G$. To see that it is the least fixed-point of 
$\mathcal{T}^\mathrm{e}_G$, one can proceed by induction on $n\ge 0$ that given any fixed-point $h$ of $\mathcal{T}^\mathrm{e}_G$, 
$\mathrm{prob}^{\mathrm{e},\max}_{n,G}\le h$ for all $n\ge 0$ by the facts that $\mathrm{prob}^{\mathrm{e},\max}_{0,G}\le h$ and 
$\mathrm{prob}^{\mathrm{e},\max}_{n+1,G}=\mathcal{T}^\mathrm{e}_G(\mathrm{prob}^{\mathrm{e},\max}_{n,G})$. It follows that 
$\mathrm{prob}^{\mathrm{e},\max}_{G}\le h$ for any fixed-point $h$ of $\mathcal{T}^\mathrm{e}_G$\enskip.

Finally, by taking the limit from (b) we can obtain that
\[
\left|\mathrm{prob}^{\mathrm{e},\max}_G(s,\mathbf{c})-\mathrm{prob}^{\mathrm{e},\max}_G(s,\mathbf{c}')\right|\le \frac{\mathbf{E}_{\max}}{\mathbf{w}_{\min}}\cdot {\parallel}{\mathbf{c}-\mathbf{c}'}{\parallel}_\infty
\]
for all $\mathbf{c},\mathbf{c}'\ge \vec{0}$ and $s\in L$\enskip.
\end{proof}

\noindent\textbf{Theorem~\ref{thm:minlatef}.}
The function $\mathrm{prob}^{\mathrm{l},\max}_G$ is the least fixed point of the high-order operator $\mathcal{T}^{\mathrm{l}}_G:\left[L\times\mathbb{R}^k\rightarrow [0,1]\right]\rightarrow\left[L\times\mathbb{R}^k\rightarrow[0,1]\right]$ defined as follows:  
\begin{itemize}\itemsep1pt \parskip0pt \parsep0pt
\item $\mathcal{T}^\mathrm{l}_G(h)(s,\mathbf{c}):=\mathbf{1}_{\mathbb{R}^k_{\ge 0}}(\mathbf{c})$ if $s\in G$;
\item If $s\not\in G$ then 
\begin{align*}
& \mathcal{T}^\mathrm{l}_G(h)(s,\mathbf{c}):=\\
& \quad\int_{0}^\infty f_{\mathbf{E}(s)}(t)\cdot\max_{a\in\mathrm{En}(s)}\left[\sum_{s'\in L}\mathbf{P}(s,a,s')\cdot h(s',\mathbf{c}-t\cdot\mathbf{w}(s))\right]\,\mathrm{d}t\enskip.
\end{align*}
\end{itemize}
Moreover, 
\[
\left|\mathrm{prob}^{\mathrm{l},\max}_G(s,\mathbf{c})-\mathrm{prob}^{\mathrm{l},\max}_G(s,\mathbf{c}')\right|\le\frac{\mathbf{E}_{\max}}{\mathbf{w}_{\min}}\cdot{\parallel}{\mathbf{c}-\mathbf{c}'}{\parallel}_\infty
\]
for all $\mathbf{c},\mathbf{c}'\ge\vec{0}$ and $s\in L$\enskip.
\begin{proof}
Define the function $\mathrm{prob}^{\mathrm{l},\max}_{n,G}:L\times\mathbb{R}^k_{\ge 0}\rightarrow [0,1]$ by
\[
\mathrm{prob}^{\mathrm{l},\max}_{n,G}(s,\mathbf{c}):=\sup_{D\in\mathcal{L}_\mathcal{M}}\mathrm{Pr}_{D,\mathcal{D}[s]}\left(\Pi_{n,G}^\mathbf{c}\right)
\]
where 
\[
\Pi_{n,G}^\mathbf{c}:=\{\pi\in\mathrm{Paths}(\mathcal{M})\mid\mathbf{C}(\pi,G)\le\mathbf{c}\mbox{ and }\pi[m]\in G\mbox{ for some }0\le m\le n\}~~.
\]
For $n\in\mathbb{N}_{\ge 0}$ and $\delta>0$, define
\begin{align*}
\epsilon(n,\delta):=\sup\bigg\{\left|\mathrm{prob}^{\mathrm{l},\max}_{n,G}(s,\mathbf{c})-\mathrm{prob}^{\mathrm{l},\max}_{n,G}(s,\mathbf{c}')\right|\mid \\
s\in L, \mathbf{c},\mathbf{c}'\ge\vec{0}\mbox{ and }\parallel\mathbf{c}-\mathbf{c'}\parallel_\infty\le\delta\bigg\}\enskip.
\end{align*}
Firstly, we prove by induction on $n\ge 0$ that the following conditions hold: 
\begin{enumerate}\itemsep1pt \parskip0pt \parsep0pt
\item[(a)] $\mathrm{prob}^{\mathrm{l},\max}_{n,G}(s,\mathbf{c})=\mathbf{1}_{\mathbb{R}_{\ge 0}^k}(\mathbf{c})$ if $s\in G$;
\item[(b)] Given any $c\in\mathbb{R}_{\ge 0}$ and $\delta>0$, $\epsilon(n,\delta)\le \frac{\mathbf{E}_{\max}}{\mathbf{w}_{\min}}\cdot\delta$\enskip;
\item[(c)] If $n>0$ and $s\in L-G$ then $\mathrm{prob}^{\mathrm{l},\max}_{n+1,G}=\mathcal{T}^\mathrm{l}_G(\mathrm{prob}^{\mathrm{l},\max}_{n,G})$\enskip.
\end{enumerate}
The base step when $n=0$ is easy. It is clear that $\mathrm{Pr}_{D, \mathcal{D}[s]}(\Pi^\mathbf{c}_{n,G})=\mathbf{1}_G(s)$ for all measurable late scheduler $D$ and $\mathbf{c}\ge\vec{0}$. Thus (a), (b) holds and (c) is vacuum true. For the inductive step, let $n=m+1$ with $m\ge 0$. It is easy to see that (a) holds. Below we prove (b) and (c). Let $\mathbf{c}\ge\vec{0}$ and $s\in L-G$. Firstly, we prove (c). By Proposition~\ref{thm:fix-late},  
\begin{align*}
& \mathrm{Pr}_{D,\mathcal{D}[s]}(\Pi^\mathbf{c}_{m+1,G})=\\
& \int_{0}^\infty f_{\mathbf{E}(s)}(t)\cdot\sum_{a\in\mathrm{En}(s)}D(s,t,a)\cdot\left[\sum_{s'\in L}\mathbf{P}(s,a,s')\cdot\mathrm{Pr}_{D[s\xrightarrow{a,t}],\mathcal{D}[s']}\left(\Pi_{m,G}^{\mathbf{c}-t\cdot\mathbf{w}(s)}\right)\right]\,\mathrm{d}t
\end{align*}
for all measurable late sheduler $D$. Note that for each $a\in\mathrm{En}(s)$, the function
\[
t\mapsto\sum_{s'\in L}\mathbf{P}(s,a,s')\cdot\mathrm{Pr}_{D[s\xrightarrow{a,t}],\mathcal{D}[s']}\left(\Pi_{m,G}^{\mathbf{c}-t\cdot\mathbf{w}(s)}\right)
\]
is measurable w.r.t $(\mathbb{R}_{\ge 0},\mathcal{B}(\mathbb{R}_{\ge 0}))$ due to Proposition~\ref{prop:init} and Proposition~\ref{prop:shiftfunc:measurability}. Thus, if we modify $D$ to $D'$ by setting (i) $D'(s,t,\centerdot)$ to
\[
\mathcal{D}\left[\argmax_{a\in\mathrm{En}(s)}\sum_{s'\in L}\mathbf{P}(s,a,s')\cdot\mathrm{Pr}_{D[s\xrightarrow{a,t}],\mathcal{D}[s']}\left(\Pi_{m,G}^{\mathbf{c}-t\cdot\mathbf{w}(s)}\right)\right]
\]
for each $t\ge 0$ and (ii) $D'(\xi,\centerdot,\centerdot)=D(\xi,\centerdot,\centerdot)$ for $\xi\ne s$, then $D'$ is still a measurable late scheduler such that 
$\mathrm{Pr}_{D,\mathcal{D}[s]}(\Pi^\mathbf{c}_{m+1,G})\le \mathrm{Pr}_{D',\mathcal{D}[s]}(\Pi^\mathbf{c}_{m+1,G})$\enskip. Then we have
\begin{align*}
& \mathrm{prob}^{\mathrm{l},\max}_{m+1,G}(s,\mathbf{c})=\\
& ~\sup_{D\in\mathcal{L}_\mathcal{M}}\int_{0}^\infty f_{\mathbf{E}(s)}(t)\cdot\max_{a\in\mathrm{En}(s)}\left[\sum_{s'\in L}\mathbf{P}(s,a,s')\cdot\mathrm{Pr}_{D[s\xrightarrow{a,t}],\mathcal{D}[s']}\left(\Pi_{m,G}^{\mathbf{c}-t\cdot\mathbf{w}(s)}\right)\right]\,\mathrm{d}t~.
\end{align*}
We prove that $\sup_{D\in\mathcal{L}_\mathcal{M}}\int_0^\infty f_{\mathbf{E}(s)}(t)\cdot\max_{a\in\mathrm{En}(s)}\mathcal{F}(D,a,t,\mathbf{c})\,\mathrm{d}t$ with
\[
\mathcal{F}(D,a,t,\mathbf{c}):=\sum_{s'\in L}\mathbf{P}(s,a,s')\cdot\mathrm{Pr}_{D[s\xrightarrow{a,t}],\mathcal{D}[s']}\left(\Pi_{m,G}^{\mathbf{c}-t\cdot\mathbf{w}(s)}\right)\,\mathrm{d}t
\]
(essentially $\mathrm{prob}^{\mathrm{l},\max}_{m+1,G}(s,\mathbf{c})$) equals $\int_0^\infty f_{\mathbf{E}(s)}(t)\cdot\max_{a\in\mathrm{En}(s)}\mathcal{G}(a,t,\mathbf{c})\,\mathrm{d}t$ with
\[
\mathcal{G}(a,t,\mathbf{c}):=\sum_{s'\in L}\mathbf{P}(s,a,s')\cdot\mathrm{prob}^{\mathrm{l},\max}_{m,G}(s',\mathbf{c}-t\cdot\mathbf{w}(s))\enskip.
\]
It is not difficult to see that the former item is no greater than the latter one. Below we prove the reverse direction. Define
$a^*(t):=\argmax_{a\in\mathrm{En}(s)}\mathcal{G}(a,t,\mathbf{c})$\enskip. We consider two cases. 

~\\
\noindent\textbf{Case 1:} $\mathbf{w}(s)=\vec{0}$. Then $a^*(t)$ is independent of $t$, which we shall denote by $a^*$. For each 
$\epsilon>0$, we choose the measurable late scheduler $D^\epsilon$ such that (i) $D^\epsilon(s,t,\centerdot)=\mathcal{D}[a^*]$ for all 
$t\ge 0$, and (ii) $D^\epsilon(s\xrightarrow{a^*,t}\xi,\tau,\centerdot)=D_{\xi[0]}^{m,\epsilon}(\xi,\tau,\centerdot)$ for all 
$s\xrightarrow{a^*,t}\xi\in Hists(\mathcal{M})$ and $\tau\ge 0$, where $D_{\xi[0]}^{m,\epsilon}$ is a measurable late scheduler such that 
\[
\mathrm{Pr}_{D_{\xi[0]}^{m,\epsilon},\mathcal{D}[\xi[0]]}\left(\Pi^\mathbf{c}_{m,G}\right)\ge\mathrm{prob}^{\mathrm{l},\max}_{m,G}(\xi[0],\mathbf{c})-\epsilon\enskip.
\]
$D^\epsilon$ satisfies that $\max_{a\in\mathrm{En}(s)}\mathcal{F}(D^\epsilon,a,t,\mathbf{c})\ge \mathcal{G}(a^*,t,\mathbf{c})-\epsilon$ for all $t\ge 0$. Thus 
\[
\mathrm{prob}^{\mathrm{l},\max}_{m+1,G}(s,\mathbf{c})=\int_0^\infty f_{\mathbf{E}(s)}(t)\cdot\mathcal{G}(a^*,t,\mathbf{c})\,\mathrm{d}t 
\]
since $\epsilon$ can be arbitrarily chosen. 

~\\
\noindent\textbf{Case 2:} $\mathbf{w}(s)\ne \vec{0}$. Then for each $a\in\mathrm{En}(s)$, $\mathcal{G}(a,t,\mathbf{c})$ takes non-zero value only on $t\in[0,T_{s}^\mathbf{c}]$ with 
\[
T_{s}^\mathbf{c}:=\min\left\{\frac{\mathbf{c}_i}{\mathbf{w}_i(s)}\mid 1\le i\le k, \mathbf{w}_i(s)>0\right\}\enskip.
\]
From the induction hypothesis (b), $\max_{a\in\mathrm{En}(s)}\mathcal{G}(a,\centerdot,\mathbf{c})$ is Lipschitz continuous on 
$[0,T_{s}^\mathbf{c}]$. For each $N\in\mathbb{N}$, we divide the interval $[0,T_{s}^\mathbf{c})$ into $N$ equal pieces $I_1,\dots I_N$ 
with $I_j=[\frac{j-1}{N}\cdot T_{s}^\mathbf{c}, \frac{j}{N}\cdot T_{s}^\mathbf{c})$. And we choose some $t_j\in I_j$ arbitrarily for each 
$1\le j\le N$. Then we construct the measurable late scheduler $D^{N,\epsilon}$ for each $\epsilon>0$ as follows: 
\[
D^{N,\epsilon}(s,t,\centerdot)=\mathcal{D}[a^*(t_j)]\mbox{ and }D^{N,\epsilon}(s\xrightarrow{a(t_j),t}\xi,\centerdot,\centerdot)=D^{m,\epsilon}_{\xi[0],j}(\xi,\centerdot,\centerdot)\mbox{ when }t\in I_j, 
\]
where $D^{m,\epsilon}_{\xi[0],j}$ is a measurable late scheduler such that 
\[
\mathrm{Pr}_{D^{m,\epsilon}_{\xi[0],j},\mathcal{D}[\xi[0]]}\left(\Pi^{\mathbf{c}-t_j\cdot\mathbf{w}(s)}_{m,G}\right)\ge \mathrm{prob}^{\mathrm{l},\max}_{m,G}(\xi[0],\mathbf{c}-t_j\cdot\mathbf{w}(s))-\epsilon\enskip.
\]
From the Lipschitz continuity, we obtain
\[
\lim_{\stackrel{N\rightarrow +\infty}{\epsilon\rightarrow 0^+}}\int_0^\infty f_{\mathbf{E}(s)}(t)\cdot \max_{a\in\mathrm{En}(s)}\mathcal{F}(D^{N,\epsilon},a,t,\mathbf{c})\mathrm{d}t=\int_0^\infty f_{\mathbf{E}(s)}(t)\cdot\max_{a\in\mathrm{En}(s)} \mathcal{G}(a,t,\mathbf{c})\,\mathrm{d}t
\]
which implies the inductive step for (c). 

It remains to show that the inductive step for (b) holds. Let $s\in L-G$. Fix some $\mathbf{c}',\mathbf{c}\ge\vec{0}$ and 
$a\in\mathrm{En}(s)$. Denote $\delta:={\parallel}{\mathbf{c}-\mathbf{c}'}{\parallel}_\infty$\enskip. If $\mathbf{w}(s)=\vec{0}$, then clearly 
$\left|\mathcal{G}(a,t,\mathbf{c})-\mathcal{G}(a,t,\mathbf{c}')\right|\le\epsilon(m,\delta)$ for all $t\ge 0$. It follows that 
$\left|\mathrm{prob}^{\max}_{m+1,G}(s,\mathbf{c})-\mathrm{prob}^{\max}_{m+1,G}(s,\mathbf{c}')\right|\le \epsilon(m,\delta)$\enskip, which 
implies the result. Otherwise, by (c) and the induction hypothesis (b), we have
\begin{eqnarray*}
& & \left|\mathrm{prob}^{\mathrm{l},\max}_{m+1,G}(s,\mathbf{c})-\mathrm{prob}^{\mathrm{l},\max}_{m+1,G}(s,\mathbf{c}')\right|\\
&\le & \int_0^{T}f_{\mathbf{E}(s)}(t)\cdot\epsilon(m,\delta)\,\mathrm{d}t+f_{\mathbf{E}(s)}(T)\cdot\left|T_{s}^{\mathbf{c}}-T_{s}^{\mathbf{c}'}\right|\\
& \le & (1-e^{-\mathbf{E}(s)\cdot T})\cdot\epsilon(m,\delta)+e^{-\mathbf{E}(s)\cdot T}\cdot\frac{\mathbf{E}_{\max}}{\mathbf{w}_{\min}}\cdot\delta\\
&\le &\frac{\mathbf{E}_{\max}}{\mathbf{w}_{\min}}\cdot\delta
\end{eqnarray*}
where $T:=\min\{T_{s}^{\mathbf{c}},T_{s}^{\mathbf{c}'}\}$\enskip.  Then the inductive step for (b) is completed. 

Secondly, we prove that $\lim\limits_{n\rightarrow\infty}\mathrm{prob}^{\mathrm{l},\max}_{n,G}=\mathrm{prob}^{\mathrm{l},\max}_G$. Let 
$s\in L$ and $\mathbf{c}\in\mathbb{R}^k$. By definition, $\mathrm{prob}^{\mathrm{l},\max}_{n,G}\le\mathrm{prob}^{\mathrm{l},\max}_{n+1,G}$ 
and $\mathrm{prob}^{\mathrm{l},\max}_{n,G}\le\mathrm{prob}^{\mathrm{l},\max}_G$ for all $n\ge 0$. Thus 
$\lim\limits_{n\rightarrow\infty}\mathrm{prob}^{\mathrm{l},\max}_{n,G}$ exists and is no greater than $\mathrm{prob}^{\mathrm{l},\max}_G$. 
For the reverse direction, fix an arbitrary $\epsilon>0$. Let $D$ be a measurable late scheduler such that 
$\mathrm{Pr}_{D,\mathcal{D}[s]}(\Pi_{G}^\mathbf{c})\ge\mathrm{prob}^{\mathrm{l},\max}_G(s,\mathbf{c})-\epsilon$. By definition, 
$\mathrm{prob}_{n,G}^{\mathrm{l},\max}(s.\mathbf{c})\ge\mathrm{Pr}_{D,\mathcal{D}[s]}(\Pi_{n,G}^\mathbf{c})$. It follows that
\[ 
\lim\limits_{n\rightarrow\infty}\mathrm{prob}_{n,G}^{\mathrm{l},\max}(s,\mathbf{c})\ge\lim\limits_{n\rightarrow\infty}\mathrm{Pr}_{D,\mathcal{D}[s]}(\Pi_{n,G}^\mathbf{c})=\mathrm{Pr}_{D,\mathcal{D}[s]}(\Pi_{G}^\mathbf{c})\ge \mathrm{prob}^{\mathrm{l},\max}_G(s,\mathbf{c})-\epsilon
\]
. Thus $\lim\limits_{n\rightarrow\infty}\mathrm{prob}_{n,G}^{\mathrm{l},\max}(s,\mathbf{c})=\mathrm{prob}^{\mathrm{l},\max}_G(s,\mathbf{c})$ since $\epsilon$ is arbitrarily chosen.

Now we prove that $\mathrm{prob}^{\mathrm{l},\max}_G$ is the least fixed point of $\mathcal{T}^{\mathrm{l}}_G$. It is clear that
$\mathrm{prob}^{\mathrm{l},\max}_G(s,\mathbf{c})=\mathbf{1}_{\mathbb{R}_{\ge 0}^k}(\mathbf{c})$ if $s\in G$. By Monotone Convergence 
Theorem, we can obtain 
$\mathrm{prob}^{\mathrm{l},\max}_G=\mathcal{T}^\mathrm{l}_G(\mathrm{prob}^{\mathrm{l},\max}_G)$ from 
$\mathrm{prob}^{\mathrm{l},\max}_{n+1,G}=\mathcal{T}^{\mathrm{l}}_G(\mathrm{prob}^{\mathrm{l},\max}_{n,G})$ 
and $\lim\limits_{n\rightarrow\infty}\mathrm{prob}^{\mathrm{l},\max}_{n,G}=\mathrm{prob}^{\mathrm{l},\max}_G$. Thus, 
$\mathrm{prob}^{\mathrm{l},\max}_G$ is a fixed-point of $\mathcal{T}^\mathrm{l}_G$. To see that it is the least fixed-point of 
$\mathcal{T}^{\mathrm{l}}_G$, one can proceed by induction on $n\ge 0$ that given any fixed-point $h$ of $\mathcal{T}^{\mathrm{l}}_G$, 
$\mathrm{prob}^{\mathrm{l},\max}_{n,G}\le h$ for all $n\ge 0$ by the facts that $\mathrm{prob}^{\mathrm{l},\max}_{0,G}\le h$ and 
$\mathrm{prob}^{\mathrm{l},\max}_{n+1,G}=\mathcal{T}^\mathrm{l}_G(\mathrm{prob}^{\mathrm{l},\max}_{n,G})$. It follows that 
$\mathrm{prob}^{\mathrm{l},\max}_{G}\le h$ for any fixed point $h$ of $\mathcal{T}^{\mathrm{l}}_G$\enskip.

Finally, by taking the limit from (b) we can obtain that
\[
\left|\mathrm{prob}^{\mathrm{l},\max}_G(s,\mathbf{c})-\mathrm{prob}^{\mathrm{l},\max}_G(s,\mathbf{c}')\right|\le \frac{\mathbf{E}_{\max}}{\mathbf{w}_{\min}}\cdot{\parallel}{\mathbf{c}-\mathbf{c}'}{\parallel}_\infty
\]
for all $\mathbf{c},\mathbf{c}'\ge \vec{0}$ and $s\in L$\enskip.
\end{proof}

\end{document}